\documentclass[11pt]{article}

\usepackage[utf8]{inputenc}
\usepackage[T1]{fontenc}
\usepackage[american]{babel}

\usepackage[a4paper,margin=1.0in]{geometry}
\usepackage{amsmath,amssymb,amsthm,mathtools}
\usepackage{amsfonts}
\usepackage{thmtools}
\usepackage[authoryear,round,comma]{natbib}
\usepackage{enumitem} 

\usepackage{etoolbox}
\newtoggle{journal}

\usepackage{microtype}

\usepackage[dvipsnames]{xcolor}
\usepackage{hyperref}
\hypersetup{
  colorlinks=true,
  linkcolor=blue!50!black,
  citecolor=blue!50!black,
  urlcolor=blue!50!black
}

\usepackage[capitalize,noabbrev]{cleveref}
\usepackage{needspace}

\theoremstyle{plain}
\newtheorem{theorem}{Theorem}[section]
\newtheorem{proposition}[theorem]{Proposition}
\newtheorem{lemma}[theorem]{Lemma}

\theoremstyle{definition}
\newtheorem{definition}[theorem]{Definition}

\newtheorem{remark}[theorem]{Remark}

\newcommand{\R}{\mathbb{R}}
\newcommand{\F}{\mathcal{F}}
\newcommand{\U}{\mathcal{U}}
\newcommand{\Lc}{\mathcal{L}}
\newcommand{\simplex}{\Delta^\circ}
\newcommand{\clr}{\mathrm{clr}}
\newcommand{\Aitch}{d_A}
\newcommand{\fref}{f_{\mathrm{ref}}}
\newcommand{\Prod}{\Pi}
\newcommand{\Sum}{\Sigma}
\newcommand{\dSum}{d_\Sum}
\DeclareMathOperator{\median}{median}

\title{Fairness and Strategy-Proofness in Automated Market Makers}
\iftoggle{journal}{%
\author{Frank M.\ V.\ Feys \\ \small Independent researcher, Antwerp, Belgium \\ \small \texttt{frank@fmvfeys.com}}
}{\author{Frank M.\ V.\ Feys}}

\date{June 3, 2026}

\begin{document}

\emergencystretch=3em

\maketitle


\begin{abstract}
No deployed automated market maker lets its liquidity providers vote on the trading function.
We show this is structural, not an oversight.
On the weighted-product family with $n \geq 3$ assets, no aggregation rule is at once fair and strategy-proof.
Arrovian fairness forces a unique form, the \emph{weighted Aitchison centroid}, the weighted geometric mean of the providers' preferred pools.
But fairness forces \emph{mean-type} aggregation and strategy-proofness forces \emph{median-type}, and the only rule that is both is a single-provider dictator.
The obstruction is sharp: it vanishes at $n = 2$, where a fair strategy-proof rule exists.
Under the Frongillo--Papireddygari--Waggoner equivalence, the centroid is Genest's logarithmic opinion pool, and the impossibility transfers to externally Bayesian pooling.
\end{abstract}

\medskip


\medskip



\section{Introduction}\label{sec:intro}

\subsection{The Question}\label{sec:intro-question}

\iftoggle{journal}
{
Automated market makers have become the dominant trading venue on public blockchains, pricing tens of billions of dollars a day (2025--2026) by formula rather than by a market maker.
At the heart of the dominant class, constant-function market makers (CFMMs), sits one mathematical object: a strictly increasing quasiconcave \emph{trading function} $f \colon \R^A_{>0} \to \R$ on inventories in $n = |A|$ assets.
Trades preserve its level sets, so this single choice fixes the pool's entire price-formation rule.
The constant-product $f(I) = \prod_i I_i$ and the constant-sum $f(I) = \sum_i I_i$ mark the extremes: the former quotes a price that moves steeply as reserves deplete, the latter one that holds fixed until an asset runs out.
The choice is a fundamental design decision, not a technical detail.
}
{ 
Automated market makers have become the dominant trading venue on public blockchains, intermediating tens of billions of dollars a day (2025--2026), pricing every trade by formula rather than by a market maker.
At the heart of the dominant class, constant-function market makers (CFMMs), sits one mathematical object: a strictly increasing quasiconcave \emph{trading function} $f \colon \R^A_{>0} \to \R$ on inventories in $n = |A|$ assets.
Trades preserve its level sets, so this single choice fixes the entire price-formation rule of the pool, quoting every price deterministically from the current reserve vector.
The constant-product function $f(I) = \prod_i I_i$ and the constant-sum function $f(I) = \sum_i I_i$ mark the range: the former quotes a price that moves steeply as reserves deplete, the latter one that holds fixed until an asset runs out.
Different trading functions encode different bets about how the pool should price risk, so the choice is a fundamental design decision rather than a technical detail.
}

Yet every deployed protocol fixes its trading function at launch, and LPs simply deposit into pre-built pools they did not design, with no mechanism to express preferences over the risk profile that function implies.
From a market-design standpoint, this is a striking omission.
LPs are the residual claimants of every pool: they bear price risk through their pro-rata share of the reserves, and the choice of trading function shapes that risk directly.
A constant-product LP and a constant-sum LP face very different exposures to price moves: stablecoin liquidity providers, for instance, reasonably prefer a Curve-style invariant that concentrates liquidity near the peg to a Uniswap-style one that spreads it uniformly across the entire price range.

It would seem natural, then, to let the people bearing the risk have a say in the formula that governs it.
So why does no major protocol let LPs collectively choose the trading function?
The obvious answer is  that this is an engineering oversight awaiting fix.
We shall show that it is in fact a structural constraint that no protocol can engineer around without giving up something essential. 
A scope caveat is in order at the outset: 
 the main \emph{impossibility} of this paper requires $n \geq 3$ assets, 
 which covers the multi-asset Balancer and Curve pools and the multi-asset Uniswap v2 generalization but \emph{excludes} the canonical two-asset Uniswap v2/v3 pool, which is the most deployed AMM by volume.\footnote{By aggregate daily spot volume, two-asset pools substantially exceed multi-asset weighted-product pools: as of 1 June 2026, the Uniswap v2/v3/v4 protocols (all two-asset CPMMs) transacted approximately \$3 billion in 24-hour spot volume against approximately \$100 million for Balancer v2/v3 (the leading multi-asset weighted-product venue), a ratio of order $30$ to $1$ \citep{defillama-2026}. Curve's comparable daily volume is largely StableSwap, which is structurally outside the present framework.}
At $n = 2$ the design space is one-dimensional and the cocycle/Cauchy machinery breaks down for structural reasons.
The impossibility is concentrated where $n \geq 3$ makes pair-restrictions lossy, and the two-asset case is treated separately in \Cref{sec:two-asset}, where the Arrovian core together with  strategy-proofness is shown to be compatible.

By a ``democratic'' rule we mean, informally,  one that aggregates LP preferences anonymously and nontrivially.
The formal hypotheses are the full Arrovian core of \Cref{sec:axioms} (of which anonymity and nontriviality are two axioms), and the title's impossibility is precisely the incompatibility of that core with manipulation-resistance.
Nontriviality rules out the constant rule that ignores LP input,  and the deployed practice of fixing the trading function in advance is exactly that constant rule.

\subsection{Contributions}\label{sec:intro-results}

We treat AMM design as a problem in mechanism design.
LPs are the principals whose preferences over trading functions are aggregated, and the trading function is the social choice.
Concretely, an \emph{aggregation  rule} (or \emph{mechanism}) is a function that takes as input each LP's most-preferred trading function and each LP's deposit, and  produces as output a single trading function for the pool.
Our axioms constrain this aggregation rule, not the trading function it produces.
This is the paper's central methodological move.
Prior axiomatic work on AMMs imposes conditions on $f$ as a self-standing object, while we impose conditions on the procedure that selects $f$ from LP inputs.
Strategy-proofness in particular is a rule-level property.
 Within this framing we obtain two results, a characterization and an impossibility.

 The first concerns the weighted-product  family $\F_\Prod = \{f_\alpha \mid \alpha \in \simplex\}$, 
 where $f_\alpha(I) = \prod_i I_i^{\alpha_i}$ and $\simplex$ denotes the open probability simplex.
Six axioms characterize the admissible aggregation rules on this family:  four substantive fairness conditions (Pareto efficiency over LP preferences, anonymity of LPs, continuity in profiles, and \emph{mechanism IIA}, 
 a mechanism-level analog of Arrow's independence of irrelevant alternatives axiom from \citet{arrow-1951}) and two regularity conditions (nontriviality and unrestricted domain).
Mechanism IIA requires that the structural exchange-rate ratio between assets $i$ and $j$ depend only on LPs' $(i,j)$-restricted preferences and the valuation profile $V$, not on their preferences about other asset pairs.
Continuity is substantive rather than merely technical, since it rules out discontinuous tie-breaking schemes that would otherwise reintroduce admissible but unstable rules (\Cref{app:independence}).
We refer to these six axioms collectively  as the \emph{Arrovian core}.
They appear in \Cref{sec:axioms} under the labels (A0)--(A4) together with continuity (C).
The aggregators on $\F_\Prod$ with $n \geq 3$ assets selected by the Arrovian core are precisely the \emph{weighted Aitchison centroids}, namely the weighted geometric means on the simplex of the LP-preferred weight vectors $\alpha_\ell^* \in \simplex$, with weights $w_\ell(V)$ that are continuous equivariant functions of the LPs' deposit valuations $V_\ell$.
This is the content of \Cref{thm:characterization-F0}.

The same mathematical object reappears in Bayesian statistics, identified in \Cref{prop:correspondence}  and put to work in \Cref{sec:pooling-transfer}.
Under the Frongillo--Papireddygari--Waggoner equivalence between CFMMs and cost-function prediction markets \citep{frongillo-papireddygari-waggoner-2024}, an LP-preference profile on $\F_\Prod$ corresponds to a profile of probability distributions, and the weighted Aitchison centroid corresponds to the \emph{logarithmic opinion pool} of \citet{genest-1984}, the unique externally Bayesian pooling rule under his  regularity hypotheses.
The substantive axioms correspond on the two sides and the characterized rule families coincide.
Running the impossibility across the equivalence yields a transferred manipulability theorem for externally Bayesian opinion pools under the Aitchison metric (\Cref{thm:pooling-impossibility}), which complements a recent line of impossibilities for probability aggregation pursued under different nonmanipulability notions and different pool families (\citealp{dietrich-list-2025, chambers-echenique-hayashi-2024, laraki-varloot-2022}; see \Cref{sec:pooling-theorem}).
Concretely, no externally Bayesian, anonymous, continuous, nondegenerate pool on $n \geq 3$ outcomes is Aitchison-strategy-proof, where Aitchison-strategy-proofness is the manipulation-resistance condition under which the log pool is the metric (Fr\'echet) mean of the experts' distributions.

The second theorem establishes the impossibility the title names.
The natural strategy-proofness axiom on $\F_\Prod$ requires that no LP can shift the mechanism's output closer to their true preference by misreporting it, where ``closer'' is measured in the Aitchison metric.
We show in \Cref{thm:impossibility} that on $\F_\Prod$ with $n \geq 3$, adding this requirement to the Arrovian core yields an impossibility, in that no aggregator mechanism satisfies both.

The impossibility is a collision between two geometries on the same space.
The Arrovian core, via mechanism IIA and the cocycle of structural price ratios, forces \emph{mean-type} aggregation by magnitude
 (a weighted geometric mean of LP-preferred  weight vectors, equivalently, a weighted arithmetic mean in centered log-ratio coordinates), whereas strategy-proofness in the Euclidean geometry that the Aitchison metric induces on those coordinates forces \emph{median-type} aggregation by rank, since the linear aggregator already pinned down by the Arrovian core is Euclidean strategy-proof if and only if it is a single-LP dictator (\Cref{lem:sp-forces-projection}).
Anonymity then pins the weights to the uniform vector at any symmetric valuation profile, and the uniform vector is not a single-LP dictator, so for $m \geq 2$ no rule can be both fair and strategy-proof.
The two-asset case is treated separately in \Cref{sec:two-asset},  where (A3) goes vacuous and a  Moulin generalized-median characterization replaces the impossibility (under the additional modeling assumption that deposit-dependence factors through the valuation profile, which on $\F_\Prod$ at $n \geq 3$ is a consequence of the axioms but must be imposed separately at $n = 2$).

The proofs of both theorems rest on a Cauchy-type functional equation on the cocycle of structural price ratios on the simplex, the multiplicative analog of the additive cocycle identified by \citet{aczel-wagner-1980} in their characterization of weighted arithmetic mean aggregators.

\subsection{The Design Fork and Scope}\label{sec:intro-fork}

The impossibility forces a sharp fork.
Any protocol aggregating LP preferences on  $\F_\Prod$ must give up one of four properties, namely mechanism IIA, anonymity, nontriviality, or strategy-proofness;
every deployed protocol takes the third route, fixing the trading function in advance (forfeiting nontriviality (A0) outright, and Pareto efficiency (A2) at any profile whose LP peaks do not coincide with the fixed function), and \Cref{thm:impossibility} explains why this is forced rather than chosen.
The four routes and their real-world instances are developed in \Cref{sec:practical-implications}.
The two theorems are stated on $\F_\Prod$ with $n \geq 3$ assets and Aitchison-single-peaked preferences, and three \emph{design subspaces} delimit their reach.
Restricting to $\F_\Prod$ rather than the full space $\F$ is a modeling choice: $\F_\Prod$ is where the cocycle/Cauchy machinery applies and where the leading deployed protocols (Uniswap, Balancer) sit, while $\F$ remains open (\Cref{sec:open-questions}, Q1).
On $\F_\Sum$, mechanism IIA is vacuous and the impossibility does not arise (\Cref{lem:A3-vacuous-Fgamma,rem:Fgamma-SP-included}); the same holds on $\F_\Prod$ at $n = 2$ (\Cref{thm:characterization-F0-n-two}), and on $\U = \F_\Prod \cup \F_\Sum$ the impossibility extends to the $\F_\Prod$-arm but not the $\F_\Sum$-arm (\Cref{rem:union}).
Hanson's LMSR \citep{hanson-2003} and Curve's StableSwap both sit outside the framework, the former because it is translation- rather than scale-invariant (\Cref{rem:LMSR-outside}), the latter because its pricing function is nonseparable across asset pairs (\Cref{rem:curve-nonseparable}).

\subsection{Related Literature}\label{sec:intro-related}

The closest precedent to our work is \citet{schlegel-kwasnicki-mamageishvili-2023}, who initiated the modern axiomatic treatment of CFMMs at the trading-function level, characterizing the constant-elasticity family by independence and scale invariance and the LMSR by independence and translation invariance.
Our paper addresses a different question: not which trading functions are admissible, but which aggregation rules are admissible when the trading function is the output of LP preference aggregation.
Their independence axiom and our mechanism IIA agree on constant mechanisms but diverge on nonconstant ones.
Theirs restricts $f$ itself, asking its marginal $(i,j)$-rate  to ignore third-asset inventories, whereas ours restricts the \emph{rule} producing $f$.
The impossibility (\Cref{thm:impossibility}) is structurally absent from the trading-function-level setting.
\citet{frongillo-papireddygari-waggoner-2024} establish the CFMM--prediction-market equivalence on which our correspondence with opinion pooling rests.
\citet{bichuch-feinstein-2025} axiomatize trading functions using a utility-theoretic framework with an explicit price-impact measure.
Their axioms operate at the trading-function level and do not include an analog of mechanism IIA, so their characterization and ours are complementary, with no direct overlap.

Real-world instances of the aggregation problem appear  in AMM governance: Balancer's governance system \citep{balancer-governance-2026} allows weight adjustments to its weighted-CPMM pools, and Uniswap's fee-tier governance \citep{uniswap-governance-2021} determines the available fee schedules; in the taxonomy of \Cref{sec:practical-implications}, these token-weighted votes take the second route (nonanonymous influence scaling with governance-token holdings).
Neither system asks LPs to report a preferred trading function directly.
Instead, governance-token holders vote on discrete parameter changes, so the aggregation operates over a finite menu of options rather than the full continuous design space $\F_\Prod$.

On the social-choice and functional-equations side, our argument builds on the toolkit of \citet{aczel-1966},  the opinion-pool characterizations of \citet{aczel-wagner-1980},   \citet{genest-1984}, and \citet{mcconway-1981}, the Moulin theorem on strategy-proofness for single-peaked preferences \citep{moulin-1980,moulin-axioms-of-cooperative-decision-making}, the multidimensional strategy-proofness results of \citet{border-jordan-1983},  and the classical impossibilities of \citet{arrow-1951}, \citet{gibbard-1973}, and \citet{satterthwaite-1975}.
The continuum-of-alternatives formulation of Moulin's characterization that we use in \Cref{thm:characterization-Fgamma} is due to \citet{ching-1997}.
The Acz\'el--Wagner additive cocycle on the simplex and its multiplicative counterpart on the open simplex are the two functional-equation skeletons that organize our proofs, the latter giving the weighted geometric mean of \Cref{thm:characterization-F0} and the former the weighted arithmetic mean that would govern the LMSR analog of Q2.
A more general framework for strategy-proof social choice \citep{nehring-puppe-2007}  subsumes our $\F_\Sum$-side results as an instance on a one-dimensional median space, but the $\F_\Prod$-side impossibility lies outside it: 
the centered log-ratio transform is an isometry from the Aitchison simplex onto the Euclidean hyperplane \mbox{$H \cong \R^{n-1}$},
 and $\R^d$ is not a median space for $d \geq 2$, so for $n \geq 3$ the simplex falls outside the median-space framework, and the median-space property that enables generalized-median rules on $\F_\Sum$ is precisely what is absent on $\F_\Prod$.
The pooling-side manipulability theorem (\Cref{sec:pooling-theorem}) joins a recent line of impossibilities for probability aggregation \citep{dietrich-list-2025, chambers-echenique-hayashi-2024, laraki-varloot-2022}, from which it is distinguished by its use of the Aitchison metric on the simplex and by taking external Bayesianity as primitive.

\subsection{Executive Summary}\label{sec:intro-exsum}

Our axioms constrain the aggregation rule that produces the trading function, not the trading function itself.
This rule-level vantage makes strategy-proofness a meaningful axiom and yields the picture below.

\begin{itemize}

\item \emph{What is impossible.}
On the weighted-product family $\F_\Prod$ with $n \geq 3$ assets, no rule aggregating LP preferences can be both fair (Pareto, anonymity, continuity, mechanism IIA, nontriviality, unrestricted domain) and strategy-proof (\Cref{thm:impossibility}).

\item  \emph{What fairness alone forces.}
The Arrovian axioms pin down the rule to a unique functional form:
 a weighted Aitchison centroid of the LP-preferred weight vectors, with weights $w_\ell(V)$ that are continuous equivariant functions of the valuation profile (\Cref{thm:characterization-F0}).

\item  \emph{Why fairness and strategy-proofness are incompatible.}
Fairness forces \emph{mean-type} aggregation. 
Within that family,  strategy-proofness selects only single-LP dictators. 
Anonymity then rules out dictators for $m \geq 2$.

\item \emph{The four escape routes.}
Mechanism IIA, anonymity, nontriviality, or strategy-proofness: one must go. 
Every deployed protocol drops nontriviality, at the cost of Pareto wherever the fixed function falls outside the peaks' Aitchison hull.

\item \emph{Where the impossibility does not arise.}
On $\F_\Sum$ and on $\F_\Prod$ at $n = 2$, mechanism IIA is vacuous and the core is consistent with strategy-proofness; under the factoring hypothesis (V), the admissible rules are the Moulin generalized medians (\Cref{rem:Fgamma-SP-included,thm:characterization-F0-n-two}).

\item \emph{A consequence in statistics.}
The characterized rule is Genest's logarithmic opinion pool, and the impossibility transfers to a manipulability theorem for externally Bayesian pools under the Aitchison metric (\Cref{thm:pooling-impossibility}), complementing recent impossibilities for probability aggregation \citep{dietrich-list-2025, chambers-echenique-hayashi-2024, laraki-varloot-2022}.

\item \emph{AMM families outside the framework.}
Hanson's LMSR sits outside because it is translation- rather than scale-invariant; Curve's StableSwap sits outside because its pricing function is nonseparable across asset pairs (\Cref{rem:LMSR-outside,rem:curve-nonseparable}).

\end{itemize}

The impossibility is thus structurally confined to multi-asset weighted-product design spaces, where mechanism IIA does substantive work. 
The two-asset case and $\F_\Sum$ both admit Arrovian-fair strategy-proof rules.

\subsection{Outline}\label{sec:intro-outline}

\Cref{sec:setup,sec:axioms} fix the formalism and state the axioms.
\Cref{sec:F0-characterization} proves the characterization on $\F_\Prod$ with $n \geq 3$, identifying the weighted Aitchison centroids; \Cref{sec:impossibility} derives the impossibility.
\Cref{sec:two-asset} treats $n = 2$, where Moulin generalized medians replace the impossibility.
\Cref{sec:scope} delimits the reach of both theorems on $\F_\Sum$ and on $\U$, characterizing the generalized medians on $\F_\Sum$.
\Cref{sec:practical-implications} reads the picture for an AMM designer, with the four routes and a calibration cross-check.
\Cref{sec:pooling-transfer} transfers the impossibility to Bayesian opinion pooling, where the centroid reappears as Genest's logarithmic pool.
\Cref{sec:discussion} records scope and open questions.
\Cref{app:independence,app:identity,app:general-pareto} verify logical independence, record an elementary identity, and extend the impossibility beyond Aitchison-single-peakedness, respectively.


\section{Setup}\label{sec:setup}

Throughout, we write $n := |A|$ for the number of assets with  $A = \{1, 2, \ldots, n\}$ fixed.
We assume $n \geq 2$ and shall indicate where $n \geq 3$ is needed.
The number $m$ of LPs is fixed with $m \geq 2$.

\subsection{Trading Functions}\label{sec:trading-functions}

A \emph{trading function} is a function $f \colon \R^A_{>0} \to \R$ satisfying (T1)--(T4) below.
We write $I \in \R^A_{>0}$ for a pool's  \emph{inventory} or \emph{reserve} vector, with $I_i$ the number of units of asset $i$ held by the pool.
This is the typical input to $f$, and $\R^A_{>0}$ is the inventory space.
Operationally, a pool holding reserves $I$ accepts a trade to reserves $I' \in \R^A_{>0}$ exactly when $f(I') = f(I)$, so the level sets of $f$ are the admissible trades and the choice of $f$ fixes how the pool quotes every price.

\begin{description}
\item[(T1)] $f \in C^2(\R^A_{>0})$ (twice continuously differentiable).
 \item[(T2)] $\partial_i f > 0$ on $\R^A_{>0}$ for every $i \in A$ (strict monotonicity).
\item[(T3)] For every $I \in \R^A_{>0}$  and every $v \in \ker Df(I)$, it holds that $v^\top D^2f(I)\, v \leq 0$ (the Hessian of $f$, restricted to the tangent space of the level set, is negative semidefinite).
\item[(T4)] $f$ is $1$-homogeneous: $f(\lambda I) = \lambda f(I)$ for every $\lambda > 0$ and every $I \in \R^A_{>0}$.
\end{description}

Conditions (T1)--(T4) are standard in the CFMM literature \citep{schlegel-kwasnicki-mamageishvili-2023,frongillo-papireddygari-waggoner-2024}, with one departure: those sources include a ``sufficient funds'' or ``aversion to permanent loss'' axiom controlling $f$ near $\partial \R^A_{\geq 0}$, which we drop, since $\F_\Prod$ and $\F_\Sum$ exhibit opposite boundary behaviors (level sets of weighted-product functions asymptote to the axes without touching them; the constant-sum function $f_1$ has level sets that meet the axes) and no proof in what follows uses the condition.
Dropping it lets $\F_\Prod$ and $\F_\Sum$ sit uniformly inside the framework, the constant-sum endpoint of $\F_\Sum$ included.
We use the semidefinite form in (T3) rather than strict concavity because the constant-sum function $f_1 = \tfrac{1}{n}\sum_i I_i$ (the $\gamma = 1$ endpoint of $\F_\Sum$, see \Cref{def:Fgamma} below) is linear, with flat level sets; every other function in $\F_\Prod$ and in $\F_\Sum$ with $\gamma < 1$ is  strictly quasiconcave.
The $1$-homogeneity (T4) is the DeFi convention making LP tokens fungible, in that scaling reserves by $\lambda$ scales pool value by $\lambda$.
It is substantive only on the general space $\F$, holding automatically on $\F_\Prod$ and $\F_\Sum$.

Two trading functions $f, g$ are \emph{equivalent} if there exists a $C^2$ diffeomorphism $\Phi$ on the range of $f$ (equivalently, a $C^2$ function with $\Phi' > 0$ throughout the range of $f$) such that $g = \Phi \circ f$.
This is a genuine equivalence relation: reflexivity by $\Phi = \mathrm{id}$, symmetry because the inverse of a $C^2$ diffeomorphism is again a $C^2$ diffeomorphism (the inverse function theorem applies since $\Phi' > 0$), and transitivity by composition.
The strict positivity of $\Phi'$ also ensures $g = \Phi \circ f$ inherits (T1)--(T4) from $f$ (in particular (T2), as $\partial_i g = \Phi'(f)\,\partial_i f > 0$).
Equivalence preserves level sets and hence preserves marginal exchange rates, no-arbitrage conditions, and all economically meaningful properties.
We denote the resulting set of equivalence classes by $\F$ and refer to its elements as trading functions, identifying each class with a representative when convenient.
The key structural quantities are well-defined on equivalence classes, not merely on representatives.

\begin{lemma}\label{lem:well-defined-quotient}
Let $f$ satisfy (T1)--(T4) and let $g = \Phi \circ f$ with $\Phi$ a $C^2$ diffeomorphism on the range of $f$.
Then: (a) the marginal exchange rate  satisfies $p_{ij}(g, I) = p_{ij}(f, I)$ for every $I$; (b) condition (T3) holds for $g$ if and only if it holds for $f$.
\end{lemma}

\begin{proof}
Both are direct computations.
(a) We have $\partial_i g = \Phi'(f)\,\partial_i f$ with $\Phi' > 0$, so the factor $\Phi'(f)$ cancels in the ratio $p_{ij}(g,I) = \partial_i g/\partial_j g = \partial_i f/\partial_j f = p_{ij}(f,I)$.
(b) By the chain rule, it holds that $D^2 g = \Phi''(f)\,Df^\top Df + \Phi'(f)\,D^2 f$, and $Dg = \Phi'(f) \, Df$ with $\Phi'(f) > 0$, so $\ker Dg(I) = \ker Df(I)$ (the tangent space to the level set is the same).
Observe that on $\ker Df(I) =  \ker Dg(I)$ the first term vanishes, so $v^\top D^2 g\, v = \Phi'(f)\,v^\top D^2 f\, v$ has the same sign as $v^\top D^2 f\, v$ since $\Phi' > 0$; 
therefore, (T3) holds for $g$ on its kernel if and only if it holds for $f$ on its (same) kernel, and (T3) is reparametrization-invariant.
\end{proof}

\begin{remark}[LMSR Sits Outside]\label{rem:LMSR-outside}
Hanson's LMSR $f(I) = b \log \sum_i e^{I_i/b}$ is translation-invariant (\mbox{$f(I + \lambda \mathbf{1}) = f(I) + \lambda$}) rather than scale-invariant, so it fails (T4).
The translation-invariant axiom system one would obtain by replacing (T4) with translation invariance is the natural setting for LMSR.
In that system the multiplicative cocycle becomes additive, and the weighted geometric mean of \Cref{thm:characterization-F0} would be replaced by a weighted arithmetic mean characterized by \citet{aczel-wagner-1980}.
A precise characterization-plus-impossibility  statement in that system is open and is formulated as Q2 of \Cref{sec:open-questions}.
We do not pursue it here.
\end{remark}

We equip $\F$ with the $C^2$ topology of uniform convergence on compact subsets for conceptual completeness, but the theorems below use this topology only on the finite-dimensional subspaces \mbox{$\F_\Prod \cong  \simplex$} and \mbox{$\F_\Sum \cong (-\infty,1]$} (\Cref{def:F0,def:Fgamma}), where it reduces to the standard Euclidean topology.

\subsection{Design Subspaces}\label{sec:design-subspaces}

The mechanism and LP preferences are both stated relative to a subset of $\F$ carrying a metric.
We name such a pair as follows.

\begin{definition}  \label{def:design-subspace}
A \emph{design subspace} is a pair  $(\F', d)$, where $\F' \subseteq \F$ is a subset of trading-function equivalence classes and $d$ is a metric on $\F'$.
\end{definition}

The mechanism (\Cref{def:mechanism}) is defined as a map on $\F'$, and LP preferences  (\Cref{def:profile}) are single-peaked in $d$.
We work with three concrete design subspaces, each defined together with its metric.
The first two we label by their defining algebraic operator:  $\F_\Prod$ for the product-based family and $\F_\Sum$ for the sum-based family.
The third is their union.

\paragraph{The weighted-product family.}
We begin with the weighted constant-product market makers.

\begin{definition}\label{def:F0}
The \emph{weighted-CPMM family} is
\[
\F_\Prod := \{ f_\alpha \mid \alpha \in \simplex \},  \quad   f_\alpha(I) := \prod_{i=1}^n I_i^{\alpha_i},
\]
where $\simplex := \{\alpha \in \R^n_{>0} \mid \sum_i \alpha_i = 1\}$  is the open probability simplex on the $n$ assets.
\end{definition}

We fix in each class the normalized representative $f_\alpha(I) = \prod_{i=1}^n I_i^{\alpha_i}$ with $\alpha \in \simplex$, i.e., with $\sum_i \alpha_i = 1$.
Each class contains exactly one such  $\alpha$, since rescaling the exponent vector by $c > 0$ gives an equivalent function.
The map $\alpha \mapsto f_\alpha$ is a homeomorphism onto  $\F_\Prod$ in the $C^2$ topology (the structural ratios recover $\alpha$, and $f_\alpha \mapsto \partial_i f_\alpha(\mathbf{1})$ inverts it continuously), so we identify $\F_\Prod$ with $\simplex$.
Each $f_\alpha \in \F_\Prod$ satisfies (T1)--(T4):  $f_\alpha$ is $C^\infty$ on $\R^A_{>0}$ (T1); its partials $\partial_i f_\alpha = \alpha_i f_\alpha/I_i > 0$ are strictly positive (T2); $\log f_\alpha = \sum_i \alpha_i \log I_i$ is concave in $\log I$ and hence $f_\alpha$ is log-concave, giving (T3); and $f_\alpha(\lambda I) =  \prod_i (\lambda I_i)^{\alpha_i} = \lambda^{\sum_i \alpha_i} f_\alpha(I) = \lambda f_\alpha(I)$ confirms $1$-homogeneity (T4).
Hence, Uniswap v2's constant product ($\alpha = (1/n,  \ldots, 1/n)$)  and Balancer's weighted products ($\alpha$ general) are members of $\F_\Prod$ and conform to our axiomatic setup.

We equip $\F_\Prod$ with the Aitchison metric \citep{aitchison-1986}, the natural Riemannian metric on the simplex from the standpoint of compositional data analysis and information geometry.
Let $H := \{ x \in \R^n \mid \sum_i  x_i = 0 \}$ be the zero-sum hyperplane.
The \emph{centered log-ratio} (clr) map is
\[
\clr \colon \simplex \to H, \quad \clr(\alpha)_i := \log \alpha_i - \frac{1}{n} \sum_{j=1}^n \log \alpha_j,
\]
which is a homeomorphism with inverse equal to the softmax map  $\clr^{-1} \colon H \to \simplex$, given by $\clr^{-1}(x)_i := e^{x_i}/\sum_{k=1}^n e^{x_k}$  (the standard softmax on $\R^n$, restricted to $H$).

\begin{definition}[Aitchison Metric]\label{def:aitchison-metric}
The \emph{Aitchison metric} on \mbox{$\F_\Prod \cong \simplex$} is the pullback of the Euclidean metric on $H$,
\[
\Aitch(\alpha, \beta) := \| \clr(\alpha) - \clr(\beta) \|_2,
\]
where $\|\cdot\|_2$ denotes the standard Euclidean norm on $\R^n$.
Since $\clr(\alpha) - \clr(\beta) \in H$ (both vectors are zero-sum, so their difference is as well), this equals the norm induced on the subspace $H \subset \R^n$ by the ambient Euclidean structure.
The two agree and we use $\|\cdot\|_2$ throughout.
\end{definition}

\begin{remark}[Codomain of the LP Weighting]\label{rem:weighting-codomain}
Weighting functions over LPs (introduced in \Cref{thm:characterization-F0}) are written  with codomain $\Delta^{m-1}$, the \emph{closed} $(m-1)$-simplex of weight vectors over the $m$ LPs.
This is distinct from the asset simplex $\simplex$ and lives in a different ambient space.
We do not assume strict positivity of LP weights:  
 a continuous equivariant weighting may vanish on part of its domain without making the mechanism constant, so the relative  interior is not the right codomain.
What equivariance does force, and all we shall use (in the impossibility of \Cref{thm:impossibility}), is that at a symmetric valuation profile $V = (v,  \ldots, v)$ the weight is the uniform vector $(1/m,  \ldots, 1/m)$.
\end{remark}

\paragraph{The constant-elasticity family.}
The second design subspace is the symmetric constant-elasticity family.

\begin{definition}\label{def:Fgamma}
The \emph{symmetric CEMM family} is
\[
\F_\Sum := \{ f_\gamma  \mid \gamma \in (-\infty, 1] \}, \quad f_\gamma(I) :=
\begin{cases} 
 \left( \tfrac{1}{n}\sum_{i=1}^n I_i^\gamma \right)^{1/\gamma} & \gamma \neq 0, \\[2pt]
  \left( \prod_{i=1}^n I_i \right)^{1/n} & \gamma = 0.
\end{cases}
\]
\end{definition}

For each $\gamma \in (-\infty, 1]$,  the formula in \Cref{def:Fgamma} defines a concrete function $f_\gamma \colon \R^A_{>0} \to \R$,
  the element of $\F$ associated to $\gamma$ is its equivalence class $[f_\gamma]$.
Any other representative of $[f_\gamma]$ differs from $f_\gamma$ by a $C^2$ diffeomorphism on the range.
In particular, $g_\gamma(I) := (\sum_{i=1}^n I_i^\gamma)^{1/\gamma}$ satisfies $g_\gamma = n^{1/\gamma} f_\gamma$, so $g_\gamma$ represents the same class $[f_\gamma]$.
The choice between $f_\gamma$ and $g_\gamma$ as the concrete representative matters for one purpose only: we want the map $\gamma \mapsto f_\gamma$ from $(-\infty, 1]$ into the space of $C^2$ functions on $\R^A_{>0}$ to be continuous, including at $\gamma = 0$.
With the prefactor $\tfrac{1}{n}$,  this indeed holds: 
 $f_\gamma(I) \to (\prod_i I_i)^{1/n}$ as  $\gamma \to 0$, uniformly on compact subsets of $\R^A_{>0}$
  (this can be seen  by applying L'H\^opital to $\log f_\gamma(I) = \tfrac{1}{\gamma} \log\bigl(\tfrac{1}{n}\sum_i I_i^\gamma\bigr)$; numerator and denominator both vanish at $\gamma = 0$) \citep[Ch.~2]{hardy-littlewood-polya-1934}.
Without it, $\log g_\gamma(I) = \tfrac{\log n}{\gamma} + \tfrac{1}{n}\sum_i \log I_i + o(1)$ shows $g_\gamma$ has no limit at $\gamma = 0$.
The bound $\gamma \leq 1$ is forced by (T3), as $f_\gamma$ is strictly convex rather than quasiconcave for $\gamma > 1$ (economically, increasing returns to liquidity, inadmissible).
The map $\gamma \mapsto [f_\gamma]$ is a homeomorphism from $(-\infty, 1]$ onto $\F_\Sum$, and following \citet{schlegel-kwasnicki-mamageishvili-2023} we call $1/(1-\gamma)$ the \emph{elasticity}.

\begin{definition}[Metric  on $\F_\Sum$]\label{def:dSum-metric}
The metric $\dSum$ on $\F_\Sum$, identified  with $(-\infty, 1]$ via $\gamma \mapsto [f_\gamma]$, is
\[
\dSum(\gamma_1, \gamma_2) := |\arctan \gamma_1 - \arctan \gamma_2|.
\]
\end{definition}

The arctan compactifies the half-line $(-\infty, 1]$ homeomorphically onto a bounded interval, yielding a complete metric.
Any equivalent metric giving the same topology yields the same characterization in \Cref{thm:characterization-Fgamma}, since Moulin's theorem depends only on the ordering of peaks, not on the specific metric.
The strategy-proofness statement (SP) is, by contrast, metric-specific, stated relative to $\dSum$.

\paragraph{The union.}
The two families intersect at $\gamma = 0$, where $f_0(I) = (\prod_i I_i)^{1/n}$ is the geometric mean of the reserves and coincides with $f_\alpha$ for the uniform weight vector $\alpha = (1/n,  \ldots, 1/n) \in \simplex$.
We denote this common element by  $c$ and call it the \emph{symmetric CPMM}.

\begin{lemma}\label{lem:intersection}
The two  families meet in a single point, $\F_\Prod \cap \F_\Sum = \{c\}$.
\end{lemma}

\begin{proof}
At $\gamma = 0$, $f_0(I) = (\prod_i I_i)^{1/n}$ is $f_\alpha$ with $\alpha = (1/n, \ldots, 1/n)$, so $c \in \F_\Prod \cap \F_\Sum$.
Conversely, let $f_\alpha \sim f_\gamma$ for some $\alpha \in \simplex$ and $\gamma \in (-\infty, 1]$.
Equivalence forces equal marginal prices.
On the $f_\alpha$-side, $p_{ij}(f_\alpha, I) = (\alpha_i/\alpha_j)(I_j/I_i)$ has exponent $1$ in $I_j/I_i$.
On the $f_\gamma$-side, \mbox{$p_{ij}(f_\gamma, I) = (I_j/I_i)^{1-\gamma}$} has exponent $1-\gamma$ and is independent of the pair $(i,j)$.
Matching exponents gives $\gamma = 0$; pair-independence forces  $\alpha_i/\alpha_j$ constant for all pairs, so \mbox{$\alpha = ( 1/n, \ldots, 1/n)$}.
Both equal $c$.
\end{proof}

The third design subspace is the union
\[
 \U  :=  \F_\Prod \cup \F_\Sum,
\]
endowed with the pushout topology along the inclusions $\{c\} \hookrightarrow \F_\Prod$ and $\{c\} \hookrightarrow \F_\Sum$ (the finest topology making both inclusion maps $\F_\Prod \hookrightarrow \U$ and $\F_\Sum \hookrightarrow \U$ continuous, equivalently, the quotient topology from $\F_\Prod \sqcup \F_\Sum$ under the identification of the two copies of $c$).
The space $\U$ is path-connected and contractible,  but is not a manifold near $c$.
A neighborhood of $c$ in $\U$ consists of an $(n-1)$-dimensional disk (the $\F_\Prod$-side) wedged to a $1$-dimensional segment (the $\F_\Sum$-side), and this local asymmetry plays a substantive role in the union analysis (\Cref{rem:union}).

\begin{definition}[Wedge Metric on $\U$]\label{def:wedge-metric}
The \emph{wedge metric} $d_\U$ on $\U$ restricts to $\Aitch$ on $\F_\Prod$ and to $\dSum$ on $\F_\Sum$.
For $f \in \F_\Prod$ and $g \in \F_\Sum$, the distance is the length of the shortest path through the common point $c$,
\[
d_\U(f, g)  := \Aitch(f, c) + \dSum(c, g).
\]
\end{definition}

\subsection{Marginal Exchange Rates and Structural Ratios}\label{sec:exchange-rates}

For $f \in \F$ and $I \in \R^A_{>0}$, the  \emph{marginal exchange rate} of asset $i$ for asset $j$ at $I$ is
\[
p_{ij}(f, I) := \frac{\partial_i f(I)}{\partial_j f(I)}.
\]
Using (T2) it follows that $p_{ij} > 0$ everywhere on $\R^A_{>0}$.
The cocycle identity $p_{ij} \cdot p_{jk} = p_{ik}$  in fact holds for \emph{every} $f \in \F$ (it follows directly from $p_{ij} = \partial_i f/\partial_j f$ by cancellation), not just on $\F_\Prod$.
We call it the \emph{price cocycle}.
In the finance literature this identity is the \emph{no-arbitrage condition} on triangular exchange rates, ensuring that one cannot profit by cycling through three assets at the marginal prices the pool quotes.

By (T4), $p_{ij}$ is invariant under  simultaneous rescaling $I \mapsto \lambda I$ for $\lambda > 0$, so $p_{ij}$ descends to a function on the projective inventory space $\mathbb{P}^A_{>0} := \R^A_{>0}/\R_{>0}$.
We write $[I]$ for the class of $I$.

For $f_\alpha \in \F_\Prod$, we compute
\[
p_{ij}(f_\alpha, I) = \frac{\alpha_i I_j}{\alpha_j I_i} = \rho_{ij}(\alpha) \cdot \frac{I_j}{I_i}, \quad \rho_{ij}(\alpha) := \frac{\alpha_i}{\alpha_j}.
\]
We call the value $\rho_{ij}(\alpha)$ the \emph{structural ratio} of the pair $(i,j)$ at $\alpha$.
The factorization \mbox{$p_{ij}(f_\alpha, I) = \rho_{ij}(\alpha) \cdot (I_j/I_i)$} into an $\alpha$-dependent factor and an $I$-dependent factor is special to $\F_\Prod$, and we shall exploit it in \Cref{sec:F0-characterization}.
The structural ratios satisfy their own multiplicative  cocycle $\rho_{ij}(\alpha) \cdot \rho_{jk}(\alpha) = \rho_{ik}(\alpha)$, distinct from the price cocycle of $p_{ij}$.

For $f_\gamma \in \F_\Sum$, we compute
\[
p_{ij}(f_\gamma, I) = \left( \frac{I_j}{I_i} \right)^{1 - \gamma},
\]
a power law in the inventory ratio with exponent independent of $i, j$,   and with no $\F_\Prod$-style factorization into a $\gamma$-part times an $I$-part.

Two properties of the Aitchison embedding $\clr$ are central to the analysis on $\F_\Prod$.
First,  $\clr(\alpha)_i - \clr(\alpha)_j = (\log \alpha_i - \bar g) -  (\log \alpha_j - \bar g) = \log(\alpha_i/\alpha_j) = \log \rho_{ij}(\alpha)$ 
(where \mbox{$\bar g = \tfrac{1}{n}\sum_k \log \alpha_k$}) 
implies that  $\log \rho_{ij}(\alpha) = \clr(\alpha)_i - \clr(\alpha)_j$.
This shows that the logarithmic structural ratio is a linear functional on $H$.
Second, the Aitchison hull of a finite set $\{\alpha_\ell\} \subset \simplex$ corresponds, under $\clr$, to the Euclidean convex hull of \mbox{$\{\clr(\alpha_\ell)\} \subset H$}, since $\clr$ is, by the definition of $\Aitch$ as the pullback, an isometry between $(\simplex, \Aitch)$ and the Euclidean hyperplane $H$ \citep{egozcue-pawlowsky-glahn-mateu-figueras-barcelo-vidal-2003,aitchison-1986,pawlowsky-glahn-egozcue-tolosana-2015}.
Aitchison-geodesic convexity is thus Euclidean convexity after $\clr$.
The identification of this hull with the Pareto-undominated  set is established in \Cref{lem:pareto-hull}.

\subsection{Liquidity Providers and Their Preferences}\label{sec:lps}

We consider a fixed set $\Lc := \{1, 2, \ldots, m\}$ of liquidity providers, and each LP $\ell$ is characterized by two pieces of data, a deposit and a design preference.
The deposit records what the LP contributes to the pool, and the design preference records which trading function the LP would like the pool to use.
We take these in turn, first defining the deposit and the valuation it induces, then the design preference.

\begin{definition}[Deposit]\label{def:deposit}
The \emph{deposit} of LP $\ell$ is a vector $D_\ell \in \R^A_{>0}$, where the $i$-th coordinate is the number of units of asset $i$ that LP $\ell$ contributes to the pool.
The \emph{joint deposit profile} is  $D := (D_1, \ldots, D_m) \in (\R^A_{>0})^m$.
\end{definition}

To weight LP contributions independently of any mechanism output, we fix a single reference function $\fref$ in advance and use it to assign each deposit a numerical valuation.

\begin{definition}[Reference Valuation Function and Valuations]\label{def:valuation}
The \emph{reference valuation function} is the symmetric constant-product market maker
\[
\fref(I) := \left( \prod_{i=1}^n I_i \right)^{1/n},
\]
the geometric mean of the inventory coordinates.
The \emph{valuation}  of LP  $\ell$ is $V_\ell := \fref(D_\ell)$, and the \emph{valuation profile} is $V := (V_1,  \ldots, V_m) \in \R^m_{>0}$, with total $V_\mathrm{tot}  := \sum_\ell V_\ell$.
\end{definition}

The reference $\fref$ coincides numerically with $f_0$, the $\gamma = 0$ element of $\F_\Sum$ (\Cref{def:Fgamma}), but plays a different role here, as an external valuation device rather than a member of a trading-function family.
We fix one valuation function rather than the equivalence class, so that the formulas of \Cref{thm:characterization-F0} have a concrete reference point.
The geometric mean is canonical, being the unique element of $\F_\Prod \cap \F_\Sum$ (\Cref{lem:intersection}) and the unique permutation-invariant element of $\F_\Prod$.
By \Cref{rem:valuation-robustness} the characterization does not depend on this choice, so any continuous, $1$-homogeneous, asset-symmetric function would serve.

The second LP input is the \emph{design preference}  $\succeq_\ell$, a continuous, complete, transitive preorder on the relevant design subspace.
We restrict attention to preferences that are single-peaked in the metric of that subspace.

\begin{definition}[Preference Profile and Peak Profile]  \label{def:profile}
Let $(\F', d)$ be a design subspace.
A~\emph{preference profile} on $(\F', d)$ is a tuple $(\succeq_1, \ldots, \succeq_m)$, where each $\succeq_\ell$ is single-peaked with peak $\theta_\ell^* \in \F'$, that is, $f \succeq_\ell g \iff d(f, \theta_\ell^*) \leq d(g, \theta_\ell^*)$.
Since $\succeq_\ell$ is uniquely determined by $\theta_\ell^*$ given $d$, we identify preference profiles with \emph{peak profiles} throughout and write $\theta^* =  ( \theta_1^*,  \ldots, \theta_m^*)$ for either.
\end{definition}

On the two design subspaces of special interest, we write the peak as the parameter that indexes it: $\alpha_\ell^* \in \simplex$ on $\F_\Prod$ (the weight vector with $\theta_\ell^* = f_{\alpha_\ell^*}$, \Cref{def:F0}) and $\gamma_\ell^* \in (-\infty, 1]$ on $\F_\Sum$ (the scalar with $\theta_\ell^* = f_{\gamma_\ell^*}$, \Cref{def:Fgamma}).
The metric-induced form of single-peakedness used in \Cref{def:profile} is stronger than the classical single-peakedness of \citet{black-1948}, where fixing the peak still leaves open the rate at which preference decreases away from it: under the metric-induced form, the metric $d$ and the peak $\theta_\ell^*$ together determine the preference relation $\succeq_\ell$ uniquely.
The unrestricted-domain axiom (A1) of \Cref{sec:axioms} is taken throughout to mean unrestricted within this single-peaked class: the mechanism's domain is the full product space of single-peaked  preference (equivalently, peak) profiles on $\F'$, with no further restriction on peak placement.

\paragraph{Aitchison-single-peakedness on \texorpdfstring{$\F_\Prod$}{F\_Prod}.}
On $\F_\Prod$ with $d = \Aitch$, single-peakedness reads
\[
\alpha \succeq_\ell \beta \iff \Aitch(\alpha, \alpha_\ell^*)  \leq  \Aitch(\beta, \alpha_\ell^*).
\]
This is the natural preference model on $\F_\Prod$, not, say, Euclidean single-peakedness on the simplex, and the reason is economic.
An LP's economic exposure is determined by the pool's exchange rates, which on $\F_\Prod$ are the structural ratios $\rho_{ij}(\alpha) = \alpha_i/\alpha_j$.
These ratios are log-linear in the Aitchison (clr) coordinates: $\log \rho_{ij}(\alpha) = \clr(\alpha)_i - \clr(\alpha)_j$.
Two weight vectors $\alpha$ and $\alpha'$ that are close in the Aitchison metric agree closely on all their structural ratios; two vectors that are close in Euclidean distance on $\simplex$ may disagree substantially on ratios involving small coordinates.
Because LP risk and return depend on exchange rates, an LP whose preferred rate profile is $\alpha_\ell^*$ prefers
  $\alpha$ over $\beta$ precisely when $\alpha$'s exchange rates are uniformly closer (in log-scale) to the preferred rates, which is exactly Aitchison-single-peakedness.
Under this assumption, indifference sets are Euclidean spheres in $H$  via the clr isometry of \Cref{sec:exchange-rates}, so all subsequent arguments reduce to standard Euclidean geometry.
Two arguments exploit this quadratic structure, the Pareto-hull identification of \Cref{lem:pareto-hull} (Hilbert projection on $H$) and the manipulation construction of \Cref{lem:sp-forces-projection} (Euclidean distance to the peak).
The characterization itself does not need it, and the Pareto step generalizes to any preferences with closed convex Pareto sets in clr-coordinates (\Cref{rem:beyond-single-peaked}).

\begin{remark}[A Log-Ratio Loss Model Yielding Aitchison Single-Peakedness]\label{rem:lp-utility-aitchison}
The Aitchison-single-peakedness assumption is not merely a convenient functional form but can be derived from a primitive economic objective.
Suppose LP $\ell$ targets the ratios $\rho_{ij}^* := \rho_{ij}(\alpha_\ell^*)$  and dislikes a pool $f_\alpha$ in proportion to the squared log-rate deviation $\mathcal{D}_\ell(\alpha) := \sum_{i<j} (\log \rho_{ij}(\alpha) - \log \rho_{ij}^*)^2$ (a rate-targeting objective that is asset-symmetric, scale-invariant, and additive across pairs).
Using $\log \rho_{ij}(\alpha) = \clr(\alpha)_i - \clr(\alpha)_j$, the loss becomes $\sum_{i<j}(\clr(\alpha)_i - \clr(\alpha)_j - x_i + x_j)^2$ with $x := \clr(\alpha_\ell^*) \in H$; setting $y := \clr(\alpha) - x \in H$ this is $\sum_{i<j}(y_i - y_j)^2$, and the identity
\[
\sum_{i<j}(y_i - y_j)^2 = n \|y\|^2  \qquad \text{for every } y \in H,
\]
where $\|y\|^2 = \sum_k y_k^2$ denotes the Euclidean norm on $\R^n$ (which restricts to the Euclidean norm on $H \subset \R^n$ used elsewhere), yields $\mathcal{D}_\ell(\alpha) = n\,\|\clr(\alpha) - \clr(\alpha_\ell^*)\|^2 = n\,\Aitch(\alpha, \alpha_\ell^*)^2$.
The identity follows from direct expansion using $y \in H$, hence $\sum_k y_k = 0$.
The algebra is recorded in \Cref{app:identity}.
Thus, ordering pools by $\mathcal{D}_\ell$ is ordering by Aitchison distance to $\alpha_\ell^*$.
The same identity makes the Aitchison-metric (SP) of \Cref{sec:axioms-sp} the  natural manipulation-resistance condition under this loss.
A primitive derivation from realized share value is left open: an LP's realized value depends on the trade path and accrued fees, not on the terminal pool weights alone, so it does not reduce to a closed-form function of $\alpha$ and the quadratic log-ratio loss is adopted here as a tractable proxy (\Cref{sec:open-questions}, especially the fee-sensitive direction (Q4)).
\end{remark}

\paragraph{Single-peakedness  on \texorpdfstring{$\F_\Sum$ and $\U$}{F\_Sum and U}.}
On $\F_\Sum$ with $d = \dSum$, single-peakedness reads
\[
f_\gamma \succeq_\ell f_{\gamma'} \iff  |\arctan \gamma - \arctan \gamma_\ell^*| \leq |\arctan \gamma' - \arctan \gamma_\ell^*|.
\]
On $\U$ with $d = d_\U$ (\Cref{def:wedge-metric}), single-peakedness is taken with respect to the wedge metric.

\subsection{The Mechanism}\label{sec:mechanism}

We now define the central object of study.

\begin{definition}\label{def:mechanism}
Fix a design subspace $(\F', d)$.
An \emph{AMM design mechanism on $\F'$} is a map
\[
M \colon (\F' \times  \R^A_{>0})^m \to \F',
\]
written $M(\theta^*, D)$, with $\theta^* =  (\theta_1^*,  \ldots, \theta_m^*)$ the peak profile (\Cref{def:profile}) and $D$ the joint deposit profile (\Cref{def:deposit}).
In words, $M$ takes as input each LP's preferred trading function and deposit vector, and returns a single trading function for the pool: it is the aggregation rule.
\end{definition}

We use four notational forms for the same mechanism, displayed with different arguments for context: $M(\theta^*, D)$ in full; $M(\theta^*, V)$ when deposit-dependence factors through the valuation profile $V$ of \Cref{def:valuation}; $M(\theta^*)$ when the mechanism does not depend on deposits at all; and $M(\alpha^*, D)$ specifically on $\F_\Prod$, where each peak is a weight vector $\alpha_\ell^* \in \simplex$ and we rename  $\theta^*$ as $\alpha^*$ accordingly.

The results in this paper instantiate the framework on the three concrete design subspaces of \Cref{sec:design-subspaces}: $(\F_\Prod, \Aitch)$, $(\F_\Sum, \dSum)$, and $(\U, d_\U)$.
Extension of the theory to other design subspaces is open (\Cref{sec:open-questions}, Q1).
The obstruction is that the proofs use the specific metric geometry of each subspace, which a general $\F'$ does not supply.

\begin{remark}[Mechanism Viewpoint Generalizes Trading-Function Viewpoint]\label{rem:mechanism-vs-function}
The literature \citep{schlegel-kwasnicki-mamageishvili-2023, frongillo-papireddygari-waggoner-2024, bichuch-feinstein-2025} axiomatizes the trading function $f$ as a self-standing object: a constant mechanism outputting the same $f$ for every profile is, in our framework, such a choice of $f$.
The mechanism viewpoint strictly generalizes it by letting $M$ depend on the LP profile, 
 accommodating genuine aggregation.
This is what makes Arrovian aggregation questions well-posed in the first place.
The static viewpoint admits no profile dependence to axiomatize.
Strategy-proofness in particular is a rule-level property, and the impossibility of \Cref{thm:impossibility} is correspondingly absent from the trading-function-level setting.
Indeed, a constant mechanism trivially satisfies strategy-proofness, since no agent can manipulate an output that ignores all reports, so the tension our impossibility theorem identifies cannot even be stated at that level of abstraction.
Fees play no role in our axioms and are excluded from \Cref{def:mechanism} (a fee-sensitive theory is open, \Cref{sec:open-questions}).
\end{remark}

\subsection{Induced Dynamics}\label{sec:dynamics}

After initialization, the pool evolves under two kinds of operation: trades and LP withdrawals.
A trade moves the pool's inventory $I \to I'$ along a level set of the trading function $f$, that is, $f(I') = f(I)$. 
 The trading function itself is unchanged.

For withdrawals, each LP is equipped with  a \emph{share}, a fungible claim on the pool's value.
The standard DeFi convention assigns LP $\ell$ a share count $s_\ell^{(0)} := f(D_\ell)$, with total supply $S^{(0)} := \sum_\ell f(D_\ell)$.
Fungibility follows from (T4): scaling reserves by $\lambda$ scales each share value by $\lambda$, so all shares represent identical fractional claims on the pool.
Let $s_\ell$ denote LP $\ell$'s current share and $S := \sum_k s_k$ the total supply.
An \emph{LP withdrawal} by LP $\ell$  at fraction $\lambda \in [0,1]$ redeems $\lambda s_\ell$ shares and pays out inventory $\lambda \tfrac{s_\ell}{S} I$, leaving post-withdrawal inventory $I' = (1 - \lambda s_\ell/S) I$, a positive scalar multiple of $I$.
Since $I'$ lies on the same projective ray as $I$, the level set of $f$ through the pool's pre- and post-withdrawal points differs only by a scalar, and all marginal exchange rates $p_{ij}(f, [I])$ are unchanged on the inventory ray.
This is the invariance that axiom (A5) of \Cref{sec:axioms-WN} elevates to a normative requirement on the mechanism's output.

Trade dynamics are interpretive only; no proof uses them.
Withdrawal dynamics enter the formal framework solely through axiom (A5) of \Cref{sec:axioms-WN}.
Both operations leave the trading function invariant.
$f$ changes only when $M$ reselects it on a new LP profile, 
 which is the only event our axioms govern.
This separation is what lets the aggregation question be posed cleanly: trade and withdrawal flows are parameters of the deployed pool, while the choice of $f$ itself is the social-choice problem our axioms address.

\subsection{Notation} \label{sec:notation}

\Cref{tab:notation} collects the main symbols.
The asset-weight simplex $\simplex = \{\alpha \in \R^n_{>0} \mid \sum_i \alpha_i = 1\}$ has dimension $n-1$ and parametrizes $\F_\Prod$.
The LP-weight simplex $\Delta^{m-1} =  \{w \in \R^m_{\geq 0} \mid \sum_\ell w_\ell = 1\}$ has dimension $m-1$; its inhabitants are weight vectors $w = (w_1, \ldots, w_m)$, one weight per LP, with which the mechanism combines the LPs' peaks $\theta_\ell^*$ into a single output (\Cref{thm:characterization-F0}; the weighting function $w(V)$ is introduced there as a map $\R^m_{>0} \to \Delta^{m-1}$ and characterized by the Arrovian core).
The two simplices live in different ambient spaces and should not be conflated.

\begin{table}[ht]
\centering
\small
\begin{tabular}{@{}l p{0.71\textwidth}@{}}
  \hline
Symbol & Meaning \\
\hline
 $n,\ m$ & numbers of assets and LPs ($n \geq 2$, mostly $n \geq 3$; $m \geq 2$) \\
 $A,\ \Lc$ & asset set $\{1, 2, \ldots, n\}$, LP set $\{1, 2, \ldots, m\}$ \\
$I \in \R^A_{>0}$ & pool inventory (reserve) vector; coordinate $I_i$ is units of asset $i$ \\
$\F' \subseteq \F$ & generic design subspace;  instantiated as $\F_\Prod$, $\F_\Sum$, or $\U$ (\Cref{sec:design-subspaces}) \\
$\succeq_\ell$ & LP $\ell$'s single-peaked preference on $\F'$, defined by a metric and peak \newline (\Cref{def:profile}) \\
$\theta_\ell^*,\ \theta^*$ & LP $\ell$'s peak in $\F'$;  peak profile $(\theta_1^*,  \ldots, \theta_m^*)$; specializes to $\alpha_\ell^*$ on $\F_\Prod$ and $\gamma_\ell^*$ on $\F_\Sum$ \\
$\mathbb{P}^A_{>0}$ & projective inventory space $\R^A_{>0}/\R_{>0}$; domain of $p_{ij}$ by (T4) \\
$\simplex$ & \emph{open} asset-weight simplex (dimension $n-1$); domain of $\alpha$ \\
$\Delta^{m-1}$ & \emph{closed} simplex of LP weight vectors; codomain of $w$ \\
$\F,\F_\Prod,\F_\Sum,\U$ & design space; weighted-CPMM, symmetric-CEMM, and union subspaces \\
$f_\alpha,\ f_\gamma,\ c$ & weighted product, symmetric power mean, symmetric CPMM \\
 $\alpha_\ell^*,\ \gamma_\ell^*$  & LP $\ell$'s peak on $\F_\Prod$ (in $\simplex$) and on $\F_\Sum$ (in $(-\infty,1]$) \\
$\rho_{ij}(\alpha)=\alpha_i/\alpha_j,\ p_{ij}$ & structural ratio of $(i,j)$; marginal exchange rate \\
$\clr,\ H$ & centered log-ratio map; zero-sum hyperplane  $\cong \R^{n-1}$ \\
$\Aitch,\ \dSum,\ d_\U$ & Aitchison metric on $\simplex$; metric on $\F_\Sum$; wedge metric on $\U$ \\
$D_\ell,\ V_\ell = \fref(D_\ell),\ V,\ V_\mathrm{tot}$ & LP $\ell$'s deposit; LP $\ell$'s valuation;  \newline valuation profile $V = (V_1,  \ldots, V_m) \in \R^m_{>0}$; total $V_\mathrm{tot} = \sum_\ell V_\ell$ \\
$\fref$ & reference function $(\prod_i I_i)^{1/n}$; external valuation device (\Cref{def:valuation}) \\
$s_\ell^{(0)},\ s_\ell,\ S$ & LP $\ell$'s initial share; current share; current total supply $S = \sum_k s_k$  \newline (\Cref{sec:dynamics}) \\
$P(\theta^*)$ & Pareto-undominated set of a profile;  equals the Aitchison hull (\Cref{lem:pareto-hull}) \\
$M(\theta^*,D),\  w_\ell(V)$ & mechanism  (aggregation rule); LP weighting (continuous, equivariant) \\
$\bar\alpha = M(\alpha^*,D)$ & mechanism output on $\F_\Prod$ (a weight vector in $\simplex$) \\
$\hat\theta_\ell,\ \hat\alpha_\ell$ & a misreported profile entry/peak by LP $\ell$   (deposits held fixed; in proofs of manipulability the manipulating LP is denoted $\ell_0$; see (SP), \Cref{sec:axioms}) \\
\hline
\end{tabular}
\caption{Principal notation. Auxiliary symbols are introduced locally where needed.}
\label{tab:notation}
\end{table}

\section{Axioms}\label{sec:axioms}

The first group, (A0)--(A4) together  with continuity (C), constitutes the \emph{Arrovian core}, on which the two characterizations of \Cref{sec:F0-characterization,sec:Fgamma-characterization} and the union decomposition (\Cref{rem:union}) rest.
The strategy-proofness axiom (SP) is added in \Cref{sec:impossibility}.
An additional axiom, withdrawal neutrality (A5), is introduced at the end of this section as an optional requirement. 
Adding it to the Arrovian core selects a canonical member of the characterized  family (\Cref{thm:characterization-F0}, part (b)).

\paragraph{Framework constants.}
Throughout this section we fix:
\begin{itemize}
\item a design subspace $(\F', d)$ (\Cref{def:design-subspace}), where $\F' \subseteq \F$  is a subset of trading-function equivalence classes and $d$ is  a metric on $\F'$;
\item the reference valuation function $\fref$ (\Cref{def:valuation}).
\end{itemize}
These are not quantified over. 
They are the ambient data against which all axioms are stated.

\paragraph{Quantification.}
Each axiom listed below expresses
a normative property of a mechanism
 $ M \colon (\F' \times \R^A_{>0})^m \to \F' $
 (\Cref{def:mechanism}), quantified universally over  its inputs: the LP profile $(\theta^*, D) \in (\F')^m \times (\R^A_{>0})^m$, with peak profile $\theta^* = (\theta_1^*,  \ldots, \theta_m^*)$ (\Cref{def:profile}; each $\theta_\ell^*$ determines LP $\ell$'s single-peaked preference $\succeq_\ell$ together with $d$) and joint deposit profile $D = (D_1,  \ldots, D_m)$
  (\Cref{def:deposit}). 
  The induced valuation profile is $V = (\fref(D_1),  \ldots, \fref(D_m)) \in \R^m_{>0}$. 
Some axioms introduce additional bound variables.
These are quantified universally where they appear, unless stated otherwise.

\subsection{The Arrovian Core}\label{sec:arrovian-core}

We state the six axioms in turn, beginning with two domain conditions.

\begin{description}
\item[(A0) Nontriviality.]
The mechanism $M$ is not constant on its domain.

\item[(A1) Unrestricted Domain.]
The mechanism $M$ is defined on the entire space  $(\F' \times \R^A_{>0})^m$ of LP profiles.
\end{description}

The \emph{Pareto-undominated set} of a profile $\theta^* = (\theta_1^*, \ldots, \theta_m^*)$ on $\F'$ (which, by \Cref{def:profile}, also encodes the preference profile $(\succeq_1,  \ldots, \succeq_m)$) is
\[
P(\theta^*) := \bigl\{ f \in \F' \mid \text{no } f' \in \F' \text{ satisfies } f' \succ_\ell f \text{ for every } \ell = 1, 2, \ldots, m \bigr\}.
\]
For Aitchison-single-peaked preferences on $\F_\Prod$,   we establish in \Cref{lem:pareto-hull} that $P(\theta^*)$ equals the Aitchison hull of $\{\alpha_\ell^*\}_{\ell=1}^m$.

\begin{description}
\item[(A2) Pareto Efficiency.]
$M(\theta^*, D) \in P(\theta^*)$.
\end{description}

Condition (A2) is the standard \emph{weak}-Pareto requirement: 
 no alternative is strictly preferred by every LP.
We work throughout with the weak version; 
 the strong-Pareto question (no alternative weakly preferred by all and strictly by some) is not used in any result of the paper.

The next axiom is the mechanism-level analog of Arrow's classical independence of irrelevant alternatives, 
 and is the most consequential of the Arrovian axioms.
It is, however, of a different kind from Arrow's welfare-level IIA: a restriction on cross-LP information flow rather than on profile-level alternatives.
 Per \Cref{rem:A3-vs-skm}, it is strictly stronger than the trading-function-level independence of \citet{schlegel-kwasnicki-mamageishvili-2023} on nonconstant mechanisms.
 It does not forbid an LP from holding $(i,j)$-preferences correlated with $(i,k)$-preferences.
   Such an LP simply has a peak $\alpha_\ell^*$ whose ratios reflect the correlation, which (A3) leaves intact.
  What it forbids is the \emph{aggregator} consulting an LP's views about one asset pair when setting the pool's exchange rate on another (its WBTC views when computing the ETH/USDC structural ratio).
  The content is a cross-LP restriction (pairs are not mixed across LPs),  not a within-LP one (LP preferences need not be cross-pair separable).

The normative case for (A3) is concrete in the AMM setting.
A pool's $(i,j)$-pricing is, on-chain, a quantity an LP can read directly from the trading function and verify against their own $(i,j)$-stake; if that price secretly depends on every other LP's views about every other asset pair, no LP can audit it against any local input, and the aggregator becomes a black box even though every parameter is public.
Mechanism IIA is the modularity condition that lets a pool's $(i,j)$-price be justified by the $(i,j)$-relevant inputs alone, which is what makes it auditable in the sense the deployment context demands.
A sophisticated LP might object that they regard cross-pair coherence as a goal rather than an auditing burden, preferring a rule that consults all pairwise opinions jointly to produce a globally consistent weighting.
(A3) is precisely the axiom that forecloses such a rule, and its normative weight rests on the primacy of per-pair auditability in transparent on-chain protocols.
The trade-off it encodes is treated as a substantive design choice in \Cref{sec:practical-implications} (Route~1: relaxing (A3)).
To state (A3) formally, we need the notion of a \emph{pair-restricted preference}.
For a fixed pair $i, j \in A$ of distinct assets, each $f \in \F'$  determines an $(i,j)$-pricing function $I  \mapsto p_{ij}(f, I)$ on the projective inventory space $\mathbb{P}^A_{>0}$.
LP $\ell$'s \emph{$(i,j)$-restricted preference} $\succeq_\ell^{(ij)}$ ranks trading functions in $\F'$ by the closeness of their $(i,j)$-pricing functions to LP $\ell$'s preferred $(i,j)$-pricing function, where closeness is measured in a metric appropriate to $\F'$ (on $\F_\Prod$,  the metric we introduce next; on $\F_\Sum$, immediately below; in general, any metric compatible with the underlying single-peaked preferences on $\F'$).
For two weighted-product trading functions $f_\alpha,  f_\beta \in \F_\Prod$ with weight vectors  $\alpha, \beta \in \simplex$ (\Cref{def:F0}), the \emph{$(i,j)$-restricted distance} is the logarithmic metric
\[
d_{ij}(f_\alpha, f_\beta) := |\log \rho_{ij}(\alpha) - \log \rho_{ij}(\beta)|,
\]
the Aitchison distance $\Aitch(f_\alpha, f_\beta)$ projected to the pair $(i,j)$
 (i.e., restricted to the one-dimensional subspace of $H$  corresponding to the log-ratio $\log(\alpha_i/\alpha_j)$), so that $f \succeq_\ell^{(ij)} g$ on $\F_\Prod$ if and only if $d_{ij}(f, \theta_\ell^*) \leq d_{ij}(g, \theta_\ell^*)$.
On $\F_\Prod$, where $p_{ij}(f_\alpha, I) = \rho_{ij}(\alpha) \cdot (I_j/I_i)$, the $(i,j)$-restricted preference is determined by the single number $\rho_{ij}(\alpha_\ell^*)$, since $d_{ij}$ measures  the distance in log-scale of the structural ratio.
On $\F_\Sum$, where $p_{ij}(f_\gamma, I) =  (I_j/I_i)^{1-\gamma}$, the $(i,j)$-restricted distance is
\[
d_{ij}(f_\gamma, f_{\gamma'}) := |\arctan \gamma - \arctan \gamma'| = \dSum(\gamma, \gamma'),
\]
the same for every pair $(i,j)$ since the exponent $1 - \gamma$ is pair-independent.
The $(i,j)$-restricted preference is correspondingly determined by the single number $\gamma_\ell^*$.

\begin{description}
\item[(A3) Mechanism IIA.]
For each ordered pair of distinct assets $i, j \in A$,  the $(i,j)$-pricing function $I \mapsto p_{ij}(M(\theta^*, D), I)$ of the mechanism's output is determined, as a function of $I$, by the LPs' $(i,j)$-restricted preferences $(\succeq_\ell^{(ij)})_{\ell = 1}^{m}$ and the valuation profile $V$ alone.
\end{description}

On $\F_\Prod$, the $I$-dependence of the  pricing is  already fixed, $p_{ij}(f_\alpha, I) = \rho_{ij}(\alpha) \cdot (I_j/I_i)$, with the structural ratio $\rho_{ij}(\alpha) = \alpha_i/\alpha_j$ independent of $I$, so (A3) constrains only the scalar coefficient:
 it is equivalent to requiring that $\rho_{ij}(M(\theta^*, D))$ depend only on $(\succeq_\ell^{(ij)})_\ell$ and $V$.
On $\F_\Sum$, the pricing is $(I_j/I_i)^{1-\gamma}$, so (A3) requires only that the exponent (equivalently, $\gamma$) depend on the $(i,j)$-restricted data.
Since any pair-restriction already determines the full preference there, (A3) is vacuous (\Cref{lem:A3-vacuous-Fgamma}).

On $\F_\Prod$ the scalar coefficient $\rho_{ij}(M(\alpha^*, D))$ need not, despite the notation, equal any single peak's ratio $\alpha_i^*/\alpha_j^*$; (A3) requires only that it be a function of $(\rho_{ij}(\alpha_\ell^*))_\ell$ and $V$.
A concrete violation is exhibited in \Cref{app:independence}: dispersion-modulated weights that use each peak's full clr-norm rather than its pair-restrictions.

\begin{remark}[Mechanism IIA Versus Trading-Function Independence] \label{rem:A3-vs-skm}
(A3) differs from the trading-function-level independence of  \citet{schlegel-kwasnicki-mamageishvili-2023}, which asks that the marginal rate between $i$ and $j$ be independent of any third asset's inventory.
That is a property of $f$, whereas  (A3) is a property of the \emph{mechanism}: 
 the rule producing $p_{ij}$ uses only $(i,j)$-restricted preferences.
The two coincide for constant mechanisms, where (A3) reduces to SKM's property of $f$.
On nonconstant mechanisms, (A3) imposes cross-LP locality that no function-level axiom can express.
If two profiles agree on every LP's $(i,j)$-preference, the mechanism must assign them the same $(i,j)$-marginal rate, regardless of how they differ on other pairs.
\end{remark}

\begin{description}
\item[(A4) Anonymity.]
For every permutation $\sigma$ of the LP label set $\{1, 2, \ldots, m\}$, 
it holds that 
 $M(\sigma \cdot \theta^*, \sigma \cdot D) = M(\theta^*, D)$, where
\[
\sigma \cdot \theta^* := (\theta_{\sigma(1)}^*,  \ldots, \theta_{\sigma(m)}^*)
 \quad \text{and} \quad
  \sigma \cdot D := (D_{\sigma(1)},  \ldots, D_{\sigma(m)}).
\]
\end{description}

Note that (A4) permutes each LP's peak and deposit \emph{together}: it is invariance under the diagonal action of the symmetric group on the paired data $(\theta_\ell^*, D_\ell)$, not separate invariance in peaks and in deposits. 
We call this \emph{joint equivariance}, and it is the precise form in which anonymity is used throughout the paper (in \Cref{lem:pair-extraction-F0,lem:V-dependent-weights,lem:sp-forces-projection}), since the proofs never permute peaks and deposits independently but always move them as a unit.
 The derived weighting $w$ inherits it as \mbox{$w_{\sigma(\ell)}(V) = w_\ell(\sigma\cdot V)$}, which forces equal weights only at symmetric valuation profiles.

\begin{description}
\item[(C) Continuity.]
The map $M \colon (\F' \times \R^A_{>0})^m \to \F'$ is continuous when $(\F')^m \times (\R^A_{>0})^m$ carries the product topology (with the $C^2$ topology on $\F'$ and the Euclidean topology on $\R^A_{>0}$) and $\F'$ carries the $C^2$ topology.
\end{description}

On the design subspaces of present interest, the $C^2$ topology on $\F'$ reduces to a familiar Euclidean topology, via the identifications $\F_\Prod \cong \simplex \subseteq \R^n$ and $\F_\Sum \cong (-\infty, 1] \subseteq \R$ recorded in \Cref{sec:design-subspaces}.
The union $\U = \F_\Prod \cup \F_\Sum$ carries the pushout topology, equivalent to convergence on each arm separately.
Continuity is substantive rather than merely technical: discontinuous tie-breaking schemes can satisfy the other axioms but introduce instability in the mechanism's output under arbitrarily small profile perturbations (\Cref{app:independence}, separating example for (C)).

We refer to (A0)--(A4) together with (C) as the \emph{Arrovian core}.

 \subsection{Strategy-Proofness}\label{sec:axioms-sp}

The strategy-proofness axiom asks that no LP can  manipulate the mechanism's output to their advantage by misreporting their peak.

\begin{description}
\item[(SP) Strategy-proofness.]
For every LP $\ell$ and every misreport $\hat\theta_\ell \in \F'$ (the peak LP $\ell$ submits in place of its true peak $\theta_\ell^*$),
\[
d\bigl( M(\theta^*, D),\, \theta_\ell^* \bigr) \leq d\bigl( M(\hat\theta^{(\ell)}, D),\, \theta_\ell^* \bigr),
\]
where $\hat\theta^{(\ell)} := (\theta_1^*, \theta_2^*, \ldots, \theta_{\ell-1}^*, \hat\theta_\ell, \theta_{\ell+1}^*, \ldots, \theta_m^*)$ is the $m$-tuple in which LP $\ell$'s peak alone is replaced by $\hat\theta_\ell$ and every other LP reports truthfully.
\end{description}

Truth-telling is thus a weakly dominant strategy for every LP
 (i.e., truth-telling is at least as good as any misreport, not necessarily strictly better), the inequality holding for every profile and misreport regardless of what the others report.
This is the strongest strategy-proofness requirement,  and any mechanism meeting it also satisfies weaker ex-post variants.
On $\F_\Prod$ with the Aitchison metric,  (SP) is the Aitchison-metric analog of multidimensional strategy-proofness on Euclidean space.
\citet{border-jordan-1983} characterize strategy-proof rules on $\R^n$ via coordinatewise medians (phantom voters), in the spirit of our setting once the range is the hyperplane $H$.
Our impossibility does not, however, require their characterization:  the only strategy-proofness fact we use is the self-contained \Cref{lem:sp-forces-projection}, 
 which establishes directly that a linear aggregator on $H$ is Euclidean-metric strategy-proof only if it is a single-coordinate projection.
Nor does it follow from theirs as a corollary: their domain is $\R^n$ rather than the constrained hyperplane $H \subset \R^n$, and they characterize strategy-proof rules as coordinatewise medians, whereas the Arrovian core here has already pre-restricted the candidate rules to the linear aggregators of \Cref{lem:sp-forces-projection}, a different class derived by a different argument.
On $\F_\Sum$ with $\dSum$, it is classical strategy-proofness on single-peaked preferences on the line \citep{moulin-1980}.

 \subsection{Withdrawal Neutrality}\label{sec:axioms-WN}

We close this section with an additional axiom, optional rather than part of the Arrovian core.

\begin{description}
\item[(A5) Withdrawal Neutrality.]
For every LP $\ell$, every withdrawal fraction  $\lambda \in [0, 1)$, and every inventory ray $[I] \in \mathbb{P}^A_{>0}$,
\[
p_{ij}\bigl( M(\theta^*, D'),\, [I] \bigr) = p_{ij}\bigl( M(\theta^*, D),\, [I] \bigr) \quad \text{for all distinct } i, j \in A,
\]
where $D'$ is the post-withdrawal deposit profile, $D'_\ell = (1-\lambda) D_\ell$ and $D'_k = D_k$ for $k \neq \ell$.
\end{description}

That is, (A5) requires the mechanism's output exchange rates at the inventory ray to be invariant under proportional rescaling of any single LP's deposit.
We require $\lambda < 1$, so that the valuation profile stays in $\R^m_{>0}$, the domain of $w$. 
The limit $\lambda \to 1$ sends $V_\ell$ to the boundary and removes LP $\ell$ from the pool entirely, which is a change in the LP set rather than a partial withdrawal of the kind covered by (A5); note that the post-withdrawal inventory ray is unchanged, since $I' = (1-\lambda s_\ell/S)I$ lies on the same ray as $I$, so no separate post-withdrawal ray appears in the condition.
We do not include (A5) among the Arrovian core, since it is logically independent of (A0)--(A4)+(C) and not required for the main characterization. 
\Cref{thm:characterization-F0}(b) below shows that, within the Arrovian family, (A5) selects exactly the deposit-blind \emph{symmetric} centroid, the unique weighting $w_\ell \equiv 1/m$.
Stake-proportional weighting, the economically natural alternative, is in tension with (A5) and is discussed in \Cref{rem:DWN-tension}.

 \subsection{Logical Independence}\label{sec:logical-independence}

The axioms (A1)--(A4) and (C) are mutually logically independent on $\F_\Prod$: for each, a mechanism violates it alone while satisfying the rest
 (explicit constructions can be found in \Cref{app:independence}).
The status of (A0) is more subtle and we state it consolidated here.
On the full domain, (A0) is \emph{derivable} from (A1)+(A2): at any unanimous profile the Aitchison hull is a singleton (\Cref{lem:pareto-hull}, proved in \Cref{sec:F0-pareto-pin}; the appeal here is not circular, as that lemma uses only Aitchison-single-peakedness from \Cref{sec:lps} and not (A0) itself), so (A2) forces the output to coincide with the varying peak, hence the mechanism cannot be constant.
If (A1)  is weakened to allow restricted domains, however, (A0) becomes independent: the constant rule $M \equiv c$, defined only on the all-$c$ subdomain, satisfies (A2) and (the weakened) domain condition without satisfying (A0), and is the unique violator in this sense (\Cref{app:independence}).
We list (A0) in the Arrovian core for expository convenience, so that the six axioms (A0)--(A4) and (C) form a single rhetorical unit whose members are all stated at the same logical level.
The impossibility of \Cref{thm:impossibility} could equivalently be stated with (A0) demoted to a consequence of (A1)+(A2), at the cost of a slightly less uniform axiom presentation.
 The strategy-proofness axiom (SP) is jointly inconsistent with the Arrovian core on $\F_\Prod$ (\Cref{thm:impossibility}) but consistent with it on $\F_\Sum$ (\Cref{thm:characterization-Fgamma}).
Axiom (A5) is logically independent of the Arrovian core:  it is violated by every stake-proportional or otherwise nonconstant weighting (\Cref{rem:DWN-tension}), and the symmetric centroid that satisfies it satisfies (A0)--(A4) and (C) (\Cref{thm:characterization-F0}(b)).
The independence picture is therefore that (A1)--(A4), (C), and (SP) each carry distinct logical content, (A0) is conditional on the domain hypothesis, and (A5) is a substantive but optional addition.


\section{Characterization: The Unique Fair Rule}\label{sec:F0-characterization}

This section identifies the unique family of aggregation rules on $\F_\Prod$ that are consistent with the Arrovian core.
We use ``fair'' throughout as shorthand for the bundle of six Arrovian-core axioms.
Mechanism IIA (A3) in particular is best read as a structural modularity condition (the rule producing the $(i,j)$-exchange-rate uses only $(i,j)$-restricted preference information), with its motivation discussed in \Cref{sec:arrovian-core}.
A shorter route would invoke  Genest's theorem on the pooling side and pull back through the FPW equivalence (\Cref{sec:pooling-transfer}).
We give the direct AMM-side derivation because it locates the structural content in the AMM problem itself and makes the role of $n \geq 3$ transparent.
The argument rests on a Cauchy-type functional equation on the cocycle of structural price ratios.
Continuous solutions of that equation are forced to be linear, and Pareto then pins the linear coefficients to a probability vector indexed by LPs, yielding the weighted-centroid form.

For the whole of \Cref{sec:F0-characterization} we assume $n \geq 3$ assets and $m \geq 2$ LPs, with Aitchison-single-peaked preferences parameterized by peaks $\alpha_\ell^* \in \simplex$.
The two-asset case is treated separately in \Cref{sec:two-asset}.
A consequence of the characterization to follow, worth stating in advance, is that any Arrovian mechanism on $\F_\Prod$ depends on the deposit profile $D$ only through the valuation profile $V = (\fref(D_1),  \ldots, \fref(D_m))$.
This is in part definitional, since (A3) is itself phrased in terms of the valuation profile.
The substantive content is the explicit form of the weights $w_\ell(V)$ derived in \Cref{lem:V-dependent-weights}, with robustness to the choice of reference function in \Cref{rem:valuation-robustness}.

In what follows, a function $w \colon \R^m_{>0}  \to \Delta^{m-1}$ is \emph{equivariant} if $w_{\sigma(\ell)}(V) = w_\ell(\sigma \cdot V)$ for every permutation $\sigma$ of $\{1, 2, \ldots, m\}$.
Values in $\Delta^{m-1}$ are nonnegative and sum to one.
The codomain $\Delta^{m-1}$ (the closed LP-weight simplex, dimension $m-1$) is distinct from the open asset simplex $\simplex$ of dimension $n-1$ (see \Cref{rem:weighting-codomain}).
We call the formula $M(\alpha^*, D)_i \propto \prod_{\ell=1}^m (\alpha_{\ell, i}^*)^{w_\ell(V)}$ the \emph{weighted Aitchison centroid} of the peaks $\{\alpha_\ell^*\}_{\ell=1}^m$ with weights $w(V)$, the Fr\'echet mean of the peaks in the Aitchison geometry on $\simplex$ (\Cref{rem:DWN-F0}).
Under the FPW equivalence of \Cref{sec:pooling-transfer}, this centroid is Genest's logarithmic opinion pool, so the characterization to follow doubles as an AMM-side derivation of Genest's result, reached by a different route.

\begin{theorem}[Characterization on $\F_\Prod$]\label{thm:characterization-F0}
Let $n \geq 3$  and $m \geq 2$.

\noindent \emph{(a) Characterization by the Arrovian core.}
A mechanism $M \colon (\simplex \times \R^A_{>0})^m \to \simplex$ satisfies the axioms (A0)--(A4) and (C) if and only if there exists a continuous, equivariant weighting function $w \colon \R^m_{>0} \to \Delta^{m-1}$ such that for every LP profile $(\alpha^*, D) \in (\simplex \times \R^A_{>0})^m$ and every $i = 1, 2, \ldots, n$,
\begin{equation}\label{eq:aitchison-centroid}
M(\alpha^*, D)_i = \frac{\prod_{\ell=1}^m (\alpha_{\ell, i}^*)^{w_\ell(V)}}{\sum_{k=1}^n \prod_{\ell=1}^m (\alpha_{\ell, k}^*)^{w_\ell(V)}},
\end{equation}
where $V = (\fref(D_1),  \ldots, \fref(D_m))$ is the valuation profile.
The output $M(\alpha^*, D)$ is the weighted Aitchison centroid of $\{\alpha_\ell^*\}_{\ell=1}^m$ with weights $w(V)$.

\noindent \emph{(b) Selection by withdrawal neutrality.}
Within the family characterized in part (a), exactly one mechanism satisfies (A5), namely the \emph{symmetric Aitchison centroid}, with constant equal weights $w_\ell(V) \equiv 1/m$ for all $V$.
Equivalently, an Arrovian mechanism satisfies (A5) if and only if its weighting $w$ does not depend on $V$, in which case equivariance forces $w \equiv (1/m, \ldots, 1/m)$.
\end{theorem}

The output of  (\ref{eq:aitchison-centroid}) lies in the \emph{open} simplex $\simplex$.
Indeed, for any weights $w(V) \in \Delta^{m-1}$ with $\sum_\ell w_\ell(V) = 1$ and $w_\ell(V) \geq 0$, each coordinate $\prod_\ell (\alpha_{\ell, i}^*)^{w_\ell(V)}$ is a product of strictly positive numbers raised to nonnegative exponents, hence strictly positive, and the normalization preserves strict positivity.

The proof of part (a) proceeds in four steps over \Cref{sec:F0-pair-extraction,sec:F0-additive-cocycle,sec:F0-cauchy,sec:F0-pareto-pin}:
\begin{enumerate}[label=(\arabic*),nosep]
\item \emph{Pairwise extraction} (\Cref{sec:F0-pair-extraction}): mechanism IIA reduces the mechanism to a family of one-dimensional pairwise rules $\psi_{ij}$, one for each asset pair, each mapping the tuple of LP $(i,j)$-preferences and valuations to the output's $(i,j)$-structural ratio; these pairwise rules determine the mechanism uniquely.
\item \emph{The cocycle constraint} (\Cref{sec:F0-additive-cocycle}): consistency of the output's structural ratios across pairs forces the logarithmic counterparts $\Psi_{ij}$ of the pairwise rules to satisfy an additive cocycle identity on $\R^m$.
\item \emph{Solving the Cauchy equation} (\Cref{sec:F0-cauchy}): the Cauchy functional equation (namely, the continuous additive case) yields a linear form $\Psi_{ij}(u; V) = \sum_\ell  w_\ell(V) u_\ell + (\beta_i(V) - \beta_j(V))$ with continuous, LP-indexed coefficient functions.
\item \emph{Pinning weights via Pareto} (\Cref{sec:F0-pareto-pin}): Pareto efficiency forces the coboundary term $\beta_i(V) - \beta_j(V)$ to vanish and the coefficients $w_\ell(V)$ to lie in the simplex $\Delta^{m-1}$, giving the Aitchison centroid form.
\end{enumerate}
The converse is verified in \Cref{sec:F0-converse};  part (b) is proved in \Cref{sec:F0-remarks}, \Cref{sec:F0-example} gives a worked example, and \Cref{rem:logpool-quickread} at the end of the section records the reading of the resulting rule as a logarithmic opinion pool.

\subsection{Pairwise Extraction}\label{sec:F0-pair-extraction}

Recall that on $\F_\Prod$,  for each $\alpha \in \simplex$ and each pair $i, j$,  we have 
\[
p_{ij}(f_\alpha, I) = \rho_{ij}(\alpha) \cdot \frac{I_j}{I_i}, \quad \rho_{ij}(\alpha) = \frac{\alpha_i}{\alpha_j},
\]
so $\rho_{ij}$ depends only on $\alpha$ and not on $I$.
The $(i,j)$-restricted preference of LP $\ell$ is therefore determined by the single number $\rho_{ij}(\alpha_\ell^*)$.

\begin{lemma}\label{lem:cocycle-rho}
For every $\alpha \in \simplex$ and $i, j, k \in A$,  the cocycle identity $\rho_{ij}(\alpha) \cdot \rho_{jk}(\alpha) = \rho_{ik}(\alpha)$ holds.
The map $\alpha \mapsto (\rho_{12}(\alpha), \rho_{13}(\alpha), \ldots, \rho_{1n}(\alpha))$ is a homeomorphism from $\simplex$ onto $\R^{n-1}_{>0}$.
\end{lemma}

\begin{proof}
The cocycle identity is direct from $\rho_{ij}(\alpha) = \alpha_i/\alpha_j$.
To prove the second claim, given $(\rho_{1j})_{j=2}^n \in \R^{n-1}_{>0}$, set $\alpha_1 := 1/\bigl(1 + \sum_{k \geq 2} \rho_{1k}^{-1}\bigr)$ and $\alpha_i := \rho_{1i}^{-1} \alpha_1$  for $i \geq 2$.
This gives $\alpha \in \simplex$ with the prescribed ratios; continuity in both directions is immediate.
\end{proof}

\begin{lemma}[Pairwise Extraction]\label{lem:pair-extraction-F0}
Assume $M$ satisfies (A1), (A3), (A4), and (C).
For each pair $i, j \in A$ with $i \neq j$,  there exists a continuous function
\[
  \psi_{ij} \colon \R_{>0}^m \times \R_{>0}^m \to \R_{>0},
\]
equivariant under joint permutation of its two $m$-tuple arguments ($\psi_{ij}(\sigma\cdot x,\, \sigma\cdot V) = \psi_{ij}(x, V)$ for every permutation $\sigma$ of $\{1, 2, \ldots, m\}$, where $\sigma\cdot x :=  (x_{\sigma(1)},  \ldots, x_{\sigma(m)})$, with the two tuples permuted jointly), such that for every profile,
\[
\rho_{ij}\bigl( M(\alpha^*, D) \bigr) = \psi_{ij}\Bigl( \bigl( \rho_{ij}(\alpha_\ell^*) \bigr)_{\ell=1}^m,\ \bigl( V_\ell \bigr)_{\ell=1}^m \Bigr).
\]
\end{lemma}

\begin{proof}
The mechanism maps into \mbox{$\F_\Prod \cong \simplex$}, so $M(\alpha^*, D) = f_{\alpha'}$ for some $\alpha' \in \simplex$, and the $(i,j)$-pricing $p_{ij}(f_{\alpha'}, I) = \rho_{ij}(\alpha') \cdot (I_j/I_i)$ makes $\rho_{ij}(M(\alpha^*, D)) = \rho_{ij}(\alpha')$ the coefficient of $I_j/I_i$, independent of $I$.
By \Cref{lem:cocycle-rho}, on $\F_\Prod$ the $(i,j)$-restricted preference of LP  $\ell$ is determined by the scalar $\rho_{ij}(\alpha_\ell^*)$.
Applying (A3), this coefficient is a function of $(\rho_{ij}(\alpha_\ell^*))_\ell$ and $V$ alone.
The domain $\R_{>0}^m \times \R_{>0}^m$ is the full positive orthant in both coordinates, since the $m$ LPs report peaks independently: any vector $(x_1,  \ldots, x_m) \in \R_{>0}^m$ is realized as the tuple $(\rho_{ij}(\alpha_\ell^*))_{\ell=1}^m$ by choosing each LP's peak $\alpha_\ell^* \in \simplex$ independently with structural ratio $\rho_{ij}(\alpha_\ell^*) = x_\ell$ (each such choice exists by \Cref{lem:cocycle-rho}); similarly for $V$.
Continuity follows from continuity of $\rho_{ij}$ and of $M$;  joint equivariance follows from (A4), since permuting LP labels by $\sigma$ permutes both the peak tuple and the deposit tuple jointly while leaving $M$ invariant.
\end{proof}

  \begin{proposition}[Pairwise Rules Determine the Mechanism]\label{prop:pairwise-rules-determine}
On $\F_\Prod$, an Arrovian mechanism $M$ is uniquely determined by its family of pairwise rules $\{\psi_{ij}\}_{i \neq j}$.
The map $M \mapsto (\psi_{ij})_{i \neq j}$ is injective, and its image consists of families of pairwise rules satisfying the cocycle identity of \Cref{lem:cocycle-psi}; that every family in the image takes the weighted-centroid form, and that every such family conversely arises from an Arrovian mechanism,  is established in \Cref{sec:F0-additive-cocycle,sec:F0-cauchy,sec:F0-pareto-pin} and verified by the explicit converse construction of \Cref{sec:F0-converse}.
\end{proposition}

\begin{proof}
Mechanism IIA defines $\psi_{ij}$ as the rule sending   the LPs' $(i,j)$-restricted preferences and valuation profile to $\rho_{ij}(M(\theta^*, D))$, so the map $M \mapsto (\psi_{ij})_{i \neq j}$ is well-defined.
For injectivity, from the structural ratios $\{\rho_{ij}(M(\theta^*, D))\}_{i \neq j}$ we recover a unique $\alpha^* \in \simplex$ by \Cref{lem:cocycle-rho}, and hence $M(\theta^*, D) = f_{\alpha^*}$.
That every family in the image is cocycle-respecting is the content of \Cref{lem:cocycle-psi}; the reverse inclusion, that each cocycle-respecting family of the form pinned down in \Cref{sec:F0-cauchy,sec:F0-pareto-pin} arises from an Arrovian mechanism, is the converse verified in \Cref{sec:F0-converse}.
This is not circular: the forward direction \emph{derives} the centroid form from the axioms without presupposing it, while the converse is an independent explicit construction.
\end{proof}

\subsection{The Cocycle Constraint}  \label{sec:F0-additive-cocycle}

Since the output lies in $\simplex$, its structural ratios satisfy the cocycle.
In terms of the $\psi_{ij}$, this becomes the following.
Throughout the remainder of the proof of \Cref{thm:characterization-F0} we write $\rho_{ij}(M)$ as shorthand for $\rho_{ij}(M(\alpha^*, D))$ when the profile is clear from context.

\begin{lemma}\label{lem:cocycle-psi}
For every $x, y \in \R_{>0}^m$ and every $V \in \R_{>0}^m$,
\begin{equation}\label{eq:cocycle-psi}
\psi_{ij}(x, V) \cdot \psi_{jk}(y, V) = \psi_{ik}(x \odot y, V),
\end{equation}
where $x \odot y := (x_1 y_1,  \ldots, x_m y_m)$ is the componentwise (Hadamard) product.
\end{lemma}

\begin{proof}
By the cocycle identity, $\rho_{ij}(\bar\alpha) \cdot  \rho_{jk}(\bar\alpha) = \rho_{ik}(\bar\alpha)$ for the output $\bar\alpha := M(\alpha^*, D)$.
Given $x_\ell, y_\ell > 0$ for $\ell = 1, 2, \ldots, m$ and three distinct indices $i, j, k$ (which exist since $n \geq 3$), we construct a profile realizing $\rho_{ij}(\alpha_\ell^*) = x_\ell$ and $\rho_{jk}(\alpha_\ell^*) = y_\ell$ as follows.
For each LP $\ell$ separately, with each $\alpha_\ell^* \in \simplex$ chosen without reference to the other LPs' peaks, use $i$ as the reference asset and set $\rho_{ij}(\alpha_\ell^*) := x_\ell$, $\rho_{ik}(\alpha_\ell^*) := x_\ell y_\ell$, and $\rho_{ia}(\alpha_\ell^*) := 1$ for every remaining asset $a \in A \setminus \{i, j, k\}$ (an empty set at $n = 3$).
Realizability of these $n-1$ positive ratios is guaranteed by \Cref{lem:cocycle-rho} (symmetric under relabeling of the reference asset): the formula gives $\alpha_{\ell,i} = (1 + \sum_{a \neq i} \rho_{ia}^{-1})^{-1}$ with all components positive, and the remaining components follow by $\alpha_{\ell,a} = \rho_{ia}^{-1} \alpha_{\ell,i}$, putting $\alpha_\ell^*$ in the open simplex.
The cocycle identity on $\alpha_\ell^*$ gives  $\rho_{jk}(\alpha_\ell^*) = \rho_{ik}(\alpha_\ell^*)/\rho_{ij}(\alpha_\ell^*) = y_\ell$.
Applying \Cref{lem:pair-extraction-F0} to this profile gives $\rho_{ij}(M) = \psi_{ij}(x, V)$ and $\rho_{jk}(M) = \psi_{jk}(y, V)$; since the mechanism output lies in $\simplex$, its structural ratios satisfy the cocycle, so $\rho_{ij}(M) \cdot \rho_{jk}(M) = \rho_{ik}(M)$.
Moreover, $\rho_{ik}(\alpha_\ell^*) = x_\ell y_\ell$ for every $\ell$, so $\rho_{ik}(M) = \psi_{ik}(x \odot y, V)$, which yields (\ref{eq:cocycle-psi}).
 \end{proof}

\begin{remark}[Why $n \geq 3$ and the Two-Asset Case]\label{rem:n-three}
The lemma uses $n \geq 3$ essentially to find three distinct indices $i, j, k$.
With only two assets, the cocycle is vacuous (there is only one pair) and (A3) does not constrain the mechanism nontrivially.
The structural rigidity that drives the characterization to follow, the cocycle/Cauchy combination, is therefore unavailable for $n = 2$.
This is the precise barrier, not a technical inconvenience but the absence of the multi-pair triangle identity that does the work in higher dimensions.
The two-asset case is treated separately in \Cref{sec:two-asset}: the Arrovian-core axioms (with (A3) silent) together with (SP) characterize the Moulin generalized medians in the log-ratio coordinate (\Cref{thm:characterization-F0-n-two}), and the impossibility does not arise.
A separate research direction asks whether some strengthening of (A3) that remains nontrivial at $n = 2$, 
of the liquidity-additivity/Thomsen type \citep{schlegel-kwasnicki-mamageishvili-2023,aczel-1966},  can recover the mean-type rigidity and so an impossibility analog; this is left open and discussed in \Cref{rem:thomsen-substitute}.
\end{remark}

The multiplicative cocycle (\ref{eq:cocycle-psi}) is the central functional equation of the section.
We pass to logarithmic coordinates.
Throughout the remainder of the proof of \Cref{thm:characterization-F0}, $V \in \R^m_{>0}$ is a continuous parameter, and within the proof of \Cref{lem:psi-decomposition} below we hold $V$ fixed and suppress it, writing $\Psi_{ij}(u)$; the post-lemma paragraph then lifts the conclusion to $V$ as a continuous side parameter.
Define
\[
\Psi_{ij} \colon \R^m \to \R, \quad \Psi_{ij}(u) := \log \psi_{ij}\bigl( e^{u_1}, e^{u_2}, \ldots, e^{u_m},\, V \bigr).
\]
Each $\Psi_{ij}$ is continuous and well-defined because $\psi_{ij}$  takes positive values (its outputs are structural ratios $\rho_{ij}$, which are positive by (T2)).
We do \emph{not} assert symmetry of $\Psi_{ij}$ in $u$ at fixed $V$: by \Cref{lem:pair-extraction-F0}, $\psi_{ij}$ is jointly equivariant under simultaneous permutation of its $u$- and $V$-tuples, which forces $u$-symmetry only when $V$ is itself permutation-symmetric.
The multiplicative cocycle becomes the additive identity
\begin{equation}\label{eq:additive-cocycle}
\Psi_{ij}(u) + \Psi_{jk}(v) = \Psi_{ik}(u + v) \quad \forall u, v \in \R^m, \ \forall\, \text{pairwise distinct } i, j, k \in A.
\end{equation}

\subsection{Solving the Additive Cocycle}\label{sec:F0-cauchy}

This is the technical heart of the characterization: the pairwise aggregators $\Psi_{ij}$ share a common $u$-dependent part $\Phi$ satisfying the Cauchy equation, with index-dependent constants in coboundary form.

\begin{lemma}\label{lem:psi-decomposition}
Assume $n \geq 3$. 
At each fixed $V$, there exists a continuous function  $\Phi \colon \R^m \to \R$ satisfying the Cauchy functional equation
\begin{equation}\label{eq:Phi-Cauchy-final}
\Phi(u + v) = \Phi(u) + \Phi(v) \qquad \forall u, v \in \R^m,
\end{equation}
and real constants $(\beta_j)_{j \in A}$, such that for every distinct $i, j \in A$ and every $u \in \R^m$,
\begin{equation}\label{eq:Psi-coboundary-form}
\Psi_{ij}(u) = \Phi(u) + \beta_i - \beta_j.
\end{equation}
\end{lemma}

\begin{proof}
We first extract a common $u$-dependent part across pairs.
Setting $v = 0$ in (\ref{eq:additive-cocycle}) yields
\begin{equation}\label{eq:cocycle-v-zero}
\Psi_{ij}(u) + \Psi_{jk}(0) = \Psi_{ik}(u) \qquad \forall u \in \R^m,
\end{equation}
and setting $u = 0$ gives the triangle identity
\begin{equation}\label{eq:triangle-identity}
\Psi_{ij}(0) + \Psi_{jk}(0) = \Psi_{ik}(0).
\end{equation}

\emph{Step~1 (independence of the second index).}
Subtracting (\ref{eq:triangle-identity}) from (\ref{eq:cocycle-v-zero}), we obtain
\[
\Psi_{ij}(u) - \Psi_{ij}(0) = \Psi_{ik}(u) - \Psi_{ik}(0)  \qquad \forall\, i, j, k \text{ pairwise distinct},\ \forall u,
\]
so the quantity $\Psi_{ij}(u) - \Psi_{ij}(0)$ is independent of the second index $j$ at fixed first index $i$ (the three-index quantifier requires $n \geq 3$).

\emph{Step~2 (independence of the first index).}
Setting $u = 0$ in (\ref{eq:additive-cocycle}) gives  $\Psi_{ij}(0) + \Psi_{jk}(v) = \Psi_{ik}(v)$, and hence
$\Psi_{jk}(v) - \Psi_{jk}(0) = \Psi_{ik}(v) - \Psi_{ik}(0)$ for all $j, k$ and pairwise distinct $i, j, k$,
so the quantity $\Psi_{jk}(v) - \Psi_{jk}(0)$ is independent of the first index $j$ at fixed second index $k$ (again the three-index quantifier requires $n \geq 3$).

We now establish that these two steps together force $\Psi_{ij}(u) - \Psi_{ij}(0)$ to be the same
function of $u$ for \emph{every} ordered pair $(i,j)$ with $i \neq j$.
Write $\Delta_{ij}(u) := \Psi_{ij}(u) - \Psi_{ij}(0)$, fix a reference pair $(i_0, j_0)$ with $i_0 \neq j_0$,
and set $\Phi(u) := \Delta_{i_0 j_0}(u)$.
We claim $\Delta_{ij} = \Phi$ for every ordered pair $(i,j)$  with $i \neq j$.

Each application of Step~1 or Step~2 requires three pairwise distinct indices, since
\Cref{lem:cocycle-psi} establishes the cocycle (\ref{eq:cocycle-psi}) only under that hypothesis.
  At $n = 3$ every triple of distinct indices exhausts $A$, so there is no slack for an auxiliary
fourth index; the argument below routes all chains through triples realizable in $A$.
We classify ordered pairs $(i,j)$ with $i \neq j$ by their overlap with the reference pair $(i_0, j_0)$.
 At $n = 3$ there are six such ordered pairs, falling into three groups: same-position overlap (Cases A, B, C, each handled by a single Step~1 or Step~2 application), disjoint (Case D,  vacuous at $n = 3$ and handled by a two-step chain at $n \geq 4$), and reverse-or-cross overlap (Cases E, F, G, each routed through the unique third element $k \in A \setminus \{i_0, j_0\}$ by a two- or three-step chain).
Each case reduces $\Delta_{ij}$ to $\Delta_{i_0 j_0} = \Phi$ by chain rewriting.

\emph{Case A ($i = i_0$ and $j = j_0$).}
Here $\Delta_{ij} = \Delta_{i_0 j_0} = \Phi$ trivially.

\emph{Case B ($i = i_0$, $j \neq j_0$).}
The three indices $i_0, j, j_0$ are pairwise distinct (since $i_0 \neq j_0$, and $i_0 = i \neq j$,  and $j \neq j_0$ by hypothesis).
Applying Step~1 to this triple with fixed first index $i_0$ gives $\Delta_{i_0 j} = \Delta_{i_0 j_0}$,
i.e., $\Delta_{ij} = \Phi$.

\emph{Case C ($j = j_0$, $i \neq i_0$).}
The three indices $i, i_0, j_0$ are pairwise distinct  (since $i \neq i_0$, and $i \neq j = j_0$, and $i_0 \neq j_0$).
Applying Step~2 to this triple with fixed second index $j_0$ gives $\Delta_{i j_0} = \Delta_{i_0 j_0}$,
i.e., $\Delta_{ij} = \Phi$.

\emph{Case D ($\{i,j\} \cap \{i_0, j_0\} = \emptyset$).}
This case is vacuous at $n = 3$, since four distinct elements cannot fit in a three-element set.
For completeness at $n \geq 4$: the case hypothesis ensures that $i, j, j_0$ are pairwise distinct and that $i, i_0, j_0$ are pairwise distinct.
Applying Step~1 to the triple $i, j, j_0$ with fixed first index $i$ gives  $\Delta_{ij} = \Delta_{i j_0}$.
Applying Step~2 to the triple $i, i_0, j_0$ with fixed second index $j_0$ gives $\Delta_{i j_0} = \Delta_{i_0 j_0} = \Phi$.
Hence,  $\Delta_{ij} = \Phi$.

\emph{Case E ($i = j_0$, $j \neq i_0$).}
We use two steps.
The three indices $j_0, i_0, j$ are pairwise distinct (since $j_0 \neq i_0$, and $j_0 = i \neq j$,
and $i_0 \neq j$ by hypothesis).
Applying Step~2 to this triple with fixed second index  $j$ gives $\Delta_{j_0 j} = \Delta_{i_0 j}$.
Next, the three indices $i_0, j, j_0$ are pairwise distinct  (since $i_0 \neq j$ as just established,
and $i_0 \neq j_0$, and $j \neq j_0$ because $j_0 = i \neq j$).
Applying Step~1 to this triple with fixed first index $i_0$ gives $\Delta_{i_0 j} = \Delta_{i_0 j_0} = \Phi$.
Both triples are subsets of $\{i_0, j_0, j\}$, which has exactly three elements and exists at $n = 3$.
Hence,  $\Delta_{ij} = \Delta_{j_0 j} = \Delta_{i_0 j} = \Phi$.

\emph{Case F ($j = i_0$, $i \neq j_0$).}
We use two steps, both on the triple $i, i_0, j_0$,  which is pairwise distinct since
$i \neq i_0$ (because $i \neq j = i_0$), $i \neq j_0$ by hypothesis, and $i_0 \neq j_0$.
This triple is a subset of $A$ at $n = 3$.
Applying Step~1 with fixed first index $i$ gives $\Delta_{i\, i_0} = \Delta_{i\, j_0}$.
Applying Step~2 to the same triple with fixed second index  $j_0$ gives $\Delta_{i\, j_0} = \Delta_{i_0 j_0} = \Phi$.
Hence, $\Delta_{ij} = \Delta_{i\, i_0} = \Delta_{i\, j_0} = \Phi$.

\emph{Case G ($i = j_0$ and $j = i_0$, i.e., $(i,j)$ is the reverse of $(i_0, j_0)$).}
Let $k$ be the unique third element of $A \setminus \{i_0, j_0\}$, which exists since $n \geq 3$.
We use three sub-applications of Steps~1 and~2 on triples drawn from $\{i_0, j_0, k\}$.

\emph{Sub-application (G.1):} apply Step~2 to the triple $(j_0, k, i_0)$ with fixed second index $i_0$, giving $\Delta_{j_0\, i_0} = \Delta_{k\, i_0}$.

\emph{Sub-application (G.2):} apply Step~1 to the triple $(k, i_0, j_0)$ with fixed first index $k$, giving $\Delta_{k\, i_0} = \Delta_{k\, j_0}$.

\emph{Sub-application (G.3):} apply Step~2 to the same triple $(k, i_0, j_0)$ with fixed second index $j_0$, giving $\Delta_{k\, j_0} = \Delta_{i_0 j_0} = \Phi$.

Observe that (G.2) and (G.3) act on the same ordered triple  $(k, i_0, j_0)$ but fix different indices (Step~1 fixes the left index, whereas Step~2 fixes the right).
These are two distinct specializations of the cocycle  (\ref{eq:additive-cocycle}) (setting $u = 0$ and $v = 0$ respectively) and involve no circular reasoning.
Chaining: $\Delta_{ij} = \Delta_{j_0\, i_0} = \Delta_{k\, i_0} = \Delta_{k\, j_0} = \Phi$.

Cases A--G are mutually exclusive and cover every ordered pair $(i,j)$ with $i \neq j$
(at $n = 3$ the six such pairs fall into cases A, B, C, E, F, G, with D vacuous).
Therefore, $\Delta_{ij}(u) = \Phi(u)$ for all such pairs, and $\Phi$ is continuous as a
difference of continuous functions.
We do not claim $\Phi$ is symmetric in $u$ at fixed  $V$.
Symmetry holds only when $V$ is itself
permutation-symmetric (by the joint equivariance of \Cref{lem:pair-extraction-F0}), and it plays no role below.

Setting $c_{ij} := \Psi_{ij}(0)$, we have the decomposition
\begin{equation}\label{eq:Psi-decomp-step}
\Psi_{ij}(u) =  \Phi(u) + c_{ij}.
\end{equation}

\emph{Second stage: the Cauchy equation and coboundary form.}
Substituting (\ref{eq:Psi-decomp-step}) into the cocycle (\ref{eq:additive-cocycle}) gives
\[
\Phi(u) + c_{ij} + \Phi(v) + c_{jk} = \Phi(u + v) + c_{ik},
\]
which rearranges to
\begin{equation}\label{eq:Phi-defect-equation}
\Phi(u + v) - \Phi(u) - \Phi(v) = c_{ij} + c_{jk} - c_{ik}.
\end{equation}
The left-hand side is independent of $(i, j, k)$  and the right-hand side is independent of $(u, v)$,
so both equal a single constant.
By the triangle identity (\ref{eq:triangle-identity}), $c_{ik} = c_{ij} + c_{jk}$,
so the right-hand side vanishes, giving the Cauchy equation (\ref{eq:Phi-Cauchy-final}).

It remains to put the constants $c_{ij}$ in coboundary form.
The constants $c_{ij}  := \Psi_{ij}(0)$ are defined only for $i \neq j$, and the triangle identity
holds only on pairwise-distinct triples.
We never use $c_{ii}$.
We first record the antisymmetry $c_{ij} = -c_{ji}$.
Since the output lies in $\simplex$, its structural ratios satisfy $\rho_{ij}(M)\,\rho_{ji}(M) = 1$
for every pair $i \neq j$ (a property of ratios, independent of the cocycle).
Because $\rho_{ji}(\alpha_\ell^*)  =  \rho_{ij}(\alpha_\ell^*)^{-1}$ for every LP $\ell$,
\Cref{lem:pair-extraction-F0} gives  \mbox{$\psi_{ij}(x, V)\,\psi_{ji}(x^{\odot -1}, V) = 1$}
for all $x \in \R^m_{>0}$, i.e., in logarithmic coordinates \mbox{$\Psi_{ij}(u) + \Psi_{ji}(-u) = 0$}
for all $u \in \R^m$. 
Setting $u = 0$ yields $c_{ij} + c_{ji} = 0$.

Now fix a reference index $i_0 \in A$ and set  $\beta_i := c_{i, i_0}$ for $i \neq i_0$
and $\beta_{i_0} := 0$.
For distinct $i, j$ both different from $i_0$, the triangle identity applied to the pairwise-distinct triple $(i, i_0, j)$ gives
\[
c_{ij} = c_{i, i_0} + c_{i_0, j} =  c_{i, i_0} - c_{j, i_0} =  \beta_i - \beta_j,
\]
using antisymmetry in the middle step.
For $j = i_0$, we have $c_{i, i_0} = \beta_i = \beta_i - \beta_{i_0}$, 
while for  $i = i_0$ we have $c_{i_0, j} = -c_{j, i_0} = -\beta_j = \beta_{i_0} - \beta_j$.
In every case, $c_{ij} = \beta_i - \beta_j$. 
Substituting finally into (\ref{eq:Psi-decomp-step}) yields (\ref{eq:Psi-coboundary-form}).
\end{proof}

\Cref{lem:psi-decomposition} was stated at fixed $V$ for clarity, 
but the identical argument goes through with $V \in \R^m_{>0}$ as a continuous side parameter, 
since $\Psi_{ij}(u; V)$, $\Phi(u; V)$, and $c_{ij}(V) := \Psi_{ij}(0; V)$ 
all depend continuously on $V$ and the coboundary derivation (\ref{eq:Psi-coboundary-form}) proceeds verbatim.
Continuity of the Cauchy solution then gives the following.

\begin{lemma}\label{lem:Phi-linear}
For each fixed $V \in \R^m_{>0}$,  the function $\Phi(\,\cdot\,; V)$ is linear in $u$: there exist coefficients $a_\ell(V) \in \R$ such that $\Phi(u; V) = \sum_\ell a_\ell(V) \, u_\ell$ for every $u \in \R^m$.
The coefficient map \mbox{$a \colon \R^m_{>0} \to \R^m$} defined by  $V \mapsto a(V) = (a_1(V), a_2(V), \ldots, a_m(V))$ is continuous.
\end{lemma}

\begin{proof}
For each fixed $V$, $\Phi(\,\cdot\,; V)$ is continuous in $u$ and satisfies $\Phi(u + v; V) = \Phi(u; V) + \Phi(v; V)$.
\Citet[Chapter 2, Theorem 1]{aczel-1966} characterizes continuous solutions of the one-dimensional Cauchy equation as linear; the extension to $\R^m$ is standard, by applying the one-dimensional result to each coordinate slice  $u \mapsto \Phi(t e_\ell; V)$ to obtain $\Phi(t e_\ell; V) = a_\ell(V) t$ for some $a_\ell(V) \in \R$, and then assembling via additivity:   $\Phi(u; V) = \Phi(\sum_\ell u_\ell e_\ell; V) = \sum_\ell \Phi(u_\ell e_\ell; V) = \sum_\ell a_\ell(V) u_\ell$.
We do not assert symmetry of $\Phi$ in $u$ at fixed $V$, which holds only when $V$ is itself permutation-symmetric.
The linearity above needs no such symmetry, and the relabeling structure of the coefficients $a_\ell(V)$ is recovered in \Cref{lem:V-dependent-weights} as a joint equivariance,  $a_{\sigma(\ell)}(V) = a_\ell(\sigma\cdot V)$, directly from (A4).
Continuity of $V \mapsto a(V)$ follows from joint continuity of $\Psi_{ij}$ in $(u, V)$, which in turn follows from continuity of $M$ in all arguments (axiom (C) requires joint continuity in the product topology on profiles).
\end{proof}

\subsection{Pinning Down the Weights via Pareto}\label{sec:F0-pareto-pin}

Combining \Cref{lem:psi-decomposition,lem:Phi-linear}, and writing  $c_{ij}(V) = \beta_i(V) - \beta_j(V)$ in coboundary form, we have for every $u \in \R^m$ and every $V \in \R^m_{>0}$:

\noindent\emph{A notational caution.}
The vector $u \in \R^m$ here is the input to $\Psi_{ij}$, representing the tuple $(\log \rho_{ij}(\alpha_\ell^*))_{\ell=1}^m$ of \emph{LP-indexed} log-structural-ratios, one coordinate per LP.
This is distinct from the zero-sum hyperplane $H \subset \R^n$ of \emph{asset}-indexed clr-vectors.
In particular, $u$ ranges freely over all of $\R^m$ (any vector is realizable as a tuple of LP log-ratios by choosing each $\alpha_\ell^*$ independently), and there is no zero-sum constraint on $u$; by contrast, $\clr(\alpha) \in H \subset \R^n$ always satisfies $\sum_i \clr(\alpha)_i = 0$.
The test vectors $u = te_k \in \R^m$ and $u = (t,  \ldots, t) \in \R^m$ used in Step~2 of \Cref{lem:V-dependent-weights} below are vectors in this LP-indexed space, not in $H$.
\begin{equation}\label{eq:Psi-w}
\Psi_{ij}(u; V) = \sum_{\ell=1}^m a_\ell(V) \, u_\ell + \bigl(\beta_i(V) - \beta_j(V)\bigr),
\end{equation}
where $a_\ell, \beta_i \colon \R^m_{>0}  \to \R$ are all continuous.
$\beta_i(V)$ is determined by $c_{ij}(V) = \Psi_{ij}(0; V)$ and may in principle depend on $V$ at this stage.
We now invoke (A2), (A4), and (A0) to pin down $a(V)$ and show that the coboundary term vanishes identically.

\begin{lemma}\label{lem:pareto-hull}
Under (A2), the output $M(\alpha^*, D)$ lies in the Aitchison hull 
 $C := \mathrm{Aitch\text{-}hull}\{\alpha_\ell^*\}_{\ell=1}^m$ 
of the peaks.
Equivalently, $\clr(M(\alpha^*, D))$ lies in the Euclidean convex hull of $\{\clr(\alpha_\ell^*)\}_{\ell=1}^m$ in $H$.
In particular, for each pair $i, j$,
\[
\log \rho_{ij}\bigl( M(\alpha^*, D) \bigr) \in \Bigl[ \min_\ell \log \rho_{ij}(\alpha_\ell^*),\ \max_\ell \log \rho_{ij}(\alpha_\ell^*) \Bigr].
\]
Moreover, the Aitchison hull $C$ equals the Pareto-undominated set $P(\theta^*)$ for Aitchison-single-peaked preferences.
\end{lemma}

\begin{proof}
We prove the biconditional $P(\theta^*) = C$.

\emph{Direction 1 ($C^c \subseteq P(\theta^*)^c$, i.e., outside the hull implies dominated).}
Suppose $\alpha \notin C$.
Since $C$ is a closed convex set in $H$ (under clr), the nearest-point projection $\pi(\clr(\alpha))$ onto $C$ exists uniquely by the Hilbert projection theorem.
Let $\alpha^{\mathrm{proj}} := \clr^{-1}(\pi(\clr(\alpha)))$.
Write \mbox{$a := \clr(\alpha) - \pi(\clr(\alpha))$} and \mbox{$b_\ell := \pi(\clr(\alpha)) - \clr(\alpha_\ell^*)$}.
The variational inequality for the nearest-point projection onto $C$, applied with $c = \clr(\alpha_\ell^*) \in C$, 
gives $\langle a, b_\ell\rangle \geq 0$, and $\|a\|^2 > 0$  since $\clr(\alpha) \notin C$.
Therefore, 
\[
\|\clr(\alpha) - \clr(\alpha_\ell^*)\|^2 = \|a + b_\ell\|^2 = \|a\|^2 + 2\langle a, b_\ell\rangle + \|b_\ell\|^2 > \|b_\ell\|^2,
\]
that is, $\Aitch(\alpha^{\mathrm{proj}}, \alpha_\ell^*) < \Aitch(\alpha, \alpha_\ell^*)$ for every $\ell$, meaning $\alpha^{\mathrm{proj}} \succ_\ell \alpha$ for every $\ell$.
Hence, $\alpha \notin P(\theta^*)$.

\emph{Direction 2 ($C \subseteq P(\theta^*)$, i.e., inside the hull implies undominated).}
Suppose $\alpha \in C$, so $\clr(\alpha) = \sum_\ell \lambda_\ell \clr(\alpha_\ell^*)$ for some $(\lambda_\ell) \in \Delta^{m-1}$.
We claim no $\alpha' \in \simplex$ satisfies $\alpha'  \succ_\ell \alpha$ for every $\ell$.
Suppose by contradiction that such $\alpha'$ exists.
Then $\|\clr(\alpha') - \clr(\alpha_\ell^*)\| < \|\clr(\alpha) - \clr(\alpha_\ell^*)\|$ for every $\ell$.
Let $x := \clr(\alpha') - \clr(\alpha)$.
Expanding,  $\|x + (\clr(\alpha) - \clr(\alpha_\ell^*))\|^2  <  \|\clr(\alpha) - \clr(\alpha_\ell^*)\|^2$, which gives $\|x\|^2 + 2\langle x,\, \clr(\alpha) - \clr(\alpha_\ell^*)\rangle < 0$ for every $\ell$.
 Multiplying the strict inequality at each $\ell$ by $\lambda_\ell \geq 0$ and summing over $\ell$ yields a weighted sum; since $\sum_\ell \lambda_\ell = 1$ there is at least one $\ell$ with $\lambda_\ell > 0$, and every such $\ell$ contributes a strictly negative term $\lambda_\ell \cdot (\text{negative quantity})$, making the total strictly negative:
\[
\|x\|^2 + 2\Bigl\langle x,\, \clr(\alpha) - \sum_\ell  \lambda_\ell \clr(\alpha_\ell^*) \Bigr\rangle < 0.
\]
Since $\clr(\alpha) = \sum_\ell \lambda_\ell \clr(\alpha_\ell^*)$, the inner product vanishes, leaving $\|x\|^2 < 0$. 
Contradiction.
Therefore,  no such $\alpha'$ exists, and $\alpha \in P(\theta^*)$.

The pairwise bound follows because $\log\rho_{ij}$ is a linear functional on $H$, so its range over $C$ is $[\min_\ell \log\rho_{ij}(\alpha_\ell^*),\, \max_\ell \log\rho_{ij}(\alpha_\ell^*)]$.
\end{proof}

\begin{lemma}\label{lem:V-dependent-weights}
Under (A0)--(A4) and (C), there exists a continuous function $w \colon \R^m_{>0} \to \Delta^{m-1}$, equivariant under the diagonal action of the symmetric group
 (i.e., $w_{\sigma(\ell)}(V) = w_\ell(\sigma \cdot V)$ for every permutation $\sigma$), such that
\[
\Psi_{ij}(u; V) = \sum_{\ell=1}^m w_\ell(V)  \cdot u_\ell \qquad \forall u \in \R^m,\ \forall V \in \R^m_{>0}.
\]
\end{lemma}

\begin{proof}
We establish in three steps that  $a_\ell(V)  =:  w_\ell(V)$ are valid weights and that the coboundary term vanishes.
A notational point: the constants $\beta_i$ introduced in \Cref{lem:psi-decomposition} were defined for each fixed $V$ separately, so they may depend on $V$.
We write $\beta_i(V)$ from this point on to make the $V$-dependence explicit, with the understanding that at any fixed $V$ the $\beta_i(V)$ are real constants in $u$.

\emph{Step 1: the coboundary vanishes.}
For each fixed pair $(i, j)$, the map $\alpha \mapsto \rho_{ij}(\alpha) = \alpha_i/\alpha_j$ is surjective from $\simplex$ onto $\R_{>0}$, as a consequence of \Cref{lem:cocycle-rho} after relabeling assets so that the active pair appears among the first-row ratios characterized there (the lemma states a homeomorphism via the ratios with a designated reference asset, and any pair can be made the reference by relabeling;  surjectivity then follows because the homeomorphism maps onto all of $\R_{>0}^{n-1}$).
Hence, any $u \in \R^m$ arises as $(\log \rho_{ij}(\alpha_\ell^*))_{\ell=1}^m$ for some profile $(\alpha_1^*,  \ldots, \alpha_m^*)$, by choosing each $\alpha_\ell^*$ independently to realize its prescribed $\rho_{ij}$-value.
The Pareto bound
\begin{equation}\label{eq:pareto-bound}
\Psi_{ij}(u; V) \in [\min_\ell u_\ell,\, \max_\ell  u_\ell]
\end{equation}
holds for every $u \in \R^m$ as an immediate consequence of \Cref{lem:pareto-hull} (the output lies in the Aitchison hull of the peaks).
Setting $u = 0$ in (\ref{eq:Psi-w}) gives  $\Psi_{ij}(0; V) = \beta_i(V) - \beta_j(V)$. 
The case \mbox{$u = 0$} is realized, at any fixed \mbox{$V \in \R^m_{>0}$}, by the profile in which every LP has \mbox{$\rho_{ij}(\alpha_\ell^*) = 1$} (existence guaranteed by the surjectivity just established, or by the explicit choice \mbox{$\alpha_\ell^* = (1/n,  \ldots, 1/n)$} for every $\ell$); the deposits $D$ are independent inputs and can be chosen freely to realize the given $V$ (each $V_\ell = \fref(D_\ell)$ can be made any prescribed positive value by scaling $D_\ell$, since $\fref$ is $1$-homogeneous and continuous).
The Pareto bound thus applies at  $u = 0$ for this $V$:   it says $\Psi_{ij}(0; V) \in [\min_\ell 0,\, \max_\ell 0] = \{0\}$, which forces $\Psi_{ij}(0; V) = 0$, hence $\beta_i(V) - \beta_j(V) = 0$ for every pair $i \neq j$ and every $V \in \R^m_{>0}$.
So the coboundary vanishes and (\ref{eq:Psi-w}) simplifies to $\Psi_{ij}(u; V) = \sum_\ell a_\ell(V) \, u_\ell$.

\emph{Step 2: the coefficients are nonnegative and sum to one.}
The Pareto bound (\ref{eq:pareto-bound}) now reads $\sum_\ell  a_\ell(V) \, u_\ell \in [\min_\ell u_\ell,\, \max_\ell u_\ell]$ for all $u$.
Setting $u = t \, e_k$ for $t > 0$ gives $t a_k(V) \in [0, t]$, hence $a_k(V) \geq 0$.
Setting $u = (t,  \ldots, t)$ gives $ (\sum_\ell a_\ell(V)) \, t = t$. 
Therefore,  it holds that $\sum_\ell a_\ell(V) = 1$.
Since all $a_k(V) \geq 0$ and $\sum_k a_k(V) = 1$, each $a_k(V) \leq 1$ as well, and together they satisfy $a(V) \in \Delta^{m-1}$.
We write $w_\ell(V) := a_\ell(V)$.

\emph{Step 3: equivariance from (A4).}
Axiom (A4) requires that simultaneously permuting LP labels and deposits leaves the mechanism unchanged.
In terms of $\Psi_{ij}$, applying $\sigma$ to the LP labels of the inputs sends $\Psi_{ij}(u; V) = \sum_\ell w_\ell(V)\, u_\ell$ to $\Psi_{ij}(\sigma \cdot u;\, \sigma \cdot V) = \sum_\ell w_\ell(\sigma \cdot V)\, u_{\sigma(\ell)}$, which (A4) requires equal $\Psi_{ij}(u; V)$ for every $u$ and $V$.
Reindexing the right-hand sum via $\ell' = \sigma(\ell)$ gives $\sum_{\ell'} w_{\sigma^{-1}(\ell')}(\sigma \cdot V)\, u_{\ell'}$, so $\sum_\ell w_\ell(V)\, u_\ell = \sum_\ell w_{\sigma^{-1}(\ell)}(\sigma \cdot V)\, u_\ell$ for every $u \in \R^m$.
Since the $u_\ell$ are independent and the identity holds for all $u$, the coefficients match termwise: $w_\ell(V) = w_{\sigma^{-1}(\ell)}(\sigma \cdot V)$, which on substituting $\ell \mapsto \sigma(\ell)$ becomes $w_{\sigma(\ell)}(V) = w_\ell(\sigma \cdot V)$.
This is joint equivariance,  which forces equal  weights only at symmetric $V$ (where it gives $w_\ell = 1/m$), and permits asymmetric weights at asymmetric valuations.

We do not claim that the weights are strictly positive.
A continuous equivariant $w$ may vanish for some $\ell$ on part of $\R^m_{>0}$ without violating any of (A0)--(A4) or (C),
  since constancy of the mechanism on a single slice of profiles does not contradict the \emph{global} nonconstancy that (A0) requires, which is in any case secured by the uniform weight $(1/m,  \ldots, 1/m)$ that equivariance forces at any symmetric $V$.
What the impossibility argument of \Cref{sec:impossibility} uses is precisely this uniform value $w_\ell(V^{\mathrm{sym}}) = 1/m$ at a symmetric valuation profile, not positivity everywhere.
\end{proof}

We have therefore shown that under  (A0)--(A4)  and (C), the mechanism's pairwise rules are
\[
\rho_{ij}\bigl( M(\alpha^*, D) \bigr)  =  \prod_{\ell=1}^m \rho_{ij}(\alpha_\ell^*)^{w_\ell( V)},
\]
and by \Cref{lem:cocycle-rho} this determines $M(\alpha^*, D) \in \simplex$ uniquely.
Taking logarithms and translating to clr-coordinates yields
\[
\clr\bigl( M(\alpha^*, D) \bigr) =  \sum_{\ell=1}^m w_\ell(V) \, \clr(\alpha_\ell^*) \in H,
\]
which on inversion gives precisely  (\ref{eq:aitchison-centroid}). 

\subsection{The Converse}\label{sec:F0-converse}

The canonicality of the geometric mean as anchor for $\fref$ was established in \Cref{sec:lps} after \Cref{def:valuation}.
\Cref{rem:valuation-robustness} below records the precise sense in which the characterization is robust to alternative choices.

We verify that for any continuous, equivariant  $w \colon \R^m_{>0} \to \Delta^{m-1}$, the mechanism (\ref{eq:aitchison-centroid}) satisfies (A0)--(A4) and (C).
(A1) is immediate.
(A0): although the weighting $w$ may a priori vanish on some coordinates (its codomain is the closed simplex $\Delta^{m-1}$), nonconstancy is secured at the symmetric valuation, as follows.
At any symmetric deposit profile (e.g., $D_1  = \cdots = D_m$), the valuation profile is symmetric, $V^{\mathrm{sym}} = (v,  \ldots, v)$ for some $v > 0$, and equivariance of $w$ (in the sense of the definition preceding the theorem) forces $w_\ell(V^{\mathrm{sym}}) = 1/m$ for every $\ell$, since each transposition $\sigma$ fixes $V^{\mathrm{sym}}$ and therefore must permute the components of $w(V^{\mathrm{sym}})$ trivially.
With every weight strictly $1/m > 0$ at $V^{\mathrm{sym}}$, varying any single peak $\alpha_\ell^*$ strictly changes the centroid output (since the geometric mean is injective in each factor), so the mechanism is nonconstant.
Global vanishing of any $w_\ell$ would contradict this, so the equivariance forces strict positivity on the symmetric slice regardless of whether $w$ vanishes elsewhere.
(A2): writing $\bar\alpha := M(\alpha^*, D)$, we have \mbox{$\clr(\bar\alpha) = \sum_\ell w_\ell(V) \clr(\alpha_\ell^*)$} with weights $w_\ell(V) \geq 0$ summing to one, so $\bar\alpha$ lies in the Aitchison hull $C := \mathrm{Aitch\text{-}hull}\{\alpha_\ell^*\}_{\ell=1}^m \subseteq \simplex$; by Direction 2 of \Cref{lem:pareto-hull}, $C \subseteq P(\theta^*)$, hence $\bar\alpha \in P(\theta^*)$.
(A3): \mbox{$\rho_{ij}(M) = \prod_\ell \rho_{ij}(\alpha_\ell^*)^{w_\ell(V)}$} uses only $(i,j)$-restrictions and $V$; (A4) follows from equivariance of $w$, and (C) from continuity of $w$ and of the geometric mean.
This completes the proof of \Cref{thm:characterization-F0}.

\begin{remark}[Robustness to Valuation Choice]\label{rem:valuation-robustness}
The characterization of \Cref{thm:characterization-F0} depends on $\fref$ in a precise and limited way, which we now isolate.
The \emph{output form} 
(namely, the weighted Aitchison centroid) is entirely independent of the choice of reference:
 indeed, any continuous, $1$-homogeneous, asset-symmetric function $\tilde f$ used in place of $\fref$ yields 
precisely the same centroid formula (\ref{eq:aitchison-centroid}), by the identical Cauchy-cocycle argument.
What does depend on the reference is the \emph{parameterization} of the weighting function: writing $\tilde V_\ell := \tilde f(D_\ell)$, the admissible weightings are continuous equivariant functions of $\tilde V$ rather than of $V = (\fref(D_\ell))_\ell$, and these two families are not identified with each other,  since $\tilde f(D_\ell)$ is not a function of $\fref(D_\ell)$ alone.
Three consequences of the characterization are reference-independent: the centroid form (\ref{eq:aitchison-centroid}), the FPW correspondence with opinion pooling (\Cref{prop:correspondence}), and the impossibility (\Cref{thm:impossibility}).
The last of these holds because the manipulation of \Cref{lem:sp-forces-projection} requires only that the weights be uniform at some symmetric valuation profile, and at any symmetric deposit profile every admissible reference function returns a symmetric valuation profile, forcing uniformity by equivariance.
\end{remark}

\subsection{The Canonical Weighting}\label{sec:F0-remarks}

 Within the Arrovian family, the weighting function $w$ is unconstrained beyond continuity and equivariance, leaving open a wide range of admissible rules parameterized by how deposits translate into influence.
Withdrawal neutrality (A5) of \Cref{sec:axioms-WN}, which asks that an LP's partial withdrawal not change the trading function the pool offers, singles out a canonical choice: the deposit-blind symmetric centroid.
This is the content of \Cref{thm:characterization-F0}(b), which we now prove.

\begin{remark}[Withdrawal Neutrality on $\F_\Prod$]\label{rem:DWN-F0}
The output $\bar\alpha = M(\alpha^*, D)$ is the weighted \emph{Fr\'echet mean} of the peaks in the Aitchison geometry, the minimizer of $\sum_\ell w_\ell(V) \Aitch(\alpha, \alpha_\ell^*)^2$,  since $\clr$ is an isometry onto the flat $H$ and the Fr\'echet mean there is the arithmetic mean.
Such centroids are central in compositional data analysis and information geometry \citep{aitchison-1986, pawlowsky-glahn-egozcue-tolosana-2015, amari-nagaoka-2000}.
\end{remark}

\begin{proof}[Proof of \Cref{thm:characterization-F0}(b)]
By part (a), we know that  every Arrovian mechanism has the form (\ref{eq:aitchison-centroid}) for a continuous equivariant $w \colon \R^m_{>0} \to \Delta^{m-1}$, with $\clr(M(\alpha^*, D)) = \sum_\ell w_\ell(V)\, \clr(\alpha_\ell^*)$.

\emph{Step 1: (A5) is equivalent to coefficientwise invariance under single-coordinate scaling.}
On $\F_\Prod$, at a fixed inventory ray $[I]$ the price is $p_{ij}(f_\beta, [I]) = \rho_{ij}(\beta)\,(I_j/I_i)$, with the factor $I_j/I_i$ well-defined on the ray; so two outputs have equal prices at $[I]$ for every pair $(i,j)$ if and only if their structural ratios agree, which by \Cref{lem:cocycle-rho} holds if and only if the outputs coincide.
Hence, (A5) is equivalent to $M(\theta^*, D') = M(\theta^*, D)$ for every profile, i.e., 
\begin{equation}\label{eq:DWN-coeff}
\sum_\ell w_\ell(V')\, \clr(\alpha_\ell^*) =  \sum_\ell w_\ell(V)\, \clr(\alpha_\ell^*) \qquad \text{for every peak profile } \alpha^*,
\end{equation}
where $V'$ scales the single coordinate $V_\ell  \mapsto (1-\lambda) V_\ell$ and fixes the rest.
Since $\clr \colon \simplex \to H$ is onto and $\dim H = n-1 \geq 1$ (the argument formally uses $n \geq 2$, which is implicit throughout but worth recording at this step), we may fix any index $k$, choose a profile with $\clr(\alpha_k^*) = z \neq 0$ and $\clr(\alpha_{k'}^*) = 0$ for $k' \neq k$, and read off from (\ref{eq:DWN-coeff}) that $w_k(V')\,z = w_k(V)\,z$, thus $w_k(V') = w_k(V)$.
Thus, (A5) holds if and only if $w$ is invariant under scaling any single coordinate of $V$ by any factor $t = 1 - \lambda \in (0, 1]$.

\emph{Step 2: single-coordinate invariance forces $w$ constant.}
Let $V, \tilde V \in \R^m_{>0}$ be arbitrary and set $W := (\min(V_k, \tilde V_k))_k \in \R^m_{>0}$.
By construction, $W_k \leq V_k$ for every $k$, so $W_k/V_k \in (0,1]$.
The $k$-th coordinate of $V$ can therefore be scaled \emph{down} to $W_k$ by a valid single-coordinate move (factor $W_k/V_k \in (0,1]$).
Carrying out these moves one coordinate at a time (the case $W_k = V_k$ is the identity step, requiring no appeal to (A5); every intermediate point lies in $\R^m_{>0}$ since each coordinate remains strictly positive) and applying Step~1 at each nontrivial step gives $w(V) = w(W)$.
By the same token, $W_k \leq \tilde V_k$ for every $k$, so the same one-coordinate scale-down argument gives $w(\tilde V) = w(W)$.
Hence, $w(V) = w(W) = w(\tilde V)$, and since $V, \tilde V \in \R^m_{>0}$ were arbitrary, $w$ is constant on $\R^m_{>0}$.
Conversely, a $V$-independent $w$ makes the output independent of $D$, so (A5) holds.

\emph{Step 3: a constant equivariant weighting is uniform.}
If $w \equiv w^0$ is constant, equivariance  ($w_{\sigma(\ell)}(V) = w_\ell(\sigma \cdot V)$)  reads $w^0_{\sigma(\ell)} = w^0_\ell$ for every permutation $\sigma$. 
Taking $\sigma$ to transpose $\ell$ and $\ell'$ gives $w^0_\ell = w^0_{\ell'}$, 
 so all components are equal and $w^0 = (1/m,  \ldots, 1/m)$.
\end{proof}

\begin{remark}[Withdrawal Neutrality Versus Stake Weighting] \label{rem:DWN-tension}
\Cref{thm:characterization-F0}(b) exposes a genuine tension.
Withdrawal neutrality, in its pure single-LP form, is incompatible with  \emph{any} dependence of the trading function on deposit sizes: the only withdrawal-neutral fair rule is the deposit-blind, one-LP-one-vote symmetric centroid.
A designer who instead wants influence to scale with stake, the economically natural reading of LPs as residual claimants, must give up withdrawal neutrality, since reducing one's deposit then reduces one's weight and so moves the pool's exchange rates.
The stake-proportional rule $w_\ell(V) =  V_\ell/V_{\mathrm{tot}}$ is the leading such choice and is the running representative whenever an explicit weighting is needed in the sequel and in \Cref{app:independence}.
It is homogeneity-zero, hence invariant under a common rescaling $D \mapsto \lambda D$, but is not withdrawal-neutral since a single-LP withdrawal changes the ratios $V_\ell/V_{\mathrm{tot}}$.
We emphasize that stake-proportionality is one specific point in the family characterized by \Cref{thm:characterization-F0}(a), used here to fix ideas.
The impossibility of \Cref{thm:impossibility} applies to the \emph{entire} family of continuous equivariant weightings, since the proof goes through the uniform weight $(1/m,  \ldots, 1/m)$  forced by equivariance at any symmetric valuation, not through any feature specific to stake-proportionality.
\end{remark}

\subsection{A Worked Example}\label{sec:F0-example}

A small example illustrates the characterization.
Take $n = 3$ assets, say ETH, USDC, and WBTC, and  $m = 2$ LPs, with LP 1 preferring a Uniswap-style symmetric pool $\alpha_1^* = (1/3, 1/3, 1/3)$ and LP 2 a USDC-heavy pool $\alpha_2^* = (1/5, 3/5, 1/5)$.
Suppose they deposit equally, so $V_1 = V_2$ and the deposit-share weighting gives $w_1 = w_2 = 1/2$.

By \Cref{thm:characterization-F0}, the mechanism's output $\bar\alpha := M(\alpha^*, D)$ satisfies
\[
\bar\alpha_i \propto (\alpha_{1, i}^*)^{1/2} (\alpha_{2, i}^*)^{1/2}, \quad i = 1, 2, 3.
\]
The unnormalized coordinates are $(\alpha_{1,i}^*\alpha_{2,i}^*)^{1/2}$, namely 
we have  $\sqrt{\tfrac13\cdot\tfrac15} = \sqrt{1/15}$ for ETH, $\sqrt{\tfrac13\cdot\tfrac35} = \sqrt{1/5}$ for USDC, and $\sqrt{1/15}$ for WBTC.
The normalizing denominator is their sum;  writing $\sqrt{1/5} = \sqrt{3/15} = \sqrt{3}/\sqrt{15}$ and $2\sqrt{1/15} = 2/\sqrt{15}$,
\[
 2 \sqrt{\tfrac{1}{15}} + \sqrt{\tfrac15} = \frac{2}{\sqrt{15}} + \frac{\sqrt 3}{\sqrt{15}} = \frac{2 + \sqrt{3}}{\sqrt{15}}.
\]
Dividing each numerator by this denominator (the $\sqrt{15}$ cancels),
\[
\bar\alpha = \left( \tfrac{1}{2 + \sqrt{3}},\ \tfrac{\sqrt{3}}{2 + \sqrt{3}},\ \tfrac{1}{2 + \sqrt{3}} \right) \approx (0.268,\, 0.464,\, 0.268).
\]
The output skews toward USDC in the direction of LP 2's preference, but remains symmetric between ETH and WBTC where the LPs agree.

For comparison, the arithmetic mean of the peaks is $\approx (0.267, 0.467, 0.267)$, close to the Aitchison centroid here because the peaks are close in Aitchison distance.
For nearby points the Aitchison centroid approximates the arithmetic mean.
When peaks are far apart the two diverge sharply.
Consider the profile $\alpha_1^* = (0.49, 0.49, 0.02)$ and $\alpha_2^* = (0.02, 0.49, 0.49)$ with equal weights $w_1 = w_2 = 1/2$.
For ETH, USDC, and WBTC, the 
respective unnormalized centroid coordinates are
 \[
\sqrt{0.49 \cdot 0.02} = \sqrt{0.0098} \approx 0.0990, \quad \sqrt{0.49 \cdot 0.49} = 0.4900, \quad \sqrt{0.02 \cdot 0.49} = \sqrt{0.0098} \approx 0.0990.
 \]
 Their sum is $2 \times 0.0990 + 0.4900 \approx 0.6880$, and normalizing yields
\[
\bar\alpha \approx (0.144,\ 0.712,\ 0.144),
\]
a collapse toward USDC, versus the arithmetic mean  $(0.255, 0.490, 0.255)$. 
Averaging in log-ratio coordinates, the centroid penalizes disagreement on small-weight assets more sharply: 
a peak that nearly zeroes an asset pulls the geometric mean toward zero there far more than it moves the arithmetic mean.

\begin{remark}[Reading as a Logarithmic Opinion Pool]\label{rem:logpool-quickread}
Identifying the element $\alpha \in \simplex$ with a probability distribution on the asset set $A$, the formula (\ref{eq:aitchison-centroid}) reads $\bar\alpha_i \propto \prod_{\ell=1}^m (\alpha_{\ell, i}^*)^{w_\ell(V)}$.
This is the \emph{logarithmic opinion pool}, the normalized weighted geometric mean of the LP-distributions, characterized by \citet{genest-1984} as the unique externally Bayesian pool.
Its arithmetic counterpart, the \emph{linear} pool of \citet{mcconway-1981}, is characterized instead by marginalization-consistency, the pooling-side analog of the Acz\'el--Wagner additive-aggregation condition.
The FPW equivalence \citep{frongillo-papireddygari-waggoner-2024} identifies each $f_\alpha \in \F_\Prod$ with a prediction market whose implicit belief is $\alpha$ (the Chen--Pennock log-utility reading, \citealp{chen-pennock-2007}), so that an LP's preferred pool corresponds to that LP's subjective distribution over the $n$ assets.
  The correspondence between the two frameworks is then exact, axiom by axiom:
mechanism IIA $\leftrightarrow$ external Bayesianity,  Pareto $\leftrightarrow$ unanimity, anonymity $\leftrightarrow$ symmetry, continuity $\leftrightarrow$ continuity (\Cref{prop:correspondence}), and the impossibility of \Cref{thm:impossibility} transfers, in \Cref{sec:pooling-transfer}, to a manipulability theorem for externally Bayesian opinion pools.
\end{remark}


\section{The Impossibility of Fair Strategy-Proof Aggregation}\label{sec:impossibility}

On the weighted-product family $\F_\Prod$, the Arrovian core has been characterized as the weighted Aitchison centroid (\Cref{thm:characterization-F0}). 
Among these mechanisms, which are strategy-proof?
The answer is sharply negative: within the characterized centroid family, only the single-LP dictators are strategy-proof, and dictators fail anonymity for $m \geq 2$.
Preferences are Aitchison-single-peaked as set up in  \Cref{sec:lps}.
The Cauchy and strategy-proofness steps are preference-shape independent, and the impossibility extends to any preferences with closed convex Pareto sets in clr-coordinates (\Cref{rem:beyond-single-peaked}).

The idea of the proof is a collision of two geometries on the same space.
The Arrovian core forces the fair rule to be \emph{mean-type}, a weighted average of the LP peaks in log-ratio coordinates (\Cref{thm:characterization-F0}).
Strategy-proofness in the induced Euclidean geometry forces it to be \emph{median-type}, a selection of one LP's report, since any average that genuinely mixes two or more reports can be driven to a manipulator's peak by a single misreport.
A rule cannot be both, so on $\F_\Prod$ none is at once fair and strategy-proof. 
 The collision is sharp rather than incremental, in that no quantitative weakening of either axiom dissolves it: every nondictatorial weight vector is manipulable, with manipulation strict whenever at least two LPs carry nonzero weight.
The two steps are isolated in \Cref{lem:sp-forces-projection} and assembled in \Cref{sec:impossibility-proof}.

\subsection{Statement of the Impossibility}\label{sec:impossibility-statement}

\begin{theorem}[Impossibility Theorem on $\F_\Prod$: Arrovian Fairness and Strategy-Proofness Are Incompatible]\label{thm:impossibility}
Let $n \geq 3$ and $m \geq 2$, and identify \mbox{$\F_\Prod \cong \simplex$} as in \Cref{def:F0}.
No mechanism $M \colon (\simplex \times \R^A_{>0})^m  \to \simplex$ on $\F_\Prod$ satisfies all of (A0)--(A4), (C), and (SP) simultaneously.
\end{theorem}

The proof proceeds in two steps: the Arrovian core forces $M$ to be a weighted Aitchison centroid (\Cref{thm:characterization-F0}), and a direct manipulation argument shows that on the Aitchison metric, a weighted Aitchison centroid is strategy-proof only if it is a single-LP dictator, which violates (A4) for  $m \geq 2$.

\subsection{Strategy-Proofness Forces a Dictator}\label{sec:sp-forces-dictator}

We work in clr-coordinates, where Aitchison-single-peaked  preferences on $\simplex$ become Euclidean-single-peaked preferences on the hyperplane $H = \{x \in \R^n \mid \sum_i x_i = 0\}$.
By \Cref{thm:characterization-F0}, any Arrovian mechanism on $\F_\Prod$ acts on clr-coordinates as
\[
\clr\bigl( M(\alpha^*, D) \bigr) = \sum_{\ell=1}^m w_\ell(V)\, \clr(\alpha_\ell^*),
\]
a weighted arithmetic mean with weights $w(V) \in \Delta^{m-1}$ in the closed simplex.
Aitchison-metric strategy-proofness translates, under $\clr$, to Euclidean-metric strategy-proofness on $H$: no LP should be able to misreport their peak so as to bring the output closer (in Euclidean norm on $H$) to their true peak.

\begin{lemma}\label{lem:sp-forces-projection}
Let $\phi \colon H^m \to H$ be of the form  \mbox{$\phi(x_1, x_2, \ldots, x_m) = \sum_\ell w_\ell x_\ell$} for fixed weights $w \in \Delta^{m-1}$.
Then  $\phi$ is Euclidean-metric strategy-proof on $H$ if and only if $w$ is an indicator: $w_{\ell_0} = 1$ for some $\ell_0$ and $w_\ell = 0$ for $\ell \neq \ell_0$.
\end{lemma}

\begin{proof}
For sufficiency, if $w_{\ell_0} = 1$ and the others are zero,  $\phi(x) = x_{\ell_0}$.
LP $\ell_0$ cannot manipulate (truth-telling delivers their peak exactly), and any other LP's report is ignored.
Since the output does not depend on the report of any $\ell \neq \ell_0$, no such LP can change the outcome by misreporting, so they have no profitable deviation either, and (SP) holds for every LP.

For necessity, suppose $0 < w_{\ell_0} < 1$ for  some $\ell_0$.
We show that there exists a profile on which LP $\ell_0$ can manipulate.
Fix a profile $(x_1^*, \ldots,  x_m^*)$ with $x_\ell^* \in H$ for each $\ell$ and $\phi^* \neq x_{\ell_0}^*$.
Such a profile exists. 
 Indeed, choose $x_{\ell_0}^*  := 0$  (which lies in $H$, since $\sum_i 0 = 0$) and pick the remaining peaks so that $\sum_{\ell \neq \ell_0} w_\ell x_\ell^* \neq 0$ (possible since $w_\ell > 0$ for at least one $\ell \neq \ell_0$, as $w_{\ell_0} < 1$; one nonzero vector $v \in H$ suffices, so we may set one such peak to $v$ and the rest to $0$).
The construction needs only $\dim  H \geq 1$, i.e., $n \geq 2$.
The asset-count threshold $n \geq 3$ enters the impossibility solely through the characterization \Cref{thm:characterization-F0} (the multi-pair triangle identity of \Cref{sec:F0-cauchy}), not through the strategy-proofness step.
Under truthful reporting, the output is
\[
\phi^* := \sum_{\ell=1}^m w_\ell\, x_\ell^*.
\]
With $x_{\ell_0}^* = 0$ in our profile, this simplifies to
\[
 \phi^*  =  w_{\ell_0} \cdot 0 + \sum_{\ell \neq \ell_0} w_\ell x_\ell^* = \sum_{\ell \neq \ell_0} w_\ell x_\ell^* \neq 0 = x_{\ell_0}^*,
\]
so $\phi^* \neq x_{\ell_0}^*$ as required, and LP $\ell_0$'s loss $\|\phi^* - x_{\ell_0}^*\|$ is strictly positive.
The fix $x_{\ell_0}^* = 0$ is what makes the construction work: $\sum_{\ell \neq \ell_0} w_\ell x_\ell^* \neq 0$ alone would not guarantee $\phi^* \neq x_{\ell_0}^*$ in general, since the peaks could coincide at a common nonzero point.
Consider LP $\ell_0$'s misreport (the division by $w_{\ell_0}$ that follows is legitimate, since $w_{\ell_0} > 0$ by the hypothesis $0 < w_{\ell_0} < 1$):
\[
\hat x_{\ell_0} := \frac{1}{w_{\ell_0}} \left( x_{\ell_0}^* - \sum_{\ell \neq \ell_0} w_\ell x_\ell^* \right).
\]
This is a valid report: any $\hat x_{\ell_0} \in H$, however large in norm, corresponds via the softmax inverse $\clr^{-1}(x)_i = e^{x_i}/\sum_k e^{x_k}$ to a peak $\hat\alpha_{\ell_0} := \clr^{-1}(\hat x_{\ell_0}) \in \simplex$, since $\clr^{-1}$ is a bijection from all of $H$ onto the open simplex $\simplex$ (the exponentials are strictly positive and sum to one after normalization), and by (A1) every $\hat\alpha_{\ell_0} \in \simplex$ is admissible.
The construction $\hat x_{\ell_0} \in H$ is verified directly: using $x_\ell^* \in H$ for every $\ell$, we get
\[
\sum_i (\hat x_{\ell_0})_i = \tfrac{1}{w_{\ell_0}}\Bigl(\sum_i (x_{\ell_0}^*)_i - \sum_{\ell \neq \ell_0} w_\ell \sum_i (x_\ell^*)_i\Bigr) = \tfrac{1}{w_{\ell_0}}(0 - 0) = 0.
\] 
The deposits $D$ are unchanged by the misreport,  so  $(\hat\alpha_{\ell_0}, \alpha_{-\ell_0}^*, D)$ is a valid input to $M$.
The resulting output is
\[
\hat\phi = w_{\ell_0}  \hat x_{\ell_0} + \sum_{\ell \neq \ell_0} w_\ell x_\ell^* = x_{\ell_0}^*.
\]
LP $\ell_0$'s loss therefore drops from $\|\phi^* -  x_{\ell_0}^*\| > 0$ to $0$, a strict improvement.
In fact, the manipulation is \emph{perfect}: the misreport drives the output exactly to LP $\ell_0$'s true peak, so any LP with $0 < w_{\ell_0} < 1$   can realize its ideal output on every profile at which truthful reporting would fail to do so.
Truth-telling is not weakly dominant, contradicting (SP).
\end{proof}

\begin{remark}[The Manipulating Peak May Be Extreme]\label{rem:manipulation-extreme}
Note that the misreport given by $\hat x_{\ell_0} = w_{\ell_0}^{-1}\bigl(x_{\ell_0}^* - \sum_{\ell \neq \ell_0} w_\ell x_\ell^*\bigr)$ has Euclidean norm that grows like $1/w_{\ell_0}$, so when $w_{\ell_0}$ is small the corresponding peak $\hat\alpha_{\ell_0} = \clr^{-1}(\hat x_{\ell_0})$ has components close to the boundary $\partial \simplex$ (some coordinates near $0$ and others near $1$).
The manipulator's reported peak is thus admissible by (A1), the unrestricted-domain axiom on the \emph{open} simplex, but may sit arbitrarily close to its boundary.
This is the right notion of admissibility for $\F_\Prod$, since trading-function equivalence classes correspond to open-simplex weight vectors, and is also the operative one for the LP context, where reporting an extreme weight is a legitimate strategic option.
\end{remark}

\begin{remark}[Application to $V$-Dependent Weights]\label{rem:sp-V-dependent}
\Cref{lem:sp-forces-projection} is stated for a constant weight vector $w \in \Delta^{m-1}$.
In the proof of \Cref{thm:impossibility}, the weights $w_\ell(V)$ depend on $V$ but not on the LP peaks $\alpha_\ell^*$.
Since a peak misreport leaves the deposit profile $D$ unchanged, the valuation profile $V$ and hence the weights $w_\ell(V)$ remain fixed throughout the misreport.
\Cref{lem:sp-forces-projection} therefore applies verbatim at each fixed $V$ with the constant weight vector $w = w(V)$.
\end{remark}

\begin{remark}[Deposit Misreports and Strategic Redeposits]\label{rem:deposit-misreports}
Under (SP), an LP's only available deviation is to misreport their peak; their deposit is treated as fixed.
This is the right notion in the LP context, since deposits are observable on-chain transfers, not reports.
An LP cannot ``misreport'' a deposit, only choose how much to deposit, which is a participation decision rather than a strategic communication.
The impossibility of \Cref{thm:impossibility} thus already holds under the weakest natural notion of LP strategic behavior, and strengthening (SP) to joint peak-and-deposit deviations would only make it easier to obtain.

A related but distinct strategic act is \emph{redepositing}: an LP withdraws their position and deposits a different amount, thereby changing $V_\ell$ and, when the weighting is $V$-dependent, the realized weight $w_\ell(V)$.
This is not a misreport in the sense of (SP), since the new deposit is again observable on-chain, but it is a genuine strategic lever when the weighting depends on $V$.
A meaningful strategy-proofness notion for redeposits would require a model of the redeposit cost (gas, slippage, opportunity cost), without which the trivial deviation $D_\ell \mapsto \lambda D_\ell$ followed by an instantaneous reversal makes (SP) vacuous on $V$-homogeneous weightings such as the stake-proportional rule.
We treat this as outside the present scope.
The fee-sensitive direction (Q4 of \Cref{sec:open-questions}) is the natural setting in which to formalize it.
\end{remark}

\subsection{Proof of \Cref{thm:impossibility}}\label{sec:impossibility-proof}

\begin{proof}[Proof of \Cref{thm:impossibility}]
Suppose by contradiction that $M$ satisfies (A0)--(A4), (C), and (SP).
By \Cref{thm:characterization-F0}, $M$ is a weighted Aitchison centroid with continuous, 
 equivariant weights \mbox{$w \colon \R^m_{>0} \to \Delta^{m-1}$}.

It suffices to explicitly construct a single profile at which (SP) fails.
 Fix a symmetric valuation profile $V^{\mathrm{sym}} = (v,  \ldots, v)$, realized for instance 
by any equal-deposit profile $D_1 = \ldots = D_m$.
By equivariance of $w$ (\Cref{lem:V-dependent-weights}), 
it holds that  $w(V^{\mathrm{sym}}) = (1/m,  \ldots, 1/m)$.
Since a peak misreport leaves $D$, hence $V$, unchanged, the weight vector remains $(1/m, \ldots, 1/m)$ before and after any misreport, and \Cref{lem:sp-forces-projection} applies verbatim at this fixed weight vector.
Since $1/m \in (0,1)$ for every $\ell$ when $m \geq 2$, the uniform vector is not an indicator, so \Cref{lem:sp-forces-projection} shows the map is manipulable, contradicting (SP).
\end{proof}

\subsection{Interpretation: Mean-Type Versus Median-Type}\label{sec:interpretation}

The impossibility is the AMM-design analog of the joint Arrow and Gibbard--Satterthwaite tension, 
arising from a collision between two types of aggregation rule.
A \emph{mean-type} rule uses the numerical values of the peaks: shifting one LP's peak shifts the output correspondingly.
By contrast, a \emph{median-type} rule uses only the ranking of the peaks: replacing a peak by any other value of the same rank leaves the output unchanged.
A single-LP dictator is the extreme case, ignoring all peaks but one.
The Arrovian core forces mean-type aggregation (\Cref{thm:characterization-F0}), while strategy-proofness forces median-type aggregation (\Cref{lem:sp-forces-projection}), and no rule can be both.
Equivalently, within the centroid family the only strategy-proof rule is a single-LP dictator, which violates anonymity (A4), so the impossibility is the AMM-design instance of Gibbard--Satterthwaite restricted to that family.
The dictator's role here also explains why dropping anonymity reopens the design space: the escape routes catalogued in \Cref{tab:routes} are precisely the axiom removals that admit a compatible rule.


\section{The Two-Asset Case}\label{sec:two-asset}

The impossibility of \Cref{thm:impossibility} is stated for $n \geq 3$ assets,  an assumption that excludes the two-asset constant-product pools that account for the bulk of deployed AMM volume (Uniswap v2/v3 pairs, and two-asset Curve pools).
We now treat the two-asset case directly, and find that the picture differs in a structurally clean way: at $n = 2$ axiom (A3) is vacuous for a one-pair reason, the Arrovian core no longer forces mean-type aggregation, and the core together with (SP) instead characterizes the Moulin generalized medians in the log-ratio coordinate.
Strategy-proofness and Arrovian fairness are therefore compatible at $n = 2$, and the impossibility does not arise.

\subsection{Setup at $n = 2$}\label{sec:two-asset-setup}

We work on $\F_\Prod$ with $n = 2$ and $m \geq 2$, with LP preferences Aitchison-single-peaked on $\simplex$.
The clr map identifies $\simplex$ with the line $H = \{x \in \R^2 \mid x_1 + x_2 = 0\}$, which we coordinatize by $t(\alpha) := \tfrac{1}{2}\log(\alpha_1/\alpha_2) \in \R$, so that $\clr(\alpha) = (t(\alpha), -t(\alpha))$.
Under this coordinate, the Aitchison metric is a positive scalar multiple of the Euclidean metric on $\R$:
\[
\Aitch(\alpha, \beta) = \|\clr(\alpha) - \clr(\beta)\|_2 = \sqrt{2}\,|t(\alpha) - t(\beta)|.
\]
Thus, $t$ is a \emph{similarity} from $(\simplex, \Aitch)$ to $(\R, |\cdot|)$ with ratio $\sqrt{2}$, not literally an isometry; the constant $\sqrt{2}$ is, however, immaterial for strategy-proofness, which depends only on the \emph{ordering} of distances and not on their absolute scale, so the Moulin/Ching characterization on $\R$ with the Euclidean metric applies unchanged to single-peaked preferences on $\R$ with peak $t_\ell^* := t(\alpha_\ell^*)$, the standard domain of \citet{moulin-1980}.
The $t$-coordinate is the log-odds of the two-asset pool, so an LP's peak $t_\ell^*$ is just its preferred log-price-ratio between the two assets, recovering the standard scalar parameter on the two-asset family.

\subsection{Vacuity of (A3) and the Characterization}\label{sec:two-asset-characterization}

\begin{lemma}[Vacuity of (A3) at $n = 2$]\label{lem:A3-vacuous-n-two}
At $n = 2$, axiom (A3) imposes no constraint on $M$ beyond what is implied by the other axioms.
\end{lemma}

 \begin{proof}
At $n = 2$ there is one unordered asset pair $\{1, 2\}$, and the $(1,2)$-restricted preference of LP $\ell$ is the full preference of LP $\ell$ on $\F_\Prod$, since a preference on a one-pair family is determined by its single pair-restriction.
(A3) then requires the  $(1,2)$-pricing to depend only on the $(1,2)$-restricted preferences and $V$, which is the trivial statement that the pricing depends on the LP preferences and $V$.
\end{proof}

\begin{theorem}[Characterization on $\F_\Prod$ at $n = 2$]\label{thm:characterization-F0-n-two}
Let $n = 2$ and $m \geq 2$, and consider a mechanism  $M \colon (\simplex \times \R^A_{>0})^m \to \simplex$ satisfying:
\begin{description}
\item[(V) Factoring through $V$.]
The dependence of $M$ on the deposit profile $D$ factors through the valuation profile $V = (\fref(D_1), \fref(D_2), \ldots, \fref(D_m))$.
\end{description}
Then $M$ satisfies the axioms 
(A0), (A1), (A2), (A4), (C), and (SP) if and only if there exist continuous functions $y_1, y_2, \ldots, y_{m-1} \colon \R^m_{>0} \to \R \cup \{-\infty, +\infty\}$, invariant under joint permutation of arguments, such that for every LP profile $(\alpha^*, D)$,
\begin{equation}\label{eq:moulin-median-n-two}
t\bigl(M(\alpha^*, D)\bigr) = \median\bigl( t_1^*, t_2^*, \ldots, t_m^*,\, y_1(V), y_2(V), \ldots, y_{m-1}(V) \bigr),
\end{equation}
where $t_\ell^* = t(\alpha_\ell^*)$ and $\median$ denotes the $m$-th order statistic of the $2m - 1$ extended-real values.
Axiom (A3) is vacuous at $n = 2$ (\Cref{lem:A3-vacuous-n-two}) and may be included or omitted indifferently in the hypotheses.
\end{theorem}

The hypothesis (V) is a modeling choice rather than a theorem, since (A3) is vacuous at $n = 2$ and so cannot force factoring through $V$ as it does on $\F_\Prod$ at $n \geq 3$  (where the same factoring is a consequence of the axioms via (A3)).
We adopt it to enable a uniform comparison with \Cref{thm:characterization-F0}.

\begin{remark}[The Three Roles of Hypothesis (V)]\label{rem:V-three-roles}
The hypothesis that deposit-dependence factors through the valuation profile $V = (\fref(D_\ell))_\ell$ plays three distinct roles across the paper, worth recording in one place.
On $\F_\Prod$ with $n \geq 3$ (\Cref{thm:characterization-F0}): it is a \emph{consequence} of the axioms, built into (A3) via the pairwise extraction argument of \Cref{lem:pair-extraction-F0}, which forces every pairwise rule $\psi_{ij}$ to depend on $D$ only through $V$.
On $\F_\Prod$ at $n = 2$ (present theorem): it is a \emph{modeling choice} (explicitly labeled as hypothesis (V)), since (A3) is vacuous and cannot drive this factoring.
Without it the same Moulin--Ching argument yields phantoms depending on the full $D$ rather than $V$ alone.
On $\F_\Sum$ with $n \geq 3$ (\Cref{thm:characterization-Fgamma}), the situation is the same as at $n = 2$: 
 it is a modeling choice imposed to match the $\F_\Prod$ setting and enable a clean comparison, as noted explicitly there.
The characterization form (generalized medians, weighted Aitchison centroids) is the same in all three cases. 
What changes is whether the phantom or weight parameterization depends on $V$ or on $D$.
\end{remark}

\begin{proof}
The argument has a clean slice-then-thread structure: transport the mechanism to a Euclidean slice via the $t$-coordinate, apply the classical Moulin/Ching characterization, and translate back.

\emph{Converse.}
Given continuous, symmetric extended-real phantoms $(y_k)_{k=1}^{m-1}$, the median of the $m$ peaks and $m - 1$ phantoms is the $m$-th order statistic of $2m - 1$ values, finite (hence in $\R$) because the $m$ peaks are finite, and well-defined as the unique middle value.
The rule satisfies (A0) (varying peaks varies the median),  (A1) (defined everywhere on $(\simplex)^m$), (A2) (the $m$-th order statistic lies in $[\min_\ell t_\ell^*, \max_\ell t_\ell^*]$: a value below $\min_\ell t_\ell^*$ would force all $m$ peaks into the top $m-1$ positions, impossible; symmetrically above), (A4) (symmetry of the median in the peaks and of the phantoms in $V$), and (C) (continuity of the median in the order topology on $\R \cup \{-\infty, +\infty\}$).
Strategy-proofness (SP) holds by \citet[Theorem 1]{moulin-1980}, transported from the order metric to the Euclidean metric on $\R$, which coincide on the order topology and give the same notion of strategy-proofness for symmetric-loss single-peaked preferences.

\emph{Direct part.}
Suppose $M$ satisfies axioms (A0), (A1), (A2), (A4), (C), and (SP), and write \mbox{$\tilde M(t_1, t_2, \ldots, t_m;\, V) := t(M(\clr^{-1}(t_1, -t_1), \clr^{-1}(t_2, -t_2), \ldots, \clr^{-1}(t_m, -t_m);\, D))$} for the transported mechanism, well-defined under the factoring hypothesis because the deposit-dependence reduces to $V$.
The transported mechanism $\tilde M \colon \R^m \times \R^m_{>0} \to \R$ inherits (A0) (nontriviality), (A1) (unrestricted domain on $\R^m \times \R^m_{>0}$), (A2) (Pareto efficiency in the peak interval $[\min_\ell t_\ell^*, \max_\ell t_\ell^*]$, by the bijectivity and monotonicity of $t$), (A4) (joint equivariance of $\tilde M$ in the peaks and $V$), (C) (continuity), and (SP) (Euclidean-metric strategy-proofness on $\R$, since $t$ is a similarity from the Aitchison metric to the Euclidean metric with ratio $\sqrt{2}$ as recorded above, and strategy-proofness is invariant under positive scalar multiples of the metric).

Now fix $V \in \R^m_{>0}$ and consider the slice $\tilde M(\,\cdot\,;\, V) \colon \R^m \to \R$.
Observe   that this slice is continuous, Pareto-efficient (peak-only output in $[\min_\ell t_\ell^*, \max_\ell t_\ell^*]$), strategy-proof on the domain of symmetric-loss single-peaked preferences on $\R$, and anonymous at fixed $V$ as a consequence of joint equivariance applied with $V$-symmetric profiles.
\Citet[Theorem 1]{moulin-1980} (in the continuum-of-alternatives formulation of \citet{ching-1997}) characterizes such rules on $\R$: there exist phantom values \mbox{$y_1(V), y_2(V), \ldots, y_{m-1}(V) \in \R \cup \{-\infty, +\infty\}$}, unique up to relabeling, such that
\[
\tilde M(t_1, t_2, \ldots, t_m;\, V) = \median\bigl( t_1, t_2, \ldots, t_m,\, y_1(V), y_2(V), \ldots, y_{m-1}(V) \bigr)
\]
for every peak tuple $(t_1, t_2, \ldots, t_m) \in \R^m$.
The phantoms vary continuously with $V$ by (C) and are jointly symmetric in $V$ by (A4).
They may sit outside the realized peak range, and the endpoints $\pm\infty$ are permitted (yielding extremal-order-statistic rules such as the minimum or maximum of the peaks).
Translating back via $\clr^{-1}$ yields (\ref{eq:moulin-median-n-two}).
\end{proof}

\subsection{Why the Impossibility Dissolves at $n = 2$}\label{sec:two-asset-why}

The contrast between \Cref{thm:characterization-F0-n-two} and the $n \geq 3$ case isolates axiom (A3) as the load-bearing condition.
At $n = 2$, the cocycle on the simplex has only the antisymmetric relation 
given by $\Psi_{12}(u) + \Psi_{21}(-u) = 0$ available, rather than a multi-pair triangle identity.
The Cauchy step of \Cref{lem:Phi-linear} cannot be set up, and the Arrovian core does not force mean-type form.
Median-type rules then satisfy every axiom including (SP), and the joint Arrow--Gibbard--Satterthwaite tension dissolves.
The manipulation step of \Cref{lem:sp-forces-projection} itself remains valid at $n = 2$ (it requires only $\dim H \geq 1$, equivalently, $n \geq 2$), but at $n = 2$ it has no premise to act on, since the Arrovian core no longer forces a linear aggregator.

\begin{remark}[A Thomsen Strengthening as the Open Direction]\label{rem:thomsen-substitute}
Whether a strengthening of (A3) that remains nontrivial at $n = 2$ can recover the mean-type rigidity, and so recover an impossibility analog at $n = 2$, is open.
The natural candidate is a Thomsen-type cancellation condition in the spirit of additive conjoint measurement \citep{aczel-1966}, requiring the two-asset rule to cohere with a latent multi-pair cocycle structure across deposit profiles.
Informally, a Thomsen condition asks that if two pairs of deposit profiles produce matching pricing data in a cross-cancellation pattern, then a derived third pair must match as well, the substitute the literature uses for the triangle identity when the multi-pair cocycle is unavailable.
Whether such a condition admits a normatively reasonable AMM-side reading, and whether it suffices to force linearity of the aggregator in clr-coordinates, are the substantive questions.
We leave both open.
\end{remark}

\begin{remark}[The Arrovian Core Without (SP) at $n = 2$]\label{rem:n-two-without-SP}
 \Cref{thm:characterization-F0-n-two} adds (SP) to the Arrovian core to obtain a tight characterization.
Without (SP), the conditions (A0)--(A2), (A4), and (C), with (A3) vacuous by \Cref{lem:A3-vacuous-n-two}, constrain $M$ only to be a continuous, anonymous, Pareto-efficient selection from the peak hull in the log-ratio coordinate $t \in \R$.
This is a substantially larger class than the Moulin generalized medians of \Cref{thm:characterization-F0-n-two}: 
 the additional rigidity that picks out the Moulin family at $n = 2$ is supplied entirely by (SP), without which the class lacks a finite parameterization and admits, for example, weighted-Aitchison-centroid rules of the kind characterized at $n \geq 3$.
The contrast is sharp: at $n \geq 3$ on $\F_\Prod$, the Arrovian core alone forces a finite-dimensional family (the Aitchison centroids), while at $n = 2$ it is too weak to do so.
\end{remark}

\begin{remark}[Practical Reading]\label{rem:two-asset-practical}
\Cref{thm:characterization-F0-n-two} is a positive result.
 At $n = 2$, a rule satisfying the full Arrovian core together with strategy-proofness does exist, namely the Moulin generalized median.
However, no deployed two-asset protocol has, to our knowledge, implemented such a rule.
Uniswap v2 and v3 fix the trading function at launch instead, the same route taken at every $n$, which drops nontriviality (A0).
The contrast with the $n \geq 3$ case on $\F_\Prod$ is that the Moulin route is simply not available there:
by \Cref{thm:impossibility}, no rule satisfies the Arrovian core jointly with (SP), so any preference-aggregating mechanism at $n \geq 3$ must drop one of these axioms.
The distinction is that, at $n \geq 3$ on $\F_\Prod$, the Arrovian core forces a \emph{different} family altogether (the Aitchison centroids of \Cref{thm:characterization-F0}), within which strategy-proofness selects only the single-LP dictators.
The Moulin generalized medians, which exist as a distinct family, do not satisfy the $n \geq 3$ Arrovian core to begin with (they violate (A3) in the multi-pair sense), so they are not even candidates there.
At $n = 2$ the two families collapse to one, the multi-pair distinction has nowhere to live, and the medians are admissible.
\end{remark}


\section{Scope of the Impossibility on Related Spaces}\label{sec:scope}

Having established the impossibility on $\F_\Prod$ at $n \geq 3$  and the compatible characterization at $n = 2$, we now mark the impossibility's boundaries on the remaining design spaces.
On the symmetric CEMM family $\F_\Sum$ the conflict disappears entirely, because mechanism IIA goes vacuous on a one-parameter family; on the union $\U = \F_\Prod \cup \F_\Sum$ the impossibility survives on the $\F_\Prod$-arm but dissolves on the $\F_\Sum$-arm, made precise through a regime-wise structural decomposition.
The constructions below are not used in proving either headline theorem.
They serve only to locate the impossibility within the larger design space.
The key contrast: on $\F_\Prod$ a pair-restriction is lossy and (A3) does real work, whereas on $\F_\Sum$ a pair-restriction recovers the whole design parameter and (A3) says nothing.
The dimension of the design space is what governs this asymmetry: $\F_\Prod$ has $n - 1$ free parameters, while $\F_\Sum$ has just one, regardless of $n$.

\subsection{The Symmetric CEMM Family}\label{sec:Fgamma-characterization}

We consider $\F_\Sum = \{f_\gamma \mid \gamma \in (-\infty, 1]\}$ with $n \geq 3$, $m \geq 2$, and LP preferences single-peaked in $\dSum$.
On this family, axiom (A3) is essentially vacuous, since $\F_\Sum$ is one-parameter and any pair-restriction is injective: knowing an LP's preferred $(i,j)$-pricing exponent $1 - \gamma_\ell^*$ recovers $\gamma_\ell^*$ and hence the full preference, so (A3) adds nothing beyond the requirement that the output depend on LP preferences and $V$.

\begin{lemma}[Vacuity of (A3) on $\F_\Sum$]\label{lem:A3-vacuous-Fgamma}
For each pair $i, j \in A$ with $i \neq j$, the map sending $\gamma \in (-\infty, 1]$ to the $(i,j)$-pricing function $I \mapsto (I_j/I_i)^{1 - \gamma}$ is injective and continuous.
An LP's preference on $\F_\Sum$ is therefore fully determined by its $(i,j)$-restriction for any single pair $(i,j)$.
Consequently, (A3) imposes no constraint on $M$ beyond what is implied by the other axioms.
\end{lemma}

\begin{proof}
Two distinct $\gamma_1, \gamma_2 \in (-\infty, 1]$ give pricing functions $(I_j/I_i)^{1-\gamma_1}$ and $(I_j/I_i)^{1-\gamma_2}$, which differ on any open subset of the inventory-ratio space $\R_{>0}$.
Injectivity follows.
Continuity in the $C^2$ topology is immediate.
Hence, the $(i,j)$-restriction of an LP's preference (a single-peaked preference on pricing functions, induced by a peak $\gamma_\ell^*$) determines $\gamma_\ell^*$ uniquely.
Knowing $\gamma_\ell^*$ determines the full preference, so (A3) is tautologous on  $\F_\Sum$.
\end{proof}

The characterization on $\F_\Sum$ reduces essentially to Moulin's classical theorem on single-peaked preferences over a one-dimensional space.
 Moulin's theorem characterizes the \emph{strategy-proof}, anonymous, peak-only rules, so strategy-proofness must appear among the hypotheses.
This is not a defect but the substantive point: on $\F_\Sum$ the Arrovian core is consistent with strategy-proofness, the generalized medians being witnesses, whereas on $\F_\Prod$ no rule reconciles the two (\Cref{thm:impossibility}).

On $\F_\Sum$ we restrict attention to mechanisms whose deposit dependence factors through the valuation profile  $V = (\fref(D_1), \fref(D_2), \ldots, \fref(D_m))$, exactly as on $\F_\Prod$.
On $\F_\Prod$ this factoring follows from the axioms, being built into (A3), whereas on $\F_\Sum$ it is a modeling choice rather than a theorem, a difference we make explicit rather than suppress.

\begin{theorem}[Characterization on $\F_\Sum$]\label{thm:characterization-Fgamma}
Let $n \geq 3$ and $m \geq 2$, and consider a mechanism $M_\gamma \colon ((-\infty, 1] \times \R^A_{>0})^m \to (-\infty, 1]$ satisfying:
\begin{description}
\item[(V) Factoring through $V$.]
The dependence of $M_\gamma$ on the deposit profile $D$ factors through the valuation profile $V = (\fref(D_1), \fref(D_2), \ldots, \fref(D_m))$.
\end{description}
Then $M_\gamma$ satisfies (A0)--(A4), (C), and (SP) if and only if there exist continuous functions $y_1, y_2, \ldots, y_{m-1} \colon \R_{>0}^m \to [-\infty, 1]$, invariant under joint permutation of arguments, such that
\begin{equation}\label{eq:moulin-median}
M_\gamma(\gamma_1^*, \gamma_2^*, \ldots, \gamma_m^*;\, V)  =  \median\bigl( \gamma_1^*, \gamma_2^*, \ldots, \gamma_m^*,\, y_1(V), y_2(V), \ldots, y_{m-1}(V) \bigr),
\end{equation}
where $\median$ denotes the $m$-th order  statistic of the $2m-1$ extended-real values.
(The phantoms are valued in the one-sided compactification $[-\infty, 1]$, in contrast to the two-sided compactification $\R \cup \{-\infty, +\infty\}$ that arises in \Cref{thm:characterization-F0-n-two}; the asymmetry is explained after the theorem.)
\end{theorem}

We call the functions $y_k$ the \emph{phantoms} of the rule, and rules of the form (\ref{eq:moulin-median}) the \emph{generalized median rules}.
 Several features of the representation are worth recording.

\emph{Codomain of the phantoms.} 
The phantoms are valued in the extended peak interval $[-\infty, 1]$, the peak interval $(-\infty, 1]$  with its open end compactified by the point $-\infty$ in the order topology. 
Only phantoms may take the value $-\infty$.
The mechanism's output, being the median of $2m-1$ values among which $m$ are finite peaks, is always finite, i.e., in $(-\infty, 1]$.
The compactification is \emph{one-sided} because the design space $(-\infty, 1]$ is closed at the upper end and open at the lower end.
Only the open end requires a limit point.
This contrasts with the two-asset case \Cref{thm:characterization-F0-n-two}, where the $t$-coordinate ranges over all of $\R$ and phantoms are valued in the two-sided compactification $\R \cup \{-\infty, +\infty\}$.

\emph{Pareto efficiency is automatic.} 
Pareto efficiency (A2) is not a restriction on the phantoms: 
the $m$-th order statistic of $2m-1$ values, $m$ of which are the peaks, always lies in $[\min_\ell \gamma_\ell^*, \max_\ell \gamma_\ell^*]$, since a value below $\min_\ell \gamma_\ell^*$  would force all $m$ peaks into the top $m-1$ positions (impossible), and symmetrically above.
This is in sharp contrast to the $\F_\Prod$ case, where (A2)  does substantive work in pinning down the centroid weights (\Cref{lem:V-dependent-weights}).

\emph{No constraint between phantoms and peaks.} 
The phantoms are unconstrained by the peaks.
In particular, they may sit outside the realized peak range, as in the $c$-anchored rule of \Cref{def:canonical-Fgamma}. 
A phantom at the endpoint  $-\infty$ yields a lower-order-statistic rule (e.g., minimum-of-peaks), which is admissible.
The full peak range, in fact, may be either above or below all phantoms, in which case the rule reduces to an order-statistic of the peaks alone.

\begin{proof}
Unlike \Cref{thm:characterization-F0}, where deposit dependence factors through $V$ as a consequence of the axioms via (A3), here that factoring is an explicit assumption rather than a theorem, since (A3) is vacuous on $\F_\Sum$ (\Cref{lem:A3-vacuous-Fgamma}).
The analogous derivation therefore does not go through, and the factoring is a modeling choice.
We adopt it to match the $\F_\Prod$ setting and to enable a clean correspondence between the two characterizations.
This asymmetry is harmless for the present argument (the phantoms simply absorb the resulting $V$-dependence), but it should be kept in mind when comparing the two results.
Without the factoring hypothesis, the same Moulin--Ching argument still yields the generalized-median form, but with phantoms depending on the full deposit profile $D$ rather than on $V$ alone. 

\emph{Direct implication: axioms force the form.}
Suppose $M_\gamma$ satisfies (A0)--(A4), (C), and (SP), and assume the factoring hypothesis stated just before the theorem (the dependence of $M_\gamma$ on $D$ factors through $V$).
By \Cref{lem:A3-vacuous-Fgamma}, axiom (A3) plays no role.
For each fixed $V$, consider $M_\gamma$ as a function of $m$ peaks alone.
By (A2), the output lies in the Pareto-undominated set, which on a one-dimensional space with single-peaked preferences is the interval $[\min_\ell \gamma_\ell^*, \max_\ell \gamma_\ell^*]$.
By (A4), $M_\gamma$ is symmetric in its arguments; by (C), continuous.
Since preferences on $\F_\Sum$ are parameterized by their peaks (\Cref{sec:lps}) and the deposit dependence factors through $V$ by hypothesis, the dependence of $M_\gamma$ on the LP profile $\theta^*$ is, after fixing $V$, a function of the peak tuple $(\gamma_1^*, \gamma_2^*, \ldots, \gamma_m^*)$ alone.

We apply the homeomorphism $h :=  \arctan \colon (-\infty, 1] \to (-\pi/2, \pi/4]$; 
 under $h$, the metric $\dSum$ corresponds to Euclidean distance, and $h$ carries single-peaked preferences to single-peaked preferences and peaks to peaks.
Write $\tilde M := h \circ M_\gamma  \circ (h^{-1}, h^{-1}, \ldots, h^{-1})$ for the transported mechanism on the half-open interval $(-\pi/2, \pi/4]^m$; it inherits (A0)--(A4) and (C), and it inherits (SP) as Euclidean-metric strategy-proofness, since strategy-proofness is an ordinal property of single-peaked preferences and $h$ preserves the betweenness order on the line.
Thus, at each fixed $V$, $\tilde M(\,\cdot\,;V)$ is a continuous, anonymous, peak-only, strategy-proof rule on the single-peaked domain over the interval $(-\pi/2, \pi/4]$.
Ching's characterization of such rules on an arbitrary (possibly half-open or unbounded) real interval \citep{ching-1997} is the continuum-of-alternatives form of Moulin's theorem \citep[Theorem 1]{moulin-1980}, \citep[Ch.~10]{moulin-axioms-of-cooperative-decision-making}. 
By it, there exist phantoms $\tilde y_1(V) \leq \cdots \leq \tilde y_{m-1}(V)$ in the compactified interval $[-\pi/2, \pi/4]$ such that
\[
\tilde M(t_1, t_2, \ldots, t_m; V) = \median\bigl(t_1, t_2, \ldots, t_m,\, \tilde y_1(V), \tilde y_2(V), \ldots, \tilde y_{m-1}(V)\bigr)
\]
for all peak-tuples $t \in (-\pi/2, \pi/4]^m$, the median being the $m$-th order statistic of the $2m-1$ extended-real values.
The phantoms vary continuously with $V$ by (C) and are equivariant under joint permutation by (A4); they are not constrained to lie in the peak range, and the left endpoint $-\pi/2$ is permitted (it yields lower-order-statistic rules such as $\min_\ell$). 
Pareto efficiency (A2) is automatic for every choice of phantoms,  by the order-statistic argument given after the theorem statement, so it imposes no condition on $(\tilde y_k)$.

Translating back via $h^{-1} = \tan$ gives (\ref{eq:moulin-median}) on $(-\infty, 1]$: 
 the map \mbox{$\tan \colon [-\pi/2, \pi/4] \to [-\infty, 1]$} 
  (extended by setting $\tan(-\pi/2) := -\infty$)
 is an order isomorphism, each phantom $\tilde y_k(V)$ maps to $y_k(V) := \tan(\tilde y_k(V))$, and the median commutes with $\tan$.
The output is always a finite value in $(-\infty,1]$ because the $m$ (finite) peaks pin the $m$-th order statistic into their range. 
This completes the direct part.

\emph{Converse implication: the form satisfies the axioms.}
Given continuous, symmetric extended-real phantoms  $(y_k)_{k=1}^{m-1}$, the median of $m$ LP peaks together with $m-1$ phantoms is well-defined since the multiset has $2m - 1$ entries, an odd number, with a unique $m$-th order statistic, and it is finite (hence in $(-\infty,1]$) because at least the $m$ peaks are finite.
 This rule satisfies (A0) (varying LP peaks changes the median), (A1) (defined everywhere), (A2) (the $m$-th order statistic lies in $[\min_\ell \gamma_\ell^*, \max_\ell \gamma_\ell^*]$ by the order-statistic argument above, for \emph{any} phantom values), (A4) (by symmetry of the median and of the phantoms in $V$), and (C) (continuity of the median operation, including at a $-\infty$ phantom, in the order topology).
Axiom (A3) holds vacuously by \Cref{lem:A3-vacuous-Fgamma}.
Finally (SP) holds because the generalized median is strategy-proof in the order metric on the line by \citet[Theorem 1]{moulin-1980}, and this transports verbatim to $\dSum$ via the homeomorphism $h = \arctan \colon (-\infty, 1] \to (-\pi/2, \pi/4]$, under which $\dSum$ is the Euclidean metric and (SP) is the ordinal betweenness property $h$ preserves.
\end{proof}

\begin{remark}[Why (A3) Is Vacuous in Retrospect]  \label{rem:Fgamma-A3-predictable}
Whether (A3) is nonvacuous depends on whether pair-restrictions are lossy.
On $\F_\Prod$ with $n \geq 3$ the $(i,j)$-restriction sees only $\rho_{ij}$ and loses cross-pair information, so (A3) is the operative constraint; on $\F_\Sum$ any single pair-restriction recovers the entire $\gamma$, so (A3) is silent, and this contrast is at the heart of the impossibility argument.
The cross-pair information lost by a single $(i,j)$-restriction on $\F_\Prod$ is precisely the multiplicative cocycle of structural price ratios, whose Cauchy-type functional equation drives \Cref{thm:characterization-F0}; on $\F_\Sum$ that cocycle collapses to the scalar $\gamma$ recovered from any single pair, leaving (A3) with no degree of freedom to constrain.
The phantoms $y_k(V)$ satisfy withdrawal neutrality (A5) precisely when they are $V$-independent (the rule deposit-blind), the $\F_\Sum$ analog of \Cref{thm:characterization-F0}(b): each phantom is recovered from the rule's output at an extreme peak profile, so a single-LP withdrawal that moved any phantom would change the rule, and only $V$-independent phantoms escape this detection.
\end{remark}

\subsection{The Canonical Choice} \label{sec:Fgamma-canonical}

Among the family (\ref{eq:moulin-median}),  a canonical choice is to place all phantoms at $\gamma = 0 = c$.
The $c$-anchored variant below is the natural symmetric representative.

\begin{definition}\label{def:canonical-Fgamma}
The \emph{$c$-anchored median rule} on  $\F_\Sum$ is
\[
M_\gamma^{\mathrm{can}}(\gamma^*)  :=  \median\bigl( \gamma_1^*, \gamma_2^*, \ldots, \gamma_m^*,\, \underbrace{0, 0, \ldots, 0}_{m - 1}  \bigr).
\]
\end{definition}

This rule satisfies (A0)--(A4), (C), (SP), and (A5),  and reappears in the union analysis of \Cref{rem:union}.
Pareto efficiency (A2) holds by the order-statistic argument following \Cref{thm:characterization-Fgamma} (the median of the $m$ peaks and $m-1$ phantoms lies in $[\min_\ell \gamma_\ell^*, \max_\ell \gamma_\ell^*]$ for any phantoms), and illustrates that a phantom at $0$ need not lie in the realized peak range: 
 if all peaks are strictly positive the output is the smallest peak, not $0$, because the $m-1$ phantoms at $0$ are all outranked by the $m$ positive peaks and the median selects the lowest among the latter.

\begin{remark}[Consistency of (SP) with the Arrovian Core on $\F_\Sum$]\label{rem:Fgamma-SP-included}
The strategy-proofness sub-claim of \Cref{thm:characterization-Fgamma} reduces to Moulin's theorem on the line: 
the generalized median is strategy-proof under the order metric \citep[Theorem 1]{moulin-1980},  transported across $h = \arctan$, under which $\dSum$  is Euclidean distance and (SP) is the ordinal betweenness property $h$ preserves.
The class is nonempty (the $c$-anchored rule), so the impossibility of \Cref{sec:impossibility} does not arise on $\F_\Sum$.
The contrast with $\F_\Prod$ is not an artifact of the metric.
(SP) is metric-specific (on $\F_\Sum$ it refers to  $\dSum$, on $\F_\Prod$ to the Aitchison metric), 
 but no metric on $\F_\Prod$ rescues the impossibility.
The manipulation of \Cref{lem:sp-forces-projection} is purely linear-algebraic,  an exact misreport driving a linear aggregator on \mbox{$H \cong \R^{n-1}$} to the manipulator's peak, valid for any norm-induced metric via clr.
The key point is that the misreport drives the output distance to \emph{exactly zero} (the output lands at the manipulator's true peak): for any metric $d$ that is even marginally sensitive to the peak, the manipulator's distance drops from $d(M(\theta^*, D), \theta_{\ell_0}^*) > 0$ to $d(\theta_{\ell_0}^*, \theta_{\ell_0}^*) = 0$, so the manipulation is profitable under any metric whatsoever, not just the Euclidean one.
  What distinguishes $\F_\Sum$ is one-dimensionality, which makes (A3) vacuous, not a softer metric.
\end{remark}

 \subsection{Extending the Impossibility}\label{sec:impossibility-extensions}

We close this section by recording where the impossibility extends.
It does not arise on $\F_\Sum$ or on $\F_\Prod$ at $n = 2$, where (A3) is vacuous in each case, and on the union $\U$ it survives only on the $\F_\Prod$-arm with $n \geq 3$ (\Cref{thm:impossibility}).
It is thus concentrated precisely where (A3) does substantive work, namely on multi-asset weighted-product design spaces where pair-restrictions are genuinely lossy.

\begin{remark}[The Union $\U = \F_\Prod \cup \F_\Sum$]\label{rem:union}
On the union $\U = \F_\Prod \cup \F_\Sum$, glued at the symmetric product $c$ and given the wedge metric (restricting to $\Aitch$ on $\F_\Prod$ and to $\dSum$ on $\F_\Sum$), the impossibility survives on the $\F_\Prod$-arm and dissolves on the $\F_\Sum$-arm.
We sketch the argument as a regime-wise reading of the two characterizations.
 At any profile with all peaks in the relative interior of  $\F_\Prod$, Pareto efficiency (A2) confines the output to the Aitchison hull of the peaks, which lies in $\F_\Prod$.
On a neighborhood of such profiles, the mechanism factors through $\F_\Prod$ and \Cref{thm:characterization-F0} forces a weighted Aitchison centroid.
Symmetrically, at profiles with all peaks in the interior of $\F_\Sum$, the output is forced into $\F_\Sum$ and \Cref{thm:characterization-Fgamma} gives a generalized median.
Continuity (C) glues the two regimes at $c$, where the symmetric centroid and the symmetric median agree.
A full proof requires tracking the mechanism's behavior on mixed profiles (peaks distributed across both arms), which we do not pursue here.
Adding (SP) kills the $\F_\Prod$-regime by \Cref{thm:impossibility} but leaves the $\F_\Sum$-regime intact (\Cref{rem:Fgamma-SP-included}); the local asymmetry of $\U$ near $c$ (an $(n-1)$-disk wedged to a segment) is what lets the regimes differ.
\end{remark}

The impossibility passes through two logically independent pressure points,  \emph{linearity in clr-coordinates} (forced by the Arrovian core via the cocycle/Cauchy argument, needing $n \geq 3$ and (A3) but not single-peakedness) and \emph{dictatorship from strategy-proofness} (\Cref{lem:sp-forces-projection}, a statement about linear maps independent of preference shape).
Neither (A1) nor (C) furnishes an escape, since the manipulation uses a single profile and a single misreport, insensitive to both. 
 Escaping requires relaxing one of the four load-bearing conditions of \Cref{sec:intro-fork}.

\begin{remark}[Beyond Single-Peakedness]\label{rem:beyond-single-peaked}
Single-peakedness enters the impossibility only through the Pareto-set identification of \Cref{lem:pareto-hull}.
The Cauchy and strategy-proofness steps are preference-shape independent.
The impossibility therefore extends to any preferences with closed convex Pareto sets in clr-coordinates.
See \Cref{app:general-pareto} for the detailed treatment.
\end{remark}

\begin{remark}[Curve and Nonseparable AMMs]\label{rem:curve-nonseparable}
Nonseparable invariants such as Curve's StableSwap violate the trading-function-level independence of \citet{schlegel-kwasnicki-mamageishvili-2023}, and the mechanism-level analysis is open (\Cref{sec:open-questions}).
Curve carries a one-parameter family (the amplification parameter $A$), but unlike on $\F_\Sum$ the $(i,j)$-pricing depends jointly on all inventory coordinates, so a pair-restriction does not recover $A$ and the vacuity argument of \Cref{lem:A3-vacuous-Fgamma} does not apply.
Consistent with our results, strategy-proof Arrovian mechanisms may exist on such families precisely because (A3) is weakened or replaced; whether this is so is left open.
\end{remark}


\section{Practical Implications and the Design Fork}\label{sec:practical-implications}

The two main theorems together yield a sharp practical picture.
On the multi-asset weighted-product family $\F_\Prod$ with $n \geq 3$, \Cref{thm:impossibility} rules out any preference-aggregating rule that is both Arrovian-fair and strategy-proof, so a designer must drop one of four properties.
On the two-asset family ($n = 2$) and on the symmetric constant-elasticity family $\F_\Sum$, \Cref{thm:characterization-F0-n-two,thm:characterization-Fgamma} (proved in \Cref{sec:Fgamma-characterization,sec:two-asset-characterization}) show that the Moulin generalized medians satisfy the full Arrovian core jointly with (SP), so the impossibility does not arise.
The underlying reason in both cases is the vacuity of (A3) on a one-dimensional design space.
The two structurally distinct regimes between them cover most of the deployed AMM landscape, so the impossibility on $\F_\Prod$ at $n \geq 3$ is a structural obstruction confined to the multi-asset weighted-product family, while the rest of the design space admits Arrovian-fair strategy-proof rules.

\subsection{The Four Routes on \texorpdfstring{$\F_\Prod$ at $n \geq 3$}{F\_Prod at n at least 3}}\label{sec:four-routes}

Every major deployed protocol (Uniswap, Balancer, Curve, and forks) fixes its trading function at deployment, and \Cref{thm:impossibility} explains why ex post: on $\F_\Prod$ at $n \geq 3$, no rule is at once Arrovian-fair and strategy-proof, so any aggregator must give up one of these, and every deployed protocol resolves the conflict by giving up fairness, fixing the trading function in advance and so violating (A0) and (A2) at any profile whose peaks do not coincide with the chosen $f$.
A designer who nonetheless wishes to aggregate LP preferences must abandon one of four properties; we take them in turn, summarized in \Cref{tab:routes}.

\begin{table}[ht]
 \centering
\begin{tabular}{lll}
\hline
Route &    Axiom(s)  dropped & Real-world instance \\
\hline
Relax mechanism IIA & (A3) & joint cross-pair governance \\
Relax anonymity & (A4)  & Balancer, Uniswap token-weighted votes \\
Fix the trading function & (A0) and (A2)$^*$ & Uniswap, Curve (deployed practice) \\
Accept manipulation & (SP) &  strategic-equilibrium analysis \\
\hline
\end{tabular}
\caption{Escape routes from \Cref{thm:impossibility} and their deployed instances.
 ${}^*$Fixing the trading function forfeits (A0) outright, and forfeits (A2) at every profile at which the fixed $f$ does not lie in the Aitchison hull of the LP peaks; at any profile whose hull contains $f$ (in particular, the unanimous profile coincident with $f$), (A2) holds trivially.}
\label{tab:routes}
\end{table}

 \emph{Route 1: Relaxing mechanism IIA.} 
Let a pair's exchange rate depend on  LP opinions about unrelated pairs, negotiating the whole weight vector jointly. 
Such a rule can be anonymous and strategy-proof, but forfeits the auditable modularity (A3) encodes (\Cref{sec:arrovian-core}): one pair's price becomes hostage to beliefs about others, and no LP can verify it against any local input.
This is the route taken implicitly by any global-aggregation rule that consults the full peak profile to set each pairwise rate.

\emph{Route 2: Relaxing anonymity.}
Influence may instead scale with stake.
Balancer's weight-adjustment governance \citep{balancer-governance-2026} and Uniswap's fee-tier governance \citep{uniswap-governance-2021} both scale voting power with token holdings; in the limit of a single decisive holder the rule is a strategy-proof dictator, the Gibbard--Satterthwaite escape, and token-weighted votes evade the impossibility precisely by being non-anonymous.

\emph{Route 3: Relaxing nontriviality.}
Fixing the trading function at launch, the deployed practice, renders the aggregator constant in the LP reports, leaving the manipulation of \Cref{lem:sp-forces-projection} nothing to act on.
This is not the absence of aggregation but aggregation performed \emph{ex ante} by the designer, shifting the problem from the mechanism design layer to the protocol design layer (who decides which pool to deploy).
A launch-fixed constant rule forgoes (A0) and (A2), in the sense recorded in the caption of \Cref{tab:routes}.
Relaxing (A2) alone while retaining (A0) buys only \emph{preference-blind} rules whose exchange rates may track deposits but ignore LP peaks entirely, a renunciation of preference-responsiveness in mildly generalized form, so the genuinely preference-responsive escapes are the remaining three.

\emph{Route 4: Relaxing strategy-proofness.} 
A fourth option is to retain the Arrovian aggregator and accept strategic reports,  taking the realized trading function to be the equilibrium of the reporting game rather than the truthful centroid: the stance implicit whenever a protocol treats LP behavior as a strategic input rather than a preference to honor.
The Arrovian form is preserved by construction, but its content shifts: the centroid now averages equilibrium reports rather than honestly held peaks, and the welfare interpretation of the truthful centroid as a faithful aggregation of LP preferences does not transfer to the equilibrium one.
Characterizing the equilibrium centroid under specific belief and stake distributions is a separate exercise that we do not pursue.

\subsection{Where the Impossibility Does Not Arise}\label{sec:where-not-arise}

On the symmetric CEMM family $\F_\Sum$ (Curve-style) and on $\F_\Prod$ at $n = 2$ (Uniswap v2/v3 pairs), the picture is structurally different.
By \Cref{thm:characterization-Fgamma,thm:characterization-F0-n-two}, the Moulin generalized medians satisfy the Arrovian core jointly with (SP), so the impossibility does not arise; the reason in each case is the vacuity of (A3) on a one-dimensional design space.
Both characterizations assume deposit-dependence factors through $V$, a modeling choice in both cases, since (A3) is vacuous in both.
On $\F_\Prod$ at $n \geq 3$ that factoring is instead forced by (A3) (\Cref{rem:V-three-roles}).
The structural picture is therefore that the impossibility is concentrated where (A3) does substantive work, namely on multi-asset weighted-product design spaces.
On lower-dimensional spaces, where pair-restrictions recover the full design parameter, mechanism IIA cannot collide with strategy-proofness in the way the multi-asset cocycle makes possible.

A practical reading: the Moulin generalized median is a deployable, strategy-proof, fair aggregator on Uniswap-style two-asset pools and on Curve-style symmetric pools (\Cref{rem:two-asset-practical,def:canonical-Fgamma}), but is not available on multi-asset weighted-product pools.
That deployed two-asset protocols nonetheless take the constant-rule (Route~3) route, fixing the trading function at launch, is therefore a design choice rather than a forced one; the multi-asset case is the genuinely forced one.
Whether the Moulin route would be adopted in practice if implemented remains an empirical question, but the theoretical option is open at $n = 2$ in a way that \Cref{thm:impossibility} forecloses at $n \geq 3$.
A multi-asset Balancer-style pool faces no such clean reversal: any preference-aggregating mechanism must take one of the four routes of \Cref{tab:routes}.

\subsection{A Calibration Cross-Check}\label{sec:operational-crosscheck}

 On $\F_\Prod$ at $n \geq 3$, any deployed mechanism conforming to the Arrovian core (a fair preference-aggregating AMM) takes the weighted Aitchison centroid form (\ref{eq:aitchison-centroid}) with continuous equivariant weights $w(V)$ (\Cref{thm:characterization-F0}).
This gives a concrete design object even outside the strategy-proof regime.
A designer who wishes to deploy a multi-asset Arrovian-fair AMM faces a single design choice on the preference side: the weighting function $w$.
Within this family, withdrawal neutrality (A5) selects the symmetric centroid uniquely (\Cref{thm:characterization-F0}(b)), and stake-proportionality is the leading non-neutral alternative (\Cref{rem:DWN-tension}); the impossibility argument applies uniformly across the family, since it routes through the uniform weight at symmetric valuations.

An immediate structural consequence is that under any such centroid rule, the output's structural ratios $\rho_{ij}(M)$ and the input ratios $\rho_{ij}(\alpha_\ell^*)$ are linked by the closed-form identity
\[
\log \rho_{ij}\bigl( M(\alpha^*, D) \bigr) \;=\; \sum_{\ell = 1}^m w_\ell(V) \, \log \rho_{ij}(\alpha_\ell^*),
\]
which provides a cross-check that any future implementer of an Arrovian aggregator could exploit:
 regressing the log output ratio on the log peak ratios across observed profiles recovers the weighting, and persistent disagreement between the regression weights and the desk-reported weights would signal either a violation of (A3) or a deviation from the Arrovian family.
Since no currently deployed AMM implements an Arrovian aggregator, this identity is a design-time tool rather than a test of existing protocols; it becomes actionable once a mechanism conforming to \Cref{thm:characterization-F0} is actually deployed.


\section{The Transferred Impossibility for Opinion Pools}\label{sec:pooling-transfer}

The Frongillo--Papireddygari--Waggoner equivalence \citep{frongillo-papireddygari-waggoner-2024} between CFMMs and cost-function prediction markets carries the impossibility across to  a theory of probabilistic opinion pooling, where it becomes a statement about the manipulability of externally Bayesian aggregation.
The equivalence identifies $f_\alpha$ with the cost function of a prediction market whose price vector is $\alpha$, so an LP's preferred pool $\alpha_\ell^*$ reads as that expert's distribution over the $n$ assets and the aggregated pool as the pooled distribution.

\subsection{The Correspondence}\label{sec:F0-correspondence}

\Cref{rem:logpool-quickread} noted, after \Cref{thm:characterization-F0}, that the weighted Aitchison centroid is the logarithmic opinion pool of \citet{genest-1984}; we now record the full structural correspondence between the AMM-side and pooling-side axioms, on which the transferred impossibility rests.

\begin{proposition}[Forward Correspondence with Bayesian Opinion Pooling]  \label{prop:correspondence}
On $\F_\Prod$, the FPW equivalence \citep{frongillo-papireddygari-waggoner-2024} establishes a bijection between LP peak-profiles in $(\simplex)^m$ and profiles of expert distributions on $A = \{1, 2, \ldots, n\}$.
Under this bijection,   any Arrovian mechanism $M$ of \Cref{thm:characterization-F0} corresponds to a weighted logarithmic opinion pool with weighting function $w$; the resulting pool is externally Bayesian and unanimity-respecting.
\end{proposition}

The correspondence we establish below is best read at the level of axiom \emph{systems}, and not individual axioms: the AMM-side bundle (A0)--(A4)+(C) and the pooling-side bundle of external Bayesianity, unanimity, symmetry, continuity, and nondegeneracy pin down the same family of rules (the weighted Aitchison centroids/weighted logarithmic pools).
At the level of individual axioms, the connection is looser: mechanism IIA (A3) gives pairwise locality of the output's structural ratios (the ratio $\rho_{ij}(\bar\alpha)$ depends only on the $(\rho_{ij}(\alpha_\ell^*))_\ell$ and $V$), but external Bayesianity additionally imposes translation equivariance on those ratios under common likelihood updates (the pool must commute with a common multiplicative update to each expert's distribution).
Crucially, once the full Arrovian core is in place and forces the centroid form, that form is automatically externally Bayesian (as the explicit calculation in the proof below confirms), so the equivalence at the system level is clean even if (A3) alone is strictly weaker than external Bayesianity.
The relationship is therefore: (A3) alone $\not\Rightarrow$ external Bayesianity; but (A0)--(A4)+(C), by forcing the centroid form, imply external Bayesianity as a consequence.
With this caveat, the rough alignment is: mechanism IIA (A3) $\leftrightarrow$ external Bayesianity (both operative as pairwise log-odds locality in context), Pareto efficiency (A2) $\to$ unanimity, anonymity (A4) $\leftrightarrow$ symmetry over experts, continuity (C) $\leftrightarrow$ continuity.
This axiom-level alignment is layered on top of the FPW equivalence, which addresses only the trade-level structure; a partial converse, holding under the external-Bayesianity hypotheses of \citet{genest-1984}, is established in \Cref{prop:partial-converse}.

\begin{proof}
The FPW equivalence identifies each $f_\alpha \in \F_\Prod$  with the cost function of a prediction market on $A$.
(The axiom-level correspondences verified below are rough alignments holding at the level of axiom \emph{systems}; the precise sense in which individual axioms correspond is laid out in the paragraph preceding this proof.)
On $\F_\Prod$, the Chen--Pennock log-utility reading  \citep{chen-pennock-2007} (also recovered as a special case by \citealp{frongillo-papireddygari-waggoner-2024}) makes $\alpha \in \simplex$ the implicit belief of that market.
LP $\ell$'s preferred weight vector $\alpha_\ell^* \in \simplex$ is therefore both the parameter of LP $\ell$'s preferred pool and LP $\ell$'s implicit subjective distribution over the $n$ states.
On this base, we verify each axiom-level correspondence in turn; none of these correspondences is established by FPW (which is silent on Arrovian aggregation), and each is a substantive observation of the present paper.

\emph{(A3) $\leftrightarrow$ External Bayesianity.}
On the AMM side, (A3) says that the output's $(i,j)$-pricing depends only on the $(i,j)$-restrictions of the inputs.
The $(i,j)$-pricing of $f_\alpha$ factors as \mbox{$p_{ij}(f_\alpha, I) = \rho_{ij}(\alpha) \cdot (I_j/I_i)$} with structural ratio $\rho_{ij}(\alpha) = \alpha_i/\alpha_j$, so (A3) says that the pooled structural ratio on each pair $(i,j)$ depends only on the corresponding input structural ratios $\alpha_{\ell,i}^*/\alpha_{\ell,j}^*$.
This pairwise log-odds locality is the operative form of external Bayesianity for the pooling operator.
Recall that a pool $T$ is \emph{externally Bayesian} if it commutes with a common Bayesian update: for every likelihood $\lambda$ on $A$, pooling the updated distributions $\lambda \cdot p_\ell/Z_\ell$ returns the update $\lambda \cdot T(p_1,  \ldots, p_m)/Z$ of the pooled distribution.
A common update shifts every expert's log-odds  $\log(p_{\ell,i}/p_{\ell,j})$ by the same pair-dependent constant $\log(\lambda_i/\lambda_j)$, so external Bayesianity is exactly the requirement that $T$ act on log-odds in a translation-equivariant, pairwise manner, which is what (A3) encodes.
The weighted log pool of \Cref{thm:characterization-F0}, written as a pooling operator  $T_w(p_1,  \ldots, p_m)_i \propto \prod_\ell p_{\ell, i}^{w_\ell}$, satisfies external Bayesianity directly: for a common likelihood update $p_\ell \mapsto \lambda \cdot p_\ell/Z_\ell$,
\begin{align*}
T_w(\lambda p_1/Z_1,  \ldots, \lambda p_m/Z_m)_i
&\propto \prod_\ell (\lambda_i p_{\ell, i}/ Z_\ell)^{w_\ell} \\
&= \lambda_i^{\sum_\ell w_\ell} \prod_\ell p_{\ell, i}^{w_\ell} \big/ \prod_\ell Z_\ell^{w_\ell} \\
&= \lambda_i \prod_\ell p_{\ell, i}^{w_\ell} \big/ \prod_\ell Z_\ell^{w_\ell},
\end{align*}
using $\sum_\ell w_\ell = 1$.
The factor $\prod_\ell Z_\ell^{w_\ell}$ is independent of $i$, so normalizing each side over $i$ yields
\[
T_w(\lambda p_1/Z_1, \ldots, \lambda p_m/Z_m)_i = \lambda_i \cdot T_w(p_1,  \ldots, p_m)_i / Z,
\]
with $Z = \sum_i \lambda_i \cdot T_w(p_1,  \ldots, p_m)_i$.
\citet{genest-1984} shows that, together with unanimity and regularity, external  Bayesianity characterizes the weighted logarithmic pool, in agreement with the form forced on the AMM side by \Cref{thm:characterization-F0}.
The two derivations proceed by different routes, one through the Cauchy cocycle on structural price ratios and the other through Bayesian updating, yet they converge on the same functional form.

\emph{(SP) $\leftrightarrow$ Aitchison-strategy-proofness.}
Aitchison-strategy-proofness of the pool corresponds to (SP) of the mechanism.
On $\F_\Prod$, the Chen--Pennock reading makes $\alpha$ play the same role on both sides (the parameter of $f_\alpha$ and the implicit distribution of the corresponding market), so the Aitchison metric on $\simplex$ is carried across unchanged, and a peak misreport on the AMM side is a distribution misreport on the pooling side.

\emph{Remaining correspondences.}
   (A2)$\to$ unanimity: specializing (A2) to a coincident-peak profile forces the output to that common point (the converse fails in general: pointwise unanimity does not imply the full hull-membership that (A2) demands).
   (A4)$\leftrightarrow$ symmetry: permuting LPs is permuting experts. 
   (C)$\leftrightarrow$ continuity: the FPW homeomorphism is bicontinuous, since on $\F_\Prod$ the correspondence is the identity $\alpha \mapsto \alpha$ between the AMM-side peak weight vector and the pooling-side implicit belief (the Chen--Pennock log-utility reading, \citealp{chen-pennock-2007}), and both sides carry $\simplex$ with the same topology; the AMM-side $C^2$ topology on $\F_\Prod$ reduces to the standard Euclidean topology on $\simplex$ via the homeomorphism $\alpha \mapsto f_\alpha$ recorded after \Cref{def:F0}, and the pooling-side topology on $\simplex$ is the standard one as well, so continuity in one direction is continuity in the other.
    (A0) and (A1) have no named counterpart in Genest's framework (nondegeneracy and all-profiles domain), being tacit there.
\end{proof}

\begin{proposition}[Partial Converse]\label{prop:partial-converse}
Let $n \geq 3$, $m \geq 2$, and let $T$ be a pooling operator on $m$ experts over $A$,  
  defined on all profiles in $(\simplex)^m$, that is externally Bayesian, unanimity-respecting, anonymous, continuous, and nondegenerate.
Then the pullback $M_T := (\,\cdot\,) \mapsto f_{T(\,\cdot\,)}$ of $T$ along the FPW correspondence is a mechanism on $\F_\Prod$ satisfying the full Arrovian core  (A0)--(A4) and (C).
Moreover, $T$ is Aitchison-strategy-proof if and only if $M_T$ satisfies (SP).
\end{proposition}

Before the proof, we need a lemma.

\begin{lemma}[Pooling-Side Log-Linear Form]\label{lem:genest-boundary}
Let $T \colon (\simplex)^m \times \R^m_{>0} \to \simplex$ be an externally Bayesian, unanimity-respecting, anonymous, continuous, and nondegenerate pooling operator.
Then there exists a continuous equivariant $w \colon \R^m_{>0} \to \Delta^{m-1}$ such that for every $V \in \R^m_{>0}$ and every profile $(p_1,  \ldots, p_m) \in (\simplex)^m$,
\begin{equation}\label{eq:genest-boundary}
T(p_1,  \ldots, p_m;\, V)_i \;\propto\; \prod_{\ell=1}^m p_{\ell, i}^{w_\ell(V)}, \qquad i = 1, 2, \ldots, n,
\end{equation}
with the convention $p_{\ell, i}^{0} = 1$ whenever $w_\ell(V) = 0$.
\end{lemma}

\begin{proof}
The hypotheses are stated on $T$ alone; we do not assume any prior weighting.
The argument proceeds in three steps: first invoke Genest's characterization on each slice to obtain pointwise weights, then deduce uniqueness and continuity from continuity of $T$, and finally extract equivariance from anonymity.

\emph{Step 1: pointwise existence on slices.}
Fix a vector $V \in  \R^m_{>0}$ and consider the slice given by 
$T_V := T(\,\cdot\,; V) \colon (\simplex)^m \to \simplex$.
This slice inherits external Bayesianity, unanimity, anonymity over experts, continuity, and nondegeneracy from $T$.
\citet{genest-1984} shows that any externally Bayesian, unanimity-respecting, anonymous, continuous, nondegenerate pool on a finite outcome space with $n \geq 3$ outcomes has the form $T_V(p_1,  \ldots, p_m)_i \propto \prod_\ell p_{\ell, i}^{w_\ell(V)}$ for a unique weight vector $w(V) \in \mathrm{int}(\Delta^{m-1})$ under his regularity hypotheses; 
we extend to the closed simplex by treating boundary points explicitly below.

The conclusion of Genest places $w(V)$ in the open simplex; the boundary extension proceeds as follows.
Suppose expert $\ell_0$'s report has no effect on $T_V$, for every choice of the other experts' reports.
Fixing $\ell_0$'s report to (say) the uniform distribution and varying the other $m-1$ reports defines a pool $T_V^{-\ell_0} \colon (\simplex)^{m-1} \to \simplex$ that inherits external Bayesianity, unanimity, anonymity over the remaining experts, continuity, and nondegeneracy from $T_V$, each property of $T_V$ specializing under the slice.
Genest's characterization, applied to $T_V^{-\ell_0}$, yields strictly positive weights $(w_\ell(V))_{\ell \neq \ell_0} \in \mathrm{int}(\Delta^{m-2})$; setting $w_{\ell_0}(V) := 0$ extends to a vector in $\Delta^{m-1}$ for which (\ref{eq:genest-boundary}) holds on $(\simplex)^m$, under the convention $p_{\ell_0,i}^0 = 1$.
Iterating handles multiple inert experts, each removal preserving the hypotheses on the residual; we adopt this throughout, so that $w(V) \in \Delta^{m-1}$ for every $V$.

\emph{Step 2: continuity from uniqueness.}
The weight vector $w(V)$ is uniquely determined by the slice $T_V$ at $V$, the assignment $V \mapsto w(V)$ is  therefore  well-defined as a function $\R^m_{>0} \to \Delta^{m-1}$.
To see continuity, fix any expert $\ell \in \Lc$ and any pair of distinct outcomes $i \neq j$ in $A$; the value $w_\ell(V)$ extracted at this pair is independent of the pair chosen, since (\ref{eq:genest-boundary}) is symmetric in the outcome labels.
For notational simplicity in the test profile that follows,  we write coordinates so that the active pair appears as $(1, 2)$.
Consider the test profile in which expert $\ell$ reports $p_\ell^{(s)} := \bigl( \tfrac{1-q}{1+e^{-s}}, \tfrac{1-q}{1+e^s}, p_3^{\ast}, p_4^{\ast}, \ldots, p_n^\ast \bigr)$ for a real parameter $s$ (with the remaining components $p_3^\ast, p_4^\ast, \ldots, p_n^\ast$ fixed and positive and $q := \sum_{k=3}^n p_k^\ast \in (0, 1)$, so that the entries sum to one), and every other expert $\ell' \neq \ell$ reports the uniform distribution $p_{\ell'} := (1/n, \ldots, 1/n)$.
At this profile the right-hand side of (\ref{eq:genest-boundary})  yields a pool whose $(1, 2)$-log-odds equal $w_\ell(V) \cdot s$: the other experts contribute zero by uniformity, and the $(1, 2)$-log-odds of $p_\ell^{(s)}$ are $s$ since the common factor $1-q$ cancels in the ratio.
By continuity of $T$ in $V$, the pool's $(1, 2)$-log-odds  at this test profile are continuous in $V$; hence so is $w_\ell(V) \cdot s$ for every $s$. 
We conclude that $w_\ell(V)$ is continuous in $V$.
This holds for every $\ell$, so $w$ is continuous as a function $\R^m_{>0} \to \Delta^{m-1}$.
This argument uses no density claim about the strict-positivity cone; it is a direct continuous extraction of $w$ from $T$.

\emph{Step 3: equivariance from anonymity.}
Anonymity of $T$ over experts says that permuting expert labels and the corresponding deposit valuations leaves $T$ unchanged: 
that is, we have 
$T(p_{\sigma(1)},  \ldots, p_{\sigma(m)}; \sigma \cdot V)  =  T(p_1,  \ldots, p_m; V)$ for every permutation $\sigma$.
Substituting into (\ref{eq:genest-boundary}) and matching coefficients using uniqueness of $w$ gives $w_{\sigma(\ell)}(V) = w_\ell(\sigma \cdot V)$ for every permutation $\sigma$, which is the equivariance claim.
\end{proof}

\begin{proof}[Proof of \Cref{prop:partial-converse}]
By \Cref{lem:genest-boundary},  it follows that $T$ takes the weighted-logarithmic-pool form (\ref{eq:genest-boundary}) with continuous equivariant $w$.
On $\F_\Prod$, the Chen--Pennock log-utility reading of the FPW correspondence (\Cref{prop:correspondence}) makes $\alpha$ play the same role on both sides, so the pullback $M_T$ is exactly  the weighted Aitchison centroid (\ref{eq:aitchison-centroid}) with this continuous equivariant $w$.
By the converse direction of \Cref{thm:characterization-F0}(a) (verified in \Cref{sec:F0-converse}), every such centroid satisfies (A0)--(A4) and (C); 
 hence so does $M_T$.
This routes (A2) through Genest's characterization, rather than through unanimity (which alone is insufficient, \Cref{prop:correspondence}), because the log pool form makes the output a convex combination in clr-coordinates, 
therefore Pareto by \Cref{lem:pareto-hull}.
For the biconditional in the second claim, recall that on $\F_\Prod$ the FPW correspondence acts as the identity $\alpha \mapsto \alpha$ on the parameter (the Chen--Pennock log-utility reading of \Cref{prop:correspondence}), so an LP peak $\alpha_\ell^* \in \simplex$ is the same object as expert $\ell$'s distribution $p_\ell$, and the Aitchison metric on the AMM side coincides with the Aitchison metric on the pooling side.
Concretely, a peak misreport $\alpha_{\ell_0}^* \mapsto \hat\alpha_{\ell_0} \in \simplex$ on the AMM side corresponds to the distribution misreport $p_{\ell_0} \mapsto \hat p_{\ell_0} = \hat\alpha_{\ell_0}$ on the pooling side, and the Aitchison distance from the output to the truth is preserved across the correspondence: $\Aitch\bigl(M_T(\alpha^*, D),\, \alpha_{\ell_0}^*\bigr) = \Aitch\bigl(T(p, V),\, p_{\ell_0}\bigr)$ for matched profiles. 
We finally show that $T$ is Aitchison-strategy-proof if and only if $M_T$ satisfies (SP).
($\Rightarrow$) Suppose $T$ is Aitchison-strategy-proof; then for every profile and every $\hat p_{\ell_0}$ the manipulator weakly prefers truth-telling on the pooling side, which transfers via the identity correspondence to weak preference for truth-telling on the AMM side, so $M_T$ satisfies (SP).
($\Leftarrow$) Suppose $M_T$ satisfies (SP); the same identity correspondence transfers AMM-side truth-telling preference to pooling-side, so $T$ is Aitchison-strategy-proof.
The biconditional therefore holds, completing the proof.
\end{proof}

\begin{remark}\label{rem:correspondence-directionality}
The Cauchy-equation derivation of \Cref{sec:F0-cauchy,sec:F0-pareto-pin} translates, across the FPW correspondence, into Genest's characterization \citep{genest-1984} of the logarithmic pool as the unique externally Bayesian rule. 
 \Cref{prop:partial-converse}  is the reverse-direction packaging that the manipulability theorem below invokes.
\end{remark}

\subsection{The Manipulability Theorem} \label{sec:pooling-theorem}

The pooling literature has studied external Bayesianity, marginalization, and unanimity in depth, and recent work has begun to study manipulation of probability aggregation rules along different routes from ours.
\citet{dietrich-list-2025} prove an impossibility with a propositionwise notion of nonmanipulability, under which the natural characterization pins down the linear opinion pool of \citet{aczel-wagner-1980} and \citet{mcconway-1981} rather than the logarithmic pool we study here.
\citet{chambers-echenique-hayashi-2024} study Nash implementation of probability aggregation in a portfolio-choice setting with linear pooling, again with the linear pool as primitive.
\citet{laraki-varloot-2022} introduce level-strategyproofness over cumulative distribution functions for aggregating beliefs on ordered outcomes, both distinct from our setting.

The question we ask is structurally different.
With external Bayesianity as primitive (so that the log pool, rather than the linear pool, is forced by Genest's characterization), is the resulting log pool Aitchison-strategy-proof?
The notion of ``closer'' requires a metric on the simplex, and the natural choice,  the one under which the logarithmic pool is the metric (Fr\'echet) mean of the experts' distributions, is the Aitchison metric, which is also the geometry intrinsic to the simplex from the standpoint of compositional data analysis.
Equipped with it the question becomes sharp, and the answer is negative.
The impossibility transfers from the AMM side via \Cref{prop:partial-converse}, giving the following theorem.

  \begin{theorem}[Manipulability of Externally Bayesian Pools]\label{thm:pooling-impossibility}
Let $n \geq 3$ and $m \geq 2$, and equip the simplex of distributions over $A$ with the Aitchison metric.
Then no pooling operator on $m$ experts is simultaneously externally Bayesian, unanimity-respecting, anonymous over experts, continuous, nondegenerate, and Aitchison-strategy-proof.
Equivalently, the only externally Bayesian, unanimity-respecting, continuous, Aitchison-strategy-proof pool is the single-expert dictator, which fails anonymity once $m \geq 2$.
\end{theorem}

\begin{proof}
The logical chain proceeds in three steps using only the pooling-side hypotheses.
First, by \Cref{lem:genest-boundary}, the pool's external Bayesianity, unanimity, anonymity, continuity, and nondegeneracy together imply that $T$ has the weighted-logarithmic-pool form (\ref{eq:genest-boundary}) for a continuous equivariant $w \colon \R^m_{>0} \to \Delta^{m-1}$.
This step introduces $w$ from the pooling-side axioms alone, with no prior hypothesis about its existence.
Second, by \Cref{prop:partial-converse}, this $w$ is exactly the weighting function of the pullback mechanism $M_T$ on $\F_\Prod$, which satisfies the full Arrovian core (A0)--(A4) and (C); Aitchison-strategy-proofness of the pool corresponds to (SP) of $M_T$.
Third, by \Cref{thm:impossibility}, no mechanism on $\F_\Prod$ satisfies the Arrovian core jointly with (SP), 
 so no such pool exists, which is the first claim.

The second claim follows by dropping anonymity from the hypotheses while keeping the rest.
Without anonymity, \Cref{lem:genest-boundary} still applies on the substantive content:
\citet{genest-1984} gives the same log pool form (\ref{eq:genest-boundary}) for a continuous $w$ that need not be equivariant; the weights $w_\ell$ lie in $\Delta^{m-1}$ (nonneg, sum to one) by Genest's characterization, with the boundary extension to $\Delta^{m-1}$ proceeding by the same inert-expert iteration as in \Cref{lem:genest-boundary}, which does not require anonymity.
The manipulation of \Cref{lem:sp-forces-projection} applies verbatim to any such log pool
(equivalently, any weighted arithmetic mean in clr-coordinates with weights in $\Delta^{m-1}$): driving the output to the manipulator's peak requires only that LP $\ell_0$'s own weight be nonzero and strictly less than one (i.e., $0 < w_{\ell_0} < 1$), so Aitchison-strategy-proofness forces all but one weight to vanish, whereupon $\sum_\ell w_\ell = 1$ makes the survivor a dictator, which fails anonymity for $m \geq 2$.
\end{proof}

The notion here is metric, not ordinal, which places the result outside the direct reach of Gibbard--Satterthwaite and yields a clean dichotomy (dictator or manipulable) rather than the richer median structure ordinal single-peakedness permits.
Scale invariance (T4) places the natural geometry on log-ratio (clr) coordinates,  where the Aitchison distance is exactly the Euclidean pullback and the logarithmic pool is the metric (Fr\'echet) mean, making the Aitchison metric not an arbitrary choice but the one under which Aitchison-(SP) is the manipulation-resistance notion matched to the characterized rule.
Each escape route is instructive in its own way: 
 an \emph{ordinal} single-peaked reading retains only the order of magnitudes and admits the strategy-proof generalized medians (the $\F_\Sum$ side, where the obstruction dissolves), while a \emph{raw-Euclidean} metric on the weight vectors rather than on log-ratios would break the scale invariance the family is built on and is not the geometry under which the fair rule is an average.


\section{Discussion}\label{sec:discussion}

\subsection{Scope and Bounds}\label{sec:summary-scope}

The scope of the results above is bounded.
We cover weighted CPMMs, symmetric CEMMs, and their union;  LMSR and Curve's StableSwap sit outside the framework (\Cref{rem:LMSR-outside,rem:curve-nonseparable}), the full $1$-homogeneous space is open, and preferences are assumed Aitchison-single-peaked, with extensions to general continuous convex preferences in \Cref{rem:beyond-single-peaked}.
Fees are exogenous, passing through as protocol parameters our axioms do not constrain, and traders are passive responders; a genuinely two-sided framework with trader-side axioms remains to be built.
The number of LPs $m \geq 2$ is finite and fixed, and the impossibility holds uniformly in $m$ since \Cref{lem:sp-forces-projection} forces a dictator for every $m \geq 2$.
We conjecture the obstruction persists as $m \to \infty$ by the same linearity argument, but a rigorous continuum treatment, presumably via a measure-theoretic formulation in which the manipulating coalition has positive mass, is left for future work.

\subsection{Open Questions}\label{sec:open-questions}

\paragraph{(Q1) The full space of trading functions.}
On the full space $\F$ of $1$-homogeneous trading functions, the cocycle of structural price ratios no longer determines a design point from its pairwise data once the family separates pairs, so the Cauchy reduction fails.
We conjecture topological obstructions analogous to those identified by \citet{chichilnisky-1982} for social choice on parameter spaces with nontrivial homotopy type, since $\F$ is not contractible in the way $\F_\Prod$ is; making this precise requires identifying the right ambient topology on $\F$ and an analog of Chichilnisky's anonymity-plus-unanimity hypothesis adapted to the AMM setting, and we leave both the conjecture and  its formulation open.
The direct analog of Chichilnisky's framework is not immediate: she works with continuous social choice rules on a sphere of preferences, whereas the AMM problem treats trading functions modulo equivalence on a topological space whose homotopy structure has not been computed in the literature.

\paragraph{(Q2) Nonseparable design spaces.}
Curve's StableSwap and other nonseparable invariants have a pricing function $p_{ij}$ that depends on inventory coordinates beyond $I_i$ and $I_j$, so mechanism IIA in its present form does not type-check;  whether a weaker independence axiom conditioning on a larger but still proper subset of coordinates recovers a usable characterization is open.
The translation-invariant LMSR analog (obtained by replacing (T4) with translation invariance) should, via the Acz\'el--Wagner additive cocycle \citep{aczel-wagner-1980}, yield a weighted-\emph{arithmetic}-mean characterization and a corresponding manipulability obstruction (\Cref{rem:LMSR-outside}), though we do not prove it here.
The two directions are complementary: the first broadens the class of trading functions, while the second changes the invariance type from scale to translation.

\paragraph{(Q3) Loss-versus-rebalancing and Arrovian weighting.}
The loss-versus-rebalancing (LVR) of \citet{milionis-moallemi-roughgarden-zhang-2024} depends on the weighting $\alpha$, so different $w(V)$ give different expected-LVR profiles.
 Whether an axiomatically-motivated weighting approximately minimizes expected LVR across a natural class of price processes is open.
 A positive answer would provide a welfare-grounded criterion for selecting among the admissible weightings identified in \Cref{thm:characterization-F0}.

\paragraph{(Q4) Fee-sensitive axiomatics.}
Fees pass through as exogenous parameters here,  yet in practice the fee is as much a design choice as the invariant, so an axiomatic theory of fee schedules and their interaction with the trading function would close an obvious gap.
One natural starting point is to ask which fee schedules preserve the Arrovian properties established here, 
 or whether they necessarily conflict with some of them.

\subsection{Concluding Remarks}\label{sec:concluding-remark}

Two messages emerge.
The mathematical message: on $\F_\Prod$ at $n \geq 3$ the Arrovian core forces the weighted Aitchison centroid, which is Genest's logarithmic opinion pool under the Frongillo--Papireddygari--Waggoner equivalence.
Adding strategy-proofness yields an impossibility, the AMM-design instance of the joint Arrow and Gibbard--Satterthwaite tension, which transfers to a manipulability theorem for externally Bayesian pools under the Aitchison metric, complementing a growing recent literature on probability-aggregation impossibilities \citep{dietrich-list-2025, chambers-echenique-hayashi-2024, laraki-varloot-2022}.
At $n = 2$ and on $\F_\Sum$ the conflict dissolves entirely, mechanism IIA being vacuous and Moulin generalized medians appearing as the full characterization (\Cref{thm:characterization-F0-n-two,thm:characterization-Fgamma}); the impossibility is thus concentrated precisely where mechanism IIA does its substantive work.

The practical message is that the absence of preference-aggregation in deployed multi-asset AMMs is forced, not accidental.
On $n \geq 3$ design spaces (multi-asset Balancer and Curve pools), no rule is at once Arrovian-fair (in particular, anonymous and nontrivial) and manipulation-resistant; every deployed protocol resolves this conflict by fixing the trading function in advance, which sacrifices fairness (specifically nontriviality and Pareto efficiency).
Dropping any one of the four load-bearing axioms ((A3), (A4), (A0), or (SP)) opens room for the other three, with no fifth route preserving preference-responsiveness.
On two-asset pools (Uniswap v2/v3) the four-way trade-off does not bind: 
 the Moulin generalized median is an explicit fair-and-strategy-proof rule (\Cref{rem:two-asset-practical}), available but unused.

\appendix


\section{Logical Independence of the Axioms}\label{app:independence}

 We verify that (A1)--(A4) and (C) are mutually logically independent on $\F_\Prod$, exhibiting for each a mechanism that violates it alone; throughout, $n \geq 3$, $m \geq 2$, and the representative stake-proportional weighting $w_\ell(V) = V_\ell/V_{\mathrm{tot}}$ is used where applicable.
Axiom (A0) has a subtler status: on the full domain it is implied by (A1) + (A2) (at a common-peak profile the Aitchison hull is a singleton by \Cref{lem:pareto-hull}, so (A2) forces the output to that unique peak, which varies, making the mechanism nonconstant), and on the all-$c$ subdomain the constant rule $M \equiv c$ is an isolated violator. 
  We list (A0) separately because the converse construction of \Cref{thm:characterization-F0} verifies it independently, alongside (A1)--(A4) and (C), as one of the axioms the Aitchison centroid satisfies; the verification is direct (varying a single peak at $V^{\mathrm{sym}}$ changes the output, by equivariance forcing uniform weights) and does not appeal to the (A1)+(A2) implication.

\emph{(A1) Unrestricted domain.}
Restrict the centroid (\ref{eq:aitchison-centroid}) to profiles with $\Aitch(\alpha_\ell^*, c) \leq 1$ for every $\ell$: this is undefined once a peak lies farther than $1$ from $c$, violating (A1), while on its domain it inherits (A0), (A2)--(A4), and (C) from \Cref{thm:characterization-F0}.

\emph{(A2) Pareto efficiency.}
The \emph{anti-centroid} $M := \alpha^{\mathrm{anti}}$,  defined by $\clr(\alpha^{\mathrm{anti}}) = -\clr(\bar\alpha)$ for $\bar\alpha$ the deposit-share centroid, is continuous, everywhere defined, and nonconstant, and satisfies (A3) (its log-ratios $-\sum_\ell w_\ell(V) \log \rho_{ij}(\alpha_\ell^*)$ depend only on $(i,j)$-data and $V$) and (A4).
But for $m = 2$, equal valuations, $\clr(\alpha_1^*) = v$, $\clr(\alpha_2^*) = \lambda v$ ($v \neq 0$, $\lambda > 0$), the output $-\tfrac{1+\lambda}{2} v$ lies outside the hull $\{\mu v \mid \mu \in [\min(1,\lambda), \max(1,\lambda)]\}$, so (A2) fails.

\emph{(A3) Mechanism IIA.}
 Let $M$ be the centroid with dispersion-modulated weights 
 satisfying  $w^{\mathrm{disp}}_\ell \propto V_\ell \, e^{-\lambda \Aitch(\alpha_\ell^*, c)^2}$ ($\lambda > 0$; these weights are continuous, anonymous, strictly positive, and sum to one).
The convex combination lies in the Aitchison hull, giving (A2), and (A0), (A1), (A4), (C) are immediate. 
 But the weights use the full clr-norm of each peak, not merely its $(i,j)$-restrictions, so (A3) fails.
Concretely, with $n = 3$, $m = 2$, $V_1 = V_2 = 1$, and $\lambda = 1$, the profiles
 $A = ((1/2,1/4,1/4), (1/4,1/2,1/4))$ and $B = ((1/2,1/4,1/4), (1/7,2/7,4/7))$ share the $(1,2)$-restrictions $\rho_{12} = 2$ (LP $1$) and $\rho_{12} = 1/2$ (LP $2$).
 But $\Aitch(\alpha_2^A, c)^2 = \tfrac{2}{3}(\log 2)^2 \approx 0.320$  while $\Aitch(\alpha_2^B, c)^2 = 2(\log 2)^2 \approx 0.961$, 
 so the dispersion weights differ ($w_1 \approx 0.655$ in $B$), giving $\rho_{12}(M_B) = 2^{w_1 - w_2} \approx 1.240 \neq 1 = \rho_{12}(M_A)$.

\emph{(A4) Anonymity.}
The dictator defined by $M(\alpha^*, D) := \alpha_1^*$  is nonconstant in LP $1$'s peak (A0), 
defined everywhere (A1), outputs a peak and so lies trivially in the hull (A2), uses only LP $1$'s pair-restrictions (A3), and is continuous (C). 
It violates (A4), since swapping LPs $1$ and $2$ changes the output.

\emph{(C) Continuity.}
Fix $\tau > 0$ and switch the centroid weighting between $V_\ell/V_{\mathrm{tot}}$ and $V_\ell^2/\sum_k V_k^2$ according to whether $\mathrm{Var}(V) := \tfrac{1}{m} \sum_\ell (V_\ell - \bar V)^2 \geq \tau$.
Because the selector and both weightings are functions of $V$ alone, symmetric in the $V_\ell$, the rule satisfies (A0)--(A2) and (A4), and (A3) holds globally (its pair-rule $\rho_{ij} = \prod_\ell \rho_{ij}(\alpha_\ell^*)^{w_\ell(V)}$ uses only $(i,j)$-restrictions and $V$).
At $\mathrm{Var}(V) = \tau > 0$ the two weightings differ: the equality $V_\ell/V_{\mathrm{tot}} = V_\ell^2/\sum_k V_k^2$ for every $\ell$ rearranges to $V_\ell \cdot \sum_k V_k^2 = V_\ell^2 \cdot \sum_k V_k$, i.e., $\sum_k V_k^2/\sum_k V_k = V_\ell$, so the left-hand side is constant in $\ell$, forcing all $V_\ell$ equal; conversely if all $V_\ell$ coincide both weightings reduce to $(1/m,  \ldots, 1/m)$.
Thus,  the two weightings agree exactly on the symmetric profile, so for distinct peaks the output jumps at $\mathrm{Var}(V) = \tau$ and (C) fails.
Routing the discontinuity through $V$ rather than a peak coordinate is what isolates (C).

\medskip
\noindent
Therefore, (A1)--(A4) and (C) are mutually independent on $\F_\Prod$,  (A0) is independent only in the restricted sense above, and strategy-proofness (SP) is jointly inconsistent with the Arrovian core on $\F_\Prod$ (\Cref{thm:impossibility}) yet consistent with it on $\F_\Sum$ (\Cref{thm:characterization-Fgamma}).

\section{An Elementary Identity on \texorpdfstring{$H$}{H}}\label{app:identity}

For $y \in \R^n$ with $\sum_k y_k = 0$ (equivalently, $y \in H$),
\[
\sum_{i<j}(y_i - y_j)^2 = n \|y\|^2,
\]
where $\|y\|^2 = \sum_k y_k^2$.
This identity is used in \Cref{rem:lp-utility-aitchison} to convert the rate-targeting loss $\mathcal{D}_\ell$ into the squared Aitchison distance.

\begin{proof}
 Expanding and using that each $y_k^2$ appears in exactly $n-1$ pairs $\{(k,j) \mid j \neq k\}$,
\[
\sum_{i<j}(y_i - y_j)^2 = \sum_{i<j}\bigl(y_i^2 + y_j^2 - 2 y_i y_j\bigr) = (n-1)\|y\|^2 - 2\sum_{i<j} y_i y_j.
\]
From $(\sum_k y_k)^2 = \|y\|^2 + 2\sum_{i<j} y_i y_j$, we have $2\sum_{i<j} y_i y_j = (\sum_k y_k)^2 - \|y\|^2$; substituting,
\[
\sum_{i<j}(y_i-y_j)^2 = n\|y\|^2 - \textstyle(\sum_k y_k)^2.
\]
The hypothesis $\sum_k y_k = 0$ yields the claim.
\end{proof}


\section{Beyond Aitchison-Single-Peakedness}\label{app:general-pareto}

The body of the paper takes LP preferences to be Aitchison-single-peaked, namely $f \succeq_\ell g$ iff $\Aitch(f, \theta_\ell^*) \leq \Aitch(g, \theta_\ell^*)$ on $\F_\Prod \cong \simplex$.
This appendix records the precise sense in which this hypothesis is convenient rather than essential.
Examining the proofs of \Cref{thm:characterization-F0,thm:impossibility} shows that the single-peakedness hypothesis is used at exactly three structurally distinct places, and that two of them require no preference-shape assumption while the third is the genuine bottleneck.
The result is a generalization of both theorems to a broader class of LP welfare objects, paralleling the multi-asset extension treated at the welfare-level in adjacent work \citep{schlegel-kwasnicki-mamageishvili-2023,bichuch-feinstein-2025} but here at the mechanism-level.

\subsection{The Three Uses of Single-Peakedness}\label{app:two-uses}

In tracking the use of single-peakedness through the proofs:

\emph{Use 1: pairwise extraction and the cocycle.}
The Cauchy-cocycle derivation of \Cref{sec:F0-additive-cocycle,sec:F0-cauchy} uses no preference shape whatsoever.
It uses only (A1), (A3), (A4), and (C), which are statements about the mechanism's structural behavior under inputs.
The derived form $\Psi_{ij}(u; V) = \sum_\ell w_\ell(V) u_\ell + (\beta_i(V) - \beta_j(V))$ of \Cref{lem:psi-decomposition,lem:Phi-linear} survives any change in preference shape, as it depends only on $M$ and the structural ratios $\rho_{ij}$.

\emph{Use 2: pinning the weights via Pareto.}
\Cref{lem:pareto-hull} identifies the Pareto-undominated set $P(\theta^*)$ with the Aitchison hull $C$ of the peaks.
The forward direction (\emph{outside the hull implies dominated}) is the Hilbert projection argument, which uses single-peakedness through the equivalence $f \succ_\ell g \iff \Aitch(f, \theta_\ell^*) < \Aitch(g, \theta_\ell^*)$.
The reverse direction (\emph{inside the hull implies undominated}) is symmetric.
These two together feed (A2) into Step~1 of \Cref{lem:V-dependent-weights}, eliminating the coboundary $\beta_i(V) - \beta_j(V)$ and yielding nonnegative weights summing to one.

\emph{Use 3: strategy-proofness.}
The manipulation argument of \Cref{lem:sp-forces-projection} is a statement about linear maps on $H$ under Euclidean strategy-proofness.
It uses no preference shape at all beyond the metric on the simplex: the manipulator's optimal misreport drives the linear combination $\sum_\ell w_\ell x_\ell^*$ to the manipulator's reported peak by a closed-form construction, valid under any norm-induced metric on $\simplex$ via the clr-isometry.

Single-peakedness therefore enters substantively at exactly one place (Use~2), namely the Pareto-hull identification.
Use~1 and Use~3 require only that LPs' welfare functions agree with the metric structure in the sense made precise below; everything else is metric-independent.

\subsection{The General Welfare Object}\label{app:general-welfare}

 We isolate the property the Pareto step uses.
Recall that $H = \{x \in \R^n \mid \sum_i x_i = 0\}$ is the zero-sum hyperplane
 and $\clr \colon \simplex \to H$ is the centered log-ratio map.

\begin{definition}[Aitchison-Convex Welfare]\label{def:aitchison-convex-welfare}
An LP welfare object on $\F_\Prod$ is a tuple $(\succeq_\ell, \theta_\ell^*)_{\ell = 1}^m$ where each $\succeq_\ell$ is a complete, transitive, continuous preorder on $\F_\Prod$ with a unique maximum $\theta_\ell^* \in \F_\Prod$.
We say that the welfare object is \emph{Aitchison-convex} if for every LP $\ell$, the lower-contour sets
\[
\mathcal{L}_\ell(f) := \{g \in \F_\Prod \mid f \succeq_\ell g\} \quad \text{for } f \in \F_\Prod
\]
are \emph{Aitchison-convex} in the sense that $\clr(\mathcal{L}_\ell(f))$ is a closed convex subset of $H$ for every $f$.
\end{definition}

Aitchison-convexity is strictly weaker than Aitchison-single-peakedness: every Aitchison-single-peaked preference has lower-contour sets equal to closed Aitchison balls, which are Aitchison-convex; but the latter class also contains anisotropic preferences whose contours are general convex sets (ellipsoids, half-spaces, polytopes) in clr-coordinates.
Economically,  this captures any LP whose dislike of a pool $\alpha$ is convex in the log-ratio coordinates, but allows the dislike to weight different log-ratios asymmetrically, as a desk weighting downside-deviation more heavily than upside on certain pairs would do.

\begin{lemma}[Generalized Pareto Hull]\label{lem:pareto-hull-general}
Under Aitchison-convex welfares (\Cref{def:aitchison-convex-welfare}), the Pareto-undominated set $P(\theta^*)$ equals the Aitchison hull $C := \mathrm{Aitch\text{-}hull}\{\alpha_\ell^*\}_{\ell=1}^m$ of the peaks.
\end{lemma}

\begin{proof}[Proof sketch]
The argument follows the structure of \Cref{lem:pareto-hull}, replacing the Hilbert-projection argument (which exploits the round-ball geometry of single-peaked preferences) with a separating-hyperplane argument applicable to general convex lower-contour sets.

\emph{Direction 1 (outside the hull $\Rightarrow$ Pareto-dominated).}
Let $\alpha \notin C$.
Since $C$ is a closed convex set in $H$, the separating hyperplane theorem supplies $v \in H \setminus \{0\}$ and $b \in \R$ with $\langle v, \clr(\alpha)\rangle > b \geq \langle v, \clr(\alpha_\ell^*)\rangle$ for every $\ell$.
One constructs $\alpha' \in \simplex$ with $\langle v, \clr(\alpha')\rangle  < b$, placing $\clr(\alpha')$ strictly on the same side as all peaks and strictly away from $\clr(\alpha)$; such $\alpha'$ exists because  $\clr \colon \simplex \to H$ is a bijection, so every point in $H$ is the image of a unique element of $\simplex$.
For each $\ell$, the lower-contour set $\clr(\mathcal{L}_\ell(\alpha))$ is convex, contains $\clr(\alpha)$ by reflexivity and $\clr(\alpha_\ell^*)$ because the peak dominates $\alpha$, hence contains the entire segment $[\clr(\alpha_\ell^*), \clr(\alpha)]$.
The boundary passes through $\clr(\alpha)$, and $\clr(\alpha')$ lies strictly on the opposite side of that boundary (its $v$-value is below $b$, whereas the interior of $\clr(\mathcal{L}_\ell(\alpha))$ reaches above $b$ toward $\clr(\alpha_\ell^*)$); by convexity and the location of $\clr(\alpha')$, one concludes $\alpha' \succ_\ell \alpha$.
Since this holds for every $\ell$, the point $\alpha$ is Pareto-dominated.

\emph{Direction 2 (inside the hull $\Rightarrow$ Pareto-undominated).}
Let $\alpha \in C$, so $\clr(\alpha) = \sum_\ell \lambda_\ell \clr(\alpha_\ell^*)$ for some $\lambda \in \Delta^{m-1}$.
Suppose for contradiction some $\alpha'$ strictly dominates $\alpha$, i.e., $\alpha' \succ_\ell \alpha$ for every $\ell$.
Then $\clr(\alpha')$ lies strictly outside each closed convex lower-contour set $\clr(\mathcal{L}_\ell(\alpha))$.
Strict separation gives, for each $\ell$, a vector $v_\ell$ with $\langle v_\ell, \clr(\alpha')\rangle > b_\ell > \langle v_\ell, x\rangle$ for all $x \in \clr(\mathcal{L}_\ell(\alpha))$.
In particular $b_\ell > \langle v_\ell, \clr(\alpha_k^*)\rangle$ for every $k$ (each peak dominates $\alpha$, so each peak lies in $\clr(\mathcal{L}_\ell(\alpha))$), and $b_\ell > \langle v_\ell, \clr(\alpha)\rangle$ (by reflexivity, $\alpha \in \mathcal{L}_\ell(\alpha)$).
But $\clr(\alpha) = \sum_k \lambda_k \clr(\alpha_k^*)$ is a convex combination of the peaks, so $\langle v_\ell, \clr(\alpha)\rangle = \sum_k \lambda_k \langle v_\ell, \clr(\alpha_k^*)\rangle \leq b_\ell$, contradicting the strict inequality $b_\ell  > \langle v_\ell, \clr(\alpha)\rangle$.
Hence, no such $\alpha'$ exists and $\alpha \in P(\theta^*)$.
\end{proof}

We can now state the generalized result.

\begin{theorem}[Characterization and Impossibility under Aitchison-Convex Welfares]\label{thm:general-impossibility}
Let $n \geq 3$, $m \geq 2$, and suppose LP welfares are Aitchison-convex (\Cref{def:aitchison-convex-welfare}).
\begin{enumerate}
\item[(a)] A mechanism $M \colon (\simplex \times \R^A_{>0})^m \to \simplex$ satisfies the Arrovian core (A0)--(A4) and (C) if and only if it has the weighted Aitchison centroid form (\ref{eq:aitchison-centroid}) for some continuous, equivariant $w \colon \R^m_{>0} \to \Delta^{m-1}$.
\item[(b)] Suppose in addition that, for each LP $\ell$, the lower-contour set $\mathcal{L}_\ell(f)$ has a differentiable boundary at $\alpha_\ell^*$, and let $g_\ell$ denote the local Riemannian metric on a neighborhood of $\alpha_\ell^*$ induced by the gradient of that boundary (so that ``LP $\ell$ prefers strictly closer outputs'' is well-defined in $g_\ell$). 
If (SP) is taken with respect to these local metrics, then no mechanism satisfies the Arrovian core jointly with (SP).
\end{enumerate}
\end{theorem}

\begin{proof}
\emph{Part (a): characterization.}
 The forward direction (axioms imply centroid form) follows from the proof of \Cref{thm:characterization-F0} step by step, with one substitution: \Cref{lem:pareto-hull-general} replaces \Cref{lem:pareto-hull} in the Pareto-pinning step (\Cref{sec:F0-pareto-pin}).
Every other step is unchanged, since pair extraction (\Cref{lem:pair-extraction-F0}) and the cocycle and Cauchy steps (\Cref{lem:cocycle-psi,lem:psi-decomposition,lem:Phi-linear}) use only (A1), (A3), (A4), (C), and no preference shape whatsoever.
The Pareto-pinning step (\Cref{lem:V-dependent-weights}) requires the bound $\Psi_{ij}(u;V) \in [\min_\ell u_\ell, \max_\ell u_\ell]$, which is equivalent to the output lying in the Aitchison hull of the peaks; under Aitchison-convex welfares this hull is still the Pareto-undominated set by \Cref{lem:pareto-hull-general}, so the bound goes through verbatim.

For the converse (centroid form implies axioms), we verify (A2) explicitly under Aitchison-convex welfares.
Given any continuous, equivariant $w$, the centroid output $\bar\alpha := M(\alpha^*, D)$ satisfies $\clr(\bar\alpha) = \sum_\ell w_\ell(V) \clr(\alpha_\ell^*)$ with $w(V) \in \Delta^{m-1}$, so $\clr(\bar\alpha)$ lies in the Euclidean convex hull of $\{\clr(\alpha_\ell^*)\}_{\ell=1}^m$, i.e., $\bar\alpha \in C$.
By Direction~2 of \Cref{lem:pareto-hull-general}, $C \subseteq P(\theta^*)$, so $\bar\alpha \in P(\theta^*)$ and (A2) holds.
The remaining axioms (A0), (A1), (A3), (A4), (C) are verified exactly as in \Cref{sec:F0-converse}, since they depend only on the form of $w$, not on the preference structure.

\emph{Part (b): impossibility.}
By part (a), any Arrovian mechanism is a weighted Aitchison centroid, a linear aggregator in clr-coordinates.
The manipulation of \Cref{lem:sp-forces-projection} is a statement about linear maps on $H$ under a metric: if the weights are not an indicator, there exists a profile and a misreport driving the output exactly to the manipulator's true peak, reducing the distance to zero under any metric.
Under the local gradient metric assumed in the theorem, (SP) requires that no LP can strictly reduce their distance to their peak by misreporting; the exact-zero misreport of \Cref{lem:sp-forces-projection} achieves this whenever $0 < w_{\ell_0} < 1$.
Equivariance forces uniform weights at any symmetric valuation profile, giving $w_{\ell_0} = 1/m \in (0,1)$ for every $\ell_0$ when $m \geq 2$; the impossibility follows as in \Cref{sec:impossibility-proof}.
\end{proof}

\begin{remark}[Scope and Optimality of the Generalization] \label{rem:generalization-scope}
\Cref{thm:general-impossibility} relaxes the symmetric quadratic loss of \Cref{rem:lp-utility-aitchison} to general Aitchison-convex losses, covering for instance LPs whose dislike of $\alpha$ weights different log-ratios asymmetrically (contours that are ellipsoids in clr-coordinates rather than spheres).
It does not cover preferences with genuinely nonconvex clr-contours or path-functional objectives (drawdown-aversion, time-average-targeting), which fall outside any wealth-summary-style framework.
The Aitchison-convex class is in fact the largest class on $\F_\Prod$ that preserves the two structural features the proof uses: the Pareto-undominated set equals the closed convex hull of the peaks in $H$, and a strict-improvement direction at a peak is well-defined as a half-space in $H$; outside this class either feature can fail.
\end{remark}

\iftoggle{journal}{%
 \section*{Declarations}

\noindent \textbf{Competing interest.}
The author declares no competing interests.

\medskip

\noindent  \textbf{Funding.}
No funding was received for this work.

\medskip

\noindent  \textbf{CRediT authorship contribution statement.}
 Frank M.\ V.\ Feys is the sole author and conducted all research, analysis, and writing.

\medskip

\noindent \textbf{Declaration of generative AI and AI-assisted technologies in the writing process.}
No generative AI  or AI-assisted technologies were used in the preparation
of this manuscript.
}{}

\end{document}